\documentclass[aos,preprint]{imsart}

\RequirePackage{amsthm,amsmath,amsfonts,amssymb}
\RequirePackage[numbers]{natbib}
\RequirePackage[colorlinks,citecolor=blue,urlcolor=blue]{hyperref}
\RequirePackage{graphicx}
\usepackage{algorithm2e}
\usepackage{bm}

\startlocaldefs
\newtheorem{thm}{Theorem}[section]
\newtheorem{lemma}{Lemma}

\newtheorem{assump}{Assumption}
\newtheorem*{assumption2prime}{Assumption $2'$}

\def\indpit{I(|\hat{\gamma}_j^*|>\lambda \hat\sigma_{Xj}^*)}

\def\var{\mbox{Var}}

\def\rmivw{\rm IVW}
\def\rmdivw{\rm dIVW}

\def\mu{\delta}
\def\rho{v}
\def\cov{\mbox{Cov}}
\def\Y{\bm{Y}}
\def\Z{\bm{Z}}
\def\X{\bm{X}}

\theoremstyle{remark}
\endlocaldefs

\begin{document}

\begin{frontmatter}
\title{Debiased Inverse-Variance Weighted Estimator in Two-Sample Summary-Data Mendelian Randomization}
\runtitle{Debiased IVW Estimator}

\begin{aug}
%%%%%%%%%%%%%%%%%%%%%%%%%%%%%%%%%%%%%%%%%%%%%%
%%Only one address is permitted per author. %%
%%Only division, organization and e-mail is %%
%%included in the address.                  %%
%%Additional information can be included in %%
%%the Acknowledgments section if necessary. %%
%%%%%%%%%%%%%%%%%%%%%%%%%%%%%%%%%%%%%%%%%%%%%%
\author[A]{\fnms{Ting} \snm{Ye} \ead[label=e1]{tingye@wharton.upenn.edu}},
\author[B,C]{\fnms{Jun} \snm{Shao}\ead[label=e2]{shao@stat.wisc.edu}}
\and
\author[B]{\fnms{Hyunseung} \snm{Kang}
	\ead[label=e3]{hyunseung@stat.wisc.edu}}

%\author[A]{\fnms{First} \snm{Author}\ead[label=e1]{first@somewhere.com}},
%\author[B]{\fnms{Second} \snm{Author}\ead[label=e2,mark]{second@somewhere.com}}
%\and
%\author[B]{\fnms{Third} \snm{Author}\ead[label=e3,mark]{third@somewhere.com}}
%%%%%%%%%%%%%%%%%%%%%%%%%%%%%%%%%%%%%%%%%%%%%%
%% Addresses                                %%
%%%%%%%%%%%%%%%%%%%%%%%%%%%%%%%%%%%%%%%%%%%%%%
\address[A]{Department of Statistics,
	University of Pennsylvania
	%Philadelphia, PA 19104
\printead{e1}}

\address[B]{Department of Statistics,
	University of Wisconsin-Madison
\printead{e2,e3}}

\address[C]{School of Statistics,
	 East China Normal University}

\end{aug}

\begin{abstract}
Mendelian randomization (MR) has become a popular approach to study the effect of a modifiable exposure on an outcome by using genetic variants as instrumental variables. A challenge in MR is that each genetic variant explains a relatively small proportion of variance in the exposure and there are many such variants, a setting known as many weak instruments. To this end, we provide a  theoretical characterization of the statistical properties of two popular estimators in MR, the inverse-variance weighted (IVW) estimator and the  IVW estimator with screened instruments using an independent selection dataset, under many weak instruments. We then propose a debiased IVW estimator, a simple modification of the IVW estimator, that is robust to many weak instruments and doesn't require screening. Additionally, we present two instrument selection methods to improve the efficiency of the new estimator when a selection dataset is available. An extension of the debiased IVW estimator to handle balanced horizontal pleiotropy is also discussed. We conclude by demonstrating our results in simulated and real datasets.
\end{abstract}



\end{frontmatter}
%%%%%%%%%%%%%%%%%%%%%%%%%%%%%%%%%%%%%%%%%%%%%%
%% Please use \tableofcontents for articles %%
%% with 50 pages and more                   %%
%%%%%%%%%%%%%%%%%%%%%%%%%%%%%%%%%%%%%%%%%%%%%%
%\tableofcontents

\section{Introduction} 
\label{sec: intro}
\subsection{Motivation: Many Weak Instruments in MR}
Instrumental variable (IV) is a well-known method to estimate the effect of a treatment, policy, or an exposure on an outcome in observational studies with unmeasured confounding \cite{Baiocchi:2014aa, Hernan-Robins, Angrist2001}. Mendelian randomization (MR), a type of IV method, 
utilizes genetic variants as instruments to study the  effect of a modifiable exposure or potential risk factor on an outcome in the presence of unmeasured confounding
\citep{Davey-Smith:2003aa, Smith:2004aa,  Lawlor:2008aa, Kamstrup:2009aa, Burgess:2015ab, Corbin:2016aa, Zheng:2017aa, Pingault:2018aa}. A distinct feature of MR is that there can be a large number of genetic variants, specifically single nucleotide polymorphisms (SNPs) from pre-existing large genome-wide association studies (GWASs), and   
many or possibly all SNPs  are weak IVs; the setting is also referred to as many weak instruments in econometrics \citep{Chao:2005aa}. In particular, these genetic instruments/SNPs can be weak for the following three reasons. First, many SNPs may have zero/null effects on the exposure.  Second, when SNPs are {\em common genetic variants}, i.e.,
their minor allele frequencies (MAF) are greater than $0.05$ \citep{Gibson:2012aa, Cirulli:2010aa}, they may have small effects on the exposure.
Third, when SNPs are {\em rare variants}, i.e., their MAF are less than $0.05$, 
they may have small or modest effects on the exposure, but their genetic variances are  small so that their total contribution to the variation of the exposure is small.

In this article, we focus on a popular setup in MR known as two-sample summary-data MR, 
where two sets of summary statistics are obtained from two GWASs \citep{Pierce:2013aa}. 
%Due to privacy, a typical MR contains two sets of summary statistics. 
The first set  from one GWAS consists of $\hat{\gamma}_j$, the estimated marginal association between the $j$th SNP and the exposure, and its standard error (SE) $\hat\sigma_{Xj}$, $j=1,...,p$. 
The second set   from another GWAS consists of $\hat{\Gamma}_j$, the estimated marginal association between the $j$th SNP and the outcome, and its SE $\hat\sigma_{Yj}$, 
$j=1,...,p$. 
In MR,  the main parameter of interest is the exposure effect  on the outcome, denoted as $\beta_0$, which can be estimated by $\hat{\beta}_j  = \hat \Gamma_j / \hat \gamma_j$ for each $j$. However,  $\hat{\beta}_j $ may be seriously biased and unstable
when SNP $j$ is weak because $\hat \gamma_j$ is close to zero  \citep{Sawa:1969aa}. This leads to several modern MR methods that aggregate many possibly unstable estimators $\hat{\beta}_j$s using a meta-analysis strategy \cite{Burgess:2013aa, Bowden:2015aa, Bowden:2016aa, Hartwig:2017aa}. The most popular among them is the inverse-variance weighted (IVW) estimator considered in \citep{Burgess:2013aa}, 
\begin{equation}
	\hat{\beta}_{{\rm IVW}} = 
	\frac{\sum_{j=1}^p \hat w_j \hat{\beta}_j}{\sum_{j=1}^p \hat w_j},  \qquad 
	\hat\beta_j = \frac{\hat{\Gamma}_j}{\hat\gamma_j}, \qquad \hat w_j = \frac{\hat\gamma_j^2}{\hat\sigma_{Yj}^2}. \label{eq: ivw} 
\end{equation}
A variant of the typical IVW estimator (\ref{eq: ivw}) is to only include SNPs that pass the genome-wide significance threshold in  a third independent GWAS, known as the selection dataset, inside the IVW estimator (\ref{eq: ivw}); we call this the IVW estimator with screening and we formally define it in equation (\ref{eq: llivw}). Despite their popularity and widespread usage, very little is known about the theoretical properties of the IVW estimator, with or without screening. Specifically, in common MR setups where there are many weak IVs, it's unknown whether the IVW estimator $\hat{\beta}_{ {\rm IVW}}$ in (\ref{eq: ivw}) or the screening counterpart are consistent or asymptotically normal.
%(i.e., p-value associated with $\hat{\gamma}_j$ is less than $5\times 10^{-8}$  \citep{Zheng:2017aa}) , which leads to the following IVW estimator with pre-screening,
%\begin{equation}\label{eq: livw} 
%\hat{\beta}_{\lambda, \rmivw}= \frac{\sum_{j \in {\cal S}_\lambda } w_j \hat{\beta}_j}{\sum_{j\in {\cal S}_\lambda }  w_j}, 
%\qquad \mbox{${\cal S}_\lambda = \{ k: $  SNP $k$ is significant$ \}$.}
%\end{equation}
%More details  are given in Section 3.
% While the use of genome-wide significance threshold is well-justified %for common generic variants
%in controlling false discovery rates from testing millions of SNPs \citep{Consortium:2005aa, Peer:2008aa},

%Indeed, in typical MR settings, we show in Section \ref{} that this threshold can be overly conservative in the sense that it may eliminate the vast majority of the relevant SNPs, resulting in loss of power to detect an exposure effect.

\subsection{Prior Work and Our Contributions}
Prior work on weak IVs in MR is vast, but mostly limited to numerical studies \cite{Burgess:2011ab, Burgess:2011aa, Burgess:2012aa, Pierce:2013aa, Burgess:2013aa}.  In econometrics, the issue of weak IVs has been studied, but the results are limited to one-sample individual-data settings; see \citet{Stock:2002aa} and \citet{NBERt0313} for surveys. Recent papers by \citet{zhao2018statistical, Zhao:2019aa} and \citet{Bowden:2019aa} proposed new two-sample summary-data MR estimators that are robust to many weak IVs. Also,  \citet{wang2019weakinstrument} proposed new tests for two-sample summary-data MR when the number of instruments is fixed, but the instruments are arbitrarily weak. To the best of our knowledge, however, 
no work has addressed the theoretical properties of the IVW estimator in (\ref{eq: ivw}) or the  IVW estimator with screening in (\ref{eq: llivw}),  arguably the most popular estimators in MR, under a typical MR setting with many weak IVs.

Our overarching goal is to characterize the properties of the IVW estimators and to propose some improvements over them. The main contributions can be divided into four parts. 
\begin{enumerate}
	\item We provide an asymptotic phase transition analysis of the IVW estimator 
	(\ref{eq: ivw})  in terms of IVs' average strength defined in (\ref{kappa}). We conduct a similar exercise for the  IVW estimator  (\ref{eq: llivw}) with screening. 
	\item We propose a simple way to improve the IVW estimator   (\ref{eq: ivw}) under many weak IVs, which we call the debiased IVW (dIVW) estimator. It is explicitly formulated as the IVW estimator multiplied by a bias correction factor; see equation \eqref{eq: dIVW}.  Unlike the IVW estimator, the dIVW estimator is robust to many weak IVs. In fact, even without screening for strong IVs, the dIVW estimator is generally consistent and asymptotically normal. As such, the dIVW estimator does not need a third independent GWAS to select  instruments to mitigate the ``winner's curse'' bias \cite{Zhao:2019aa, Burgess:2013aa}.  Finally, our dIVW estimators stand in contrast to recent optimization-based estimators (e.g., \citep{zhao2018statistical,Zhao:2019aa}) that are robust to many weak IVs, but are arguably more complex than the dIVW estimators. As an example, the optimization-based estimator in \cite{Zhao:2019aa} may not have  unique estimates in every data generating scenario.
	\item
	We depart from past theoretical studies in MR by considering the case where $\hat\sigma_{Xj}^2$ and $\hat\sigma_{Yj}^2$ are estimates, not respectively equal to $\sigma_{Xj}^2$ and $\sigma_{Yj}^2$, the true but unknown variances of $\hat\gamma_j$ and $\hat\Gamma_j$. We assess the impact of using estimated variances in the properties of the dIVW and the IVW estimators.
	%Third, we study whether the p-value cut off of $5\times 10^{-8}$ typically used in MR to select strong instruments is theoretically justified for reducing weak instrument bias. We do this by recasting the IVW estimator as a hard-thresholded IVW estimator and studying whether the static hard threshold of $5\times 10^{-8}$ is optimal for a wide variety of data generating processes. We show that for the IVW estimator, $5\times 10^{-8}$ is reasonable {\color{red}only if the average IV strength among the selected is large}. %But, if the pool is much smaller, the threshold is too conservative in the sense that it selects fewer SNPs than necessary. 
	\item To improve the efficiency of the dIVW estimator, we propose two  methods to select ``efficiency-increasing'' SNPs for the dIVW estimator,  when a third GWAS dataset is available for screening. The first one is straightforward and capable of eliminating IVs with no association to the exposure. The second one is data-driven and iteratively selects a threshold, leading to the most efficient estimator in a given class. % including the dIVW estimator without screening.
\end{enumerate}

The rest of this paper is organized as follows. Section \ref{sec: notation and setup} introduces notation, setup,  and assumptions. Section \ref{sec: IVW} characterizes the consistency and asymptotic normality of the IVW estimator and the IVW estimator with screening. Section \ref{sec: dIVW} proposes the dIVW estimator and dIVW estimator with screening for improving efficiency, and
establishes their asymptotic properties. Also included in Section \ref{sec: dIVW} are  two methods for selecting a threshold for screening,  and an extension of the dIVW estimator to balanced horizontal pleiotropy \cite{Solovieff:2013aa, Verbanck:2018aa, Hemani:2018aa}.  Results from simulation studies and a real data analysis are presented in Sections \ref{sec: simu} and \ref{sec: real}, respectively. We conclude with a summary and discussion. Technical proofs and some additional results are in the Supplementary Material.

\section{Notation, Setup, and Assumption}
\label{sec: notation and setup}
As part of a common data cleaning and pre-processing step in MR studies,  millions of SNPs are de-correlated through linkage disequilibrium pruning or clumping via software \cite{Hemani:2018ab}. We assume that this initial step produces  $p$ independent SNPs, represented by bounded and mutually independent random variables $Z_1,..., Z_p$.

Let $X$ be the exposure and $Y$ be the continuous outcome. 
Following the two-sample summary-data MR literature \cite{Pierce:2013aa, Bowden:2017aa, zhao2018statistical}, we assume models 
\begin{align}
	&X=\sum_{j=1}^p \gamma_j Z_j + \eta_X U+E_X, \label{eq: exposure-IV}\\
	&Y=\beta_0 X %+\sum_{j=1}^{p} \alpha_jZ_j
	+\eta_Y U+E_Y, \label{eq: outcome-exposure}
\end{align} 
where $\eta_X$, $\eta_Y$, $\beta_0$, $\gamma_1,...,\gamma_p$ are unknown parameters, $U$ is an unmeasured confounder independent of $Z_1,...,Z_p$, $E_X$ and $E_Y$ are mutually independent random noises that are also independent of $ (Z_1, \dots, Z_p, U)$, 
and $U$, $E_X$ and $E_Y$ have finite 4th order moments. 
%For discrete outcome  such as binary $Y$, model (\ref{eq: outcome-exposure})  can be replaced by a generalized linear model such as a logistic model \textcolor{red}{Are we planning to explain this further or cite Qingyuan's paper?}.

The goal in an MR analysis is to estimate the effect of the exposure $X$ on the  outcome  $Y$, which is represented by $\beta_0$. 
Since the unobserved $U$ is  related with $X$, estimating $\beta_0$ using only model (\ref{eq: outcome-exposure})  with ordinary least squares leads to biased estimates. Instead, an MR approach to estimating $\beta_0$ typically assumes a model for $X$ like (\ref{eq: exposure-IV}) and makes three core assumptions
\cite{ Davey-Smith:2003aa, Smith:2004aa,  Lawlor:2008aa}. The first assumption is that instruments are associated with the exposure $X$, which amounts to $\gamma_j$'s in model (\ref{eq: exposure-IV}) not simultaneously  being zero.  We call an instrument with $\gamma_j \neq 0$ to be a relevant or non-null IV and an instrument with $\gamma_j = 0$ to be an irrelevant or  null IV. The second assumption is that instruments are independent of the unmeasured confounder $U$; this is encoded by assuming $ U$ is independent of $Z_{1},\ldots,Z_{p}$ in (\ref{eq: exposure-IV})-(\ref{eq: outcome-exposure}). 
The third and last core assumption is that instruments affect the outcome $Y$ only through the exposure $X$; this is true under  (\ref{eq: exposure-IV})-(\ref{eq: outcome-exposure}) since  (\ref{eq: outcome-exposure}) does not involve $Z_j$'s. 
However, this last assumption may be violated in some studies; see Section \ref{subsec: pleiotropy} for one example based on balanced horizontal pleiotropy. For more detailed discussions on the core assumptions, models, and their implication in MR, see \cite{Didelez:2007aa} and \cite{Bowden:2017aa}. %The paper will focus on estimation and inference.
%We also let $\calS=\{j\in \calP: \gamma_j\neq 0\}$ be the set of relevant IVs, $\calS^c=\{j\in \calP: \gamma_j= 0\}$ be the set of null IVs, and $s = |\calS|$ be the number of relevant IVs. 

In classic IV settings, estimation of $\beta_0$ is based on $n$ independent and identically distributed (i.i.d.) observations of $(Z_1,\ldots,Z_p, X, Y)$. 
In two-sample MR, estimation is based on  $n_X$ i.i.d. observations of $(X, Z_1, \dots, Z_p)$ from the exposure dataset and $n_Y$ i.i.d. observations of $(Y, Z_1, \dots, Z_p)$ from the outcome dataset. 
The two datasets are assumed to be independent of each other and we never jointly observe $Y$ and $X$. 

In two-sample summary-data MR, which is the most popular data setting in MR and the setting considered in this paper, only summary statistics from the exposure and outcome datasets are available for analysis, not the individual-level data. Specifically, from the exposure dataset, we have $\hat{\gamma}_j$, the ordinary least square estimate from a linear regression of $X$ on $Z_j$, and its  SE $\hat\sigma_{Xj}$, $j=1,...,p$. From the outcome dataset, we obtain $\hat{\Gamma}_j$, the ordinary least square estimate from a linear regression of $Y$ on $Z_j$, and its SE $\hat\sigma_{Yj}$, $j=1,...,p$. Note that models (\ref{eq: exposure-IV})-(\ref{eq: outcome-exposure}) and the independence of instruments  imply that $\hat{\gamma}_j$  consistently estimates $\gamma_j$ and $\hat{\Gamma}_j$ consistently estimates $\beta_0 \gamma_j$ for each $j$.

Many of the $p$ SNPs in (\ref{eq: exposure-IV}), produced by the de-correlation step, could be  potentially weak IVs with zero or small values of $\gamma_j^2 {\rm Var}(Z_j)$.  
It is therefore common in MR studies to screen  IVs and  include only selected  IVs in the IVW estimator.
To avoid selection bias or the ``winner's curse'', it is usually recommended to use a third independent dataset of size $n_{X^*}$ under model (\ref{eq: exposure-IV}), called the selection dataset, solely for screening IVs
\cite{Burgess:2013aa,Zhao:2019aa}.
Typically, because only summary statistics are available from the selection dataset, thresholding is applied to screen out SNPs in (\ref{eq: exposure-IV}) with the smallest marginal z-scores calculated from the summary statistics in the selection dataset. Future research will analyze a more sophisticated IV selection that simultaneously incorporates  de-correlation and IV strength and its effects on estimation.
%through linkage disequilibrium  pruning or clumping that produces  $p$ independent SNPs 	$Z_1,...,Z_p$,  many included SNPs  may still be weak IVs  as described in the beginning of Section \ref{sec: intro}. 

Formally, the IVW estimator with screening is a hard-thresholding estimator with a z-score threshold $\lambda \geq 0$, 
\begin{equation}
	\hat{\beta}_{\lambda, \rmivw} =\frac{\sum_{j\in S_\lambda }\hat w_j \hat{\beta}_j}{\sum_{j\in S_\lambda } \hat{w}_j}  = \frac{\sum_{j\in S_\lambda }\hat{\Gamma}_j \hat{\gamma}_j\hat\sigma_{Yj}^{-2}}{\sum_{j\in S_\lambda } \hat{\gamma}_j^2\hat\sigma_{Yj}^{-2}} , \quad S_\lambda = \{ j: | \hat\gamma_j^* | > \lambda \hat\sigma_{Xj}^*  \} \label{eq: llivw}
\end{equation}
where $\hat{\gamma}^*_j$ and $ \hat\sigma_{Xj}^* $ are counterparts of $\hat\gamma_j$ and $\hat\sigma_{Xj}$  computed from the selection dataset. If $\lambda = 0$, then $\hat{\beta}_{\lambda,\rmivw}$ reduces to the original IVW estimator $\hat{\beta}_{\rmivw}$ in \eqref{eq: ivw}. If $\lambda > 0$, only instruments with absolute value of z-scores higher than $\lambda$ are selected into the IVW estimator. A value for $\lambda$ that is used widely in MR is the genome-wide significance threshold   $\lambda \approx 5.45$, which corresponds to  screening out IVs whose p-values associated with  $\hat{\gamma}_{j}$'s are above the genome-wide significance level  $5 \times 10^{-8}$. More discussions about  this 
genome-wide significance level 
are given in later sections. 

%{\red The screening set $S_{\lambda}$ in (\ref{eq: llivw}) mainly concerns IV strength, as 
%de-correlation has been done in the pre-processing step previously described.
%IV selection that simultaneously incorporates  de-correlation  and IV strength in  estimation will be an interesting direction for future research}.

We make the following two assumptions for our asymptotic analysis.
\begin{assump} \label{assump: 1} 
	The sample sizes  $n_{X}$ and $n_Y$ (and $n_{X^*}$ of the selection dataset if it exists) diverge to infinity with the same order.  
	The number of SNPs, $p$, diverges to infinity. % the ratio ${\rm Var} (X) / {\rm Var} (Y)$	is bounded away from 0 and infinity. %; and $\max_j \{\gamma_j^2 {\rm Var}(Z_j)\}/{\rm Var}(X) \to 0$.  
\end{assump}
The conditions on sample sizes and $p$ are reasonable in our setup as many modern GWASs involve 10 to 100 thousands of participants and a few thousands of SNPs  are typically found to be independent after the de-correlation pre-processing step. % The last condition regarding $\gamma_j^2 {\rm Var}(Z_j)$ {\red is reasonable since} 
%an individual SNP typically only explains a small proportion of the total
%variance in the exposure. 

%Consider the formulas for the variances of $\hat{\gamma}_{j}$ and $\hat{\Gamma}_j$ from ordinary least squares,
%	\[\sigma_{Xj}^2=\frac{\var (X)-\gamma_j^2\var(Z_j)}{n_X\var (Z_j)}\quad \mbox{and} \quad \sigma_{Yj}^2=\frac{\var (Y)-\beta_0^2\gamma_j^2\var(Z_j)}{n_Y\var (Z_j)}. \]
%the terms $\gamma_j^2\var(Z_j)$ and $\beta_0^2\gamma_j^2\var(Z_j)$ are small compared to $\var(X)$ and $\var(Y)$ and 

The next assumption about summary statistics  is also assumed in 
\cite{zhao2018statistical}. 

\begin{assump} \label{assump: 2}
	$\{ \hat{\gamma}_j,\hat{\Gamma}_j, \hat{\gamma}_j^*, j=1,...,p \}$ are mutually independent and, for  every $j$,  
	%$\sigma_{Xj}^2$, $\sigma_{Yj}^2$, $\sigma_{Xj}^{*2}$ are known,
	$\hat{\gamma}_j\sim N(\gamma_j,\sigma_{Xj}^2)$, $\hat{\Gamma}_j\sim N(\beta_0\gamma_j,\sigma_{Yj}^2)$,  and $\hat{\gamma}^*_j\sim N(\gamma_j,\sigma_{Xj}^{*2})$. The variance ratios $\sigma_{Xj}^2/\sigma_{Yj}^2$ and $\sigma_{Xj}^2/\sigma^{*2}_{Xj}$ for all $j$ are bounded away from 0 and infinity. 
\end{assump}

We briefly assess the plausibility of Assumption \ref{assump: 2}.
%Consider Assumption \ref{assump: 1} first. 
%Condition $\sum_{j=1}^p \gamma_j^2 \var (Z_j)/n_X \to 0$ holds when there are $s$ non-null SNPs with $s/n_X \to 0$ and bounded  $\gamma_j^2 $'s, or when many non-null SNPs have small $\gamma_j^2$'s. 
With large sample sizes under Assumption \ref{assump: 1}, the normality of $\hat{\Gamma}_j$, $\hat{\gamma}_{j}$, and $\hat{\gamma}_{j}^*$ is plausible. 
%For the same reason, the standard errors of these estimates are precisely estimated and assuming them to be known is a reasonable approximation. 
The two-sample MR data structure guarantees the independence of $\hat{\gamma}_j$'s and $\hat{\Gamma}_j$'s  (and $ \hat{\gamma}_j^* $'s if they exist). Also, two-sample MR prunes/clumps SNPs to be far apart in genetic distance and each SNP only explains a very small proportion of the total variance in the exposure and outcome variables, making the independence within  $\hat{\gamma}_j$'s, $\hat{\Gamma}_j$'s (and $\hat{\gamma}^*_j$'s if they exist) as well as the boundedness of variance ratios likely. 
%{\red Using a similar reasoning, the ratio $ \sigma_{Xj}^2/\sigma_{Yj}^2 $ is approximately   $ n_Yn_X^{-1} \var(X)/\var(Y)$, which is bounded away from 0 and infinity under Assumption \ref{assump: 1}. The same conclusion can be made about the ratio $ \sigma_{Xj}^2/\sigma_{Xj}^{*2} $.}  
Furthermore, if  $Y$ is binary, Assumption \ref{assump: 2} is a first-order local approximation of a logistic outcome model \citep{zhao2018statistical, Zhao:2019aa}. 

%Under Assumption \ref{assump: 1}, $c_-\leq \sigma_{Xj}/\sigma_{Yj} \leq c_+$ for all $j$  and two	fixed positive constants  $c_-$ and $ c_+$. This can be seen as follows. 

%If SNP $j$ is a rare variant, then $\var(Z_j)$ is small and hence, this ratio is about $n_Yn_X^{-1}\var(X)/\var(Y)$. 
%When SNP $j$ is a common variant, the terms $\gamma_j^2\var(Z_j)$ and $\beta_0^2\gamma_j^2\var(Z_j)$ are still small compared to $\var(X)$ and $\var(Y)$ because an individual SNP typically explains a small proportion of the total variance in the exposure and outcome variables. Thus, if the sample sizes  $n_{X}$ and $n_Y$ are of the same order, Assumption \ref{assump: 2} will hold in both cases.

We define the average strength of $p$ IVs as 
\begin{equation}
	\kappa=\frac{1}{p} \sum_{j=1}^p \frac{\gamma_j^2}{\sigma_{Xj}^2} %\qquad \mu_j=\frac{\gamma_j}{\sigma_{Xj}} . 
	\label{kappa}
\end{equation}
where  $\gamma_j/\sigma_{Xj}$ is a normalized effect of SNP $j$ on $X$. 
If $\kappa$ is small, SNPs are, on average, weakly associated with the exposure. If $\kappa$ is large, SNPs are, on average, strongly associated with the exposure. 
We also define the average strength of IVs for IVW estimators with screening,
\begin{equation}
	\kappa_\lambda=\frac{1}{p_\lambda}\sum_{j =1}^p\frac{\gamma_j^2}{\sigma_{Xj}^2} q_{\lambda,j},  
	\label{kappa1}
\end{equation}
where $q_{\lambda,j}=P(|\hat{\gamma}_j^*|>\lambda \sigma_{Xj}^*)$  and 
$p_\lambda=\sum_{j=1}^p q_{\lambda,j}$. Clearly, if $\lambda=0$ so that all the IVs are included in the IVW estimator,  $q_{\lambda,j}$, $\kappa_\lambda$, and $p_{\lambda}$ become 1, $\kappa$, and $p$, respectively. As we will see in Sections \ref{sec: IVW}-\ref{sec: dIVW}, the limiting values of $\kappa$ and $\kappa_\lambda$ play a key role in characterizing the asymptotic properties of the IVW estimators.

The parameters $\kappa_\lambda$, $\kappa$,  and $p_\lambda$
can be estimated by  $\hat\kappa_\lambda$, $\hat\kappa$, and $\hat{p}_\lambda$, respectively, where 
\begin{equation}\label{hatkappa}
	\hat{\kappa}_\lambda=  \frac{1}{\hat{p}_{\lambda}} \sum_{j \in S_\lambda }
	\frac{\hat{\gamma}_j^2 }{\hat\sigma_{Xj}^2} -1 , \qquad \hat\kappa = \hat\kappa_0 , \qquad 
	\hat{p}_\lambda= \mbox{the size of $S_\lambda$}  . \end{equation}
We later show how to use these estimators in practice to check the theoretical conditions underlying the properties of the IVW estimator with or without screening.

\section{Properties of the IVW Estimators}
\label{sec: IVW}
We study the consistency and asymptotic normality of the IVW estimators described in Sections \ref{sec: intro}-\ref{sec: notation and setup} under different limiting values of $\kappa$ and $\kappa_\lambda$ defined in (\ref{kappa}) and (\ref{kappa1}).

In the MR literature, it is common to assume 
that the standard deviations (SDs) $\sigma_{Xj}$, $\sigma_{Yj}$, and $\sigma_{Xj}^{*}$ 
(in Assumption \ref{assump: 2}) are known (e.g., \cite{Burgess:2013aa, Bowden:2015aa, zhao2018statistical, Zhao:2019aa, Qi:2019aa}) so that 
$\hat\sigma_{Yj} = \sigma_{Yj}$ and $\hat\sigma^*_{Xj} = \sigma^*_{Xj}$ are used in (\ref{eq: ivw}) and  (\ref{eq: llivw}).
This is  motivated by the fact that the sample sizes $ n_X$, $n_Y$, and $ n_X^* $ are usually very large in modern GWASs and the aforementioned references show empirically that such approximation works well  in practice. 
%In fact, we have not seen any theoretical investigation about this issue.
In this section, we confine ourselves to the situation where the SDs are known and $\hat\sigma_{Yj} = \sigma_{Yj}$ and $\hat\sigma^*_{Xj} = \sigma^*_{Xj}$. The study of more general and realistic case where $\hat\sigma_{Yj} \neq \sigma_{Yj}$ and $\hat\sigma^*_{Xj} \neq \sigma^*_{Xj}$ is deferred to Section \ref{sec: dIVW}.

In what follows, $\xrightarrow{P}$ denotes convergence in probability and 
$\xrightarrow{D}$ denotes convergence in distribution. 

\begin{thm} \label{theo: p}
	Assume models (\ref{eq: exposure-IV})-(\ref{eq: outcome-exposure}) and Assumptions \ref{assump: 1}-\ref{assump: 2}.  Also, assume that 
	$\hat\sigma_{Yj} = \sigma_{Yj}$ and $\hat\sigma_{Xj}^* = \sigma_{Xj}^*$ in (\ref{eq: ivw}) and (\ref{eq: llivw}). When $\beta_0 \neq 0$, 
	we have the following conclusions for either $\lambda =0$ or $\lambda >0$. 
	\begin{enumerate}
		\item[(a)] If $\kappa_\lambda/p_\lambda\rightarrow \infty$,  $\max_j (\gamma_j^2 \sigma_{Xj}^{-2} q_{\lambda,j}) / (\kappa_\lambda p_\lambda) \rightarrow 0$, and when $\lambda \neq 0$, $\kappa_\lambda \sqrt{p_\lambda}/\lambda^2$ $\rightarrow\infty$, then $\hat{\beta}_{\lambda,\rmivw}$ is consistent and asymptotically normal, i.e., 
		\begin{equation}
			V_{\lambda, \rmivw}^{-1/2} \left(\hat{\beta}_{\lambda,\rmivw} -\beta_0\right)\xrightarrow{D}N(0,1), \label{normal}
		\end{equation}
		where 
		\[V_{\lambda, \rmivw}=\frac{\sum_{j =1}^p [(w_j+\rho_j)q_{\lambda,j}+\beta_0^2\rho_j(w_j+3\rho_j)q_{\lambda,j}-\beta_0^2\rho_j^2q_{\lambda,j}^{2}]}{[\sum_{j=1}^p(w_j+\rho_j)q_{\lambda,j}]^2},
		\] 
		$w_j = \gamma_j^2 / \sigma_{Yj}^2$, and $\rho_j= \sigma_{Xj}^2/\sigma_{Yj}^2$, $j=1,...,p$.
		\item [(b)]	If $\kappa_\lambda\rightarrow \infty$ and when $\lambda \neq 0$, $\kappa_\lambda \sqrt{p_\lambda}/\lambda^2\rightarrow\infty$, then $\hat{\beta}_{\lambda,\rmivw} \xrightarrow{P} \beta_0$.
		\item[(c)] If $\kappa_\lambda\rightarrow c>0$ and $\sqrt{p_\lambda}/\max(1,\lambda^2)\rightarrow\infty$, then $$\hat{\beta}_{\lambda, \rmivw}-\beta_0\frac{\sum_{j =1}^p w_j q_{\lambda, j}}{\sum_{j =1}^p (w_j+ \rho_{j})q_{\lambda, j}}\xrightarrow{P} 0.$$
		\item[(d)] If $\kappa_\lambda\rightarrow 0$ and $\sqrt{p_\lambda}/\max(1,\lambda^2)\rightarrow\infty$, then $\hat{\beta}_{\lambda, \rmivw}\xrightarrow{P} 0$.
	\end{enumerate}
	When $\beta_0 = 0$, we have the following conclusion for either $\lambda =0$ or $\lambda >0$. 
	%{\red for a given threshold $ \lambda\geq 0 $},
	\begin{enumerate}
		\item[(e)] 	If 
		$\max_j (\gamma_j^2 \sigma_{Xj}^{-2} q_{\lambda,j}) / (\kappa_\lambda p_\lambda+ p_\lambda) \rightarrow 0$ and when $ \lambda\neq 0 $, $(\kappa_\lambda \sqrt{p_\lambda}+ \sqrt{p_\lambda})/\max(1, \lambda^2)$ $\rightarrow\infty$, then 
		$\hat{\beta}_{\lambda,\rmivw}\xrightarrow{P} 0$ and
		$V_{\lambda, \rmivw }^{-1/2} \hat{\beta}_{\lambda,\rmivw} 
		\xrightarrow{D}N(0,1)$, where 
		$ V_{\lambda , \rmivw}$ is the same as that in (\ref{normal}) with $\beta_0 = 0$. 
	\end{enumerate}
\end{thm}

We now elaborate the results in Theorem  \ref{theo: p}.

First, consider the case of $\lambda =0$, i.e., the IVW estimator $\hat\beta_{\rmivw}$  in (\ref{eq: ivw}) without screening.
Part (a)  of Theorem \ref{theo: p} is the only regime where $\hat\beta_{\rmivw}$  is consistent and asymptotically normal when $\beta_0 \neq 0$. The main condition in Theorem  \ref{theo: p}(a) under $\lambda =0$,  
%Although  $\max_j (\gamma_j^2\sigma_{Xj}^{-2} )/(\kappa p)$ $\rightarrow0$ says that a single SNP's strength is small compared to the total IV strength $\kappa p$, a reasonable assumption in  MR studies, 
$\kappa/p \rightarrow \infty$,  means that the average IV strength $\kappa$ diverges to infinity at a  rate faster than $p$, which is   unlikely in MR studies.
In fact,  it is shown in the Supplementary Material that,   approximately, 
$ \kappa /p \leq  n_X / p^2  $ and thus,
$\kappa / p \to \infty$ implies $n_X/ p^2  \to \infty$. This rate is unrealistic in a typical MR study where the number of de-correlated SNPs $p $ is one thousand and, even when $n_X $ is as large as one million, we still have $n_X /p^2 = 1$. 
To explain why  $\kappa / p \to \infty$  is needed for 
$\hat{\beta}_{\rmivw}$ to be asymptotically normal with mean $\beta_0$,
consider 
$$
\hat{\beta}_{ \rmivw}-\beta_0=\frac{\sum_{j =1}^p (\hat{\Gamma}_j\hat{\gamma}_j-\beta_0\hat{\gamma}_j^2)\sigma_{Yj}^{-2}}{\sum_{j =1}^p \hat{\gamma}_j^2\sigma_{Yj}^{-2}} $$
whose numerator and denominator have expectations $-\beta_0  \sum_{j=1}^p \rho_j$ and  \linebreak $\sum_{j=1}^p (w_j+v_j)$, respectively, 
as  $E(\hat{\Gamma}_j\hat{\gamma}_j )= E(\hat{\Gamma}_j)E(\hat{\gamma}_j )= \beta_0\gamma_j^2$ and $E(\hat{\gamma}_j^2)=\gamma_j^2+\sigma_{Xj}^2$. 
Thus, as noticed by \cite{zhao2018statistical}, the asymptotic bias of $\hat{\beta}_{\rmivw}$ is
$$ \mbox{abias} (\hat{\beta}_{\rmivw} )= 
\frac{-\beta_0  \sum_{j=1}^p \rho_j}{\sum_{j=1}^p (w_j+v_j)}.$$
It is shown in the Supplementary Material that 
$$V_{ 0, \rmivw}^{-1/2}\{\hat{\beta}_{\rmivw} -\beta_0  - \mbox{abias} (\hat{\beta}_{\rmivw} ) \}
\xrightarrow{D}N(0,1), $$
where $V_{0, \rmivw} $ 
%\[V_{0, \rmivw} =\frac{\sum_{j=1}^p [ (w_j+\rho_j )+\beta_0^2\rho_j(w_j+2\rho_j )]}{ [\sum_{j=1}^p (w_j+\rho_j)]^2} \]
is the asymptotic variance of $\hat\beta_{\rmivw}$ given in (\ref{normal}) with $\lambda =0$. % and is of {\red the order $( \kappa p )^{-1}$}.
When $\beta_0 \neq 0$, this means that $V_{0, \rmivw}^{-1/2}(\hat{\beta}_{\rmivw} -\beta_0  )\xrightarrow{D}N(0,1)
$ cannot hold unless  
\begin{equation}
	\frac{\mbox{abias($\hat{\beta}_{\rmivw}$)} }{V_{0, \rmivw}^{1/2} }=
	\frac{- \beta_0  \sum_{j=1}^p \rho_j}{[ \sum_{j=1}^p \{ (w_j+\rho_j )+\beta_0^2\rho_j(w_j+2\rho_j )\} ]^{1/2}}  \to 0  \label{conda}
\end{equation}
i.e., the asymptotic bias of $\hat{\beta}_{\rmivw}$ tends to $0$ faster than the 
standard deviation of $\hat{\beta}_{\rmivw}$.
Since the numerator of the right side of (\ref{conda}) has  order $p$ and the denominator of the right  side has order $\{\max ( p, \kappa p )\}^{1/2}$,
(\ref{conda}) holds if  $\kappa / p \to \infty$. 
%the key condition required by Theorem  \ref{theo: p}(a). 
We can easily construct an example in which   (\ref{conda}) does not hold when 
$\kappa / p \not \to \infty$; in fact, the quantity in (\ref{conda})  diverges to infinity when $\kappa$ is bounded. 

In short, our theory explains  
the numerical observations in the literature concerning poor normal approximation
to  $\hat{\beta}_{\rmivw} -\beta_0 $ in the presence of weak IVs. 

Second, consider the case where $\lambda >0$ in part (a) of Theorem \ref{theo: p}. Screening with $\lambda >0$ is a way to relax the condition required for the asymptotic normality of IVW estimator. 
Specifically, if $\lambda >0$, $ \kappa/p \to \infty$ in part (a) is replaced by
$\kappa_\lambda/p_\lambda\rightarrow \infty$ and $\kappa_\lambda \sqrt{p_\lambda}/\lambda^2\rightarrow\infty$. 
%The condition \linebreak $\max_j (\gamma_j^2\sigma_{Xj}^{-2}q_{\lambda,j})/(\kappa_\lambda p_\lambda)\rightarrow 0$  is  reasonable if a selected IV's strength is small compared to the total expected IV strength among selected instruments. 
The Supplementary Material shows that $\kappa_\lambda$ is approximately increasing in $\lambda$ % and since $p_\lambda$ is clearly decreasing in $\lambda$
and thus $\kappa_\lambda / p_\lambda \to \infty$ is weaker than $\kappa /p \to \infty$.
A similar analysis (Supplementary Material) shows that  the counterpart of the asymptotic bias and standard deviation ratio in (\ref{conda}) for $\hat\beta_{\lambda , \rmivw}$ is of the order $p_\lambda^{1/2}/(1+ \kappa_\lambda)^{1/2}$, which tends to 0 as $\kappa_\lambda / p_\lambda \to \infty$.
In short, screening reduces the bias but increases the standard deviation of the IVW estimator (\ref{eq: ivw}), which is how 
$\hat\beta_{\lambda , \rmivw}$ becomes consistent and asymptotically normal. 

%Unfortunately, choosing $\lambda$ to achieve this rate is difficult in practice. For example, $\lambda \approx 5.45$, which is the common genome-wide significance threshold f $ 5\times10^{-8}$, may not always satisfy part (a). Or if all IVs are weak, say $\gamma_j^2\leq c \sigma_{Xj}^2 $ for all $j$ and some positive constant $c$, then $\kappa_\lambda$ is bounded regardless the choice of $\lambda$ and $\kappa_{\lambda} / p_{\lambda}$ never diverges. Finally, the second way to achieve asymptotic Normality is if the true effect $\beta_0= 0$ as laid out in part (e) of Theorem \ref{theo: p}; in this case, the asymptotic bias in (\ref{conda}) equals zero and $\hat{\beta}_{\rmivw} $ is consistent and asymptotically normal. Practically speaking, the first case is cumbersome, requiring a good $\lambda$ and a third selection dataset, and the second case is an exception rather than the norm. The estimator we propose in Section \ref{sec: dIVW} overcome these obstacles.}

%The condition $\kappa_\lambda/p_\lambda\rightarrow\infty$ requires that all selected SNPs be strong. 
%{\red The asymptotic variance $	V_{\lambda, \rmivw}$ with $\lambda >0$ has  a  rather complicated expression; see  the Supplementary Material for details}.

Third, if we forgo asymptotic normality, part (b) of Theorem \ref{theo: p} shows that the IVW estimator, with or without screening, is consistent for non-zero $\beta_0$ if
the  average IV strength $\kappa_\lambda$ diverges to infinity.
%  We remark that unlike other cases, part (b) Theorem \ref{theo: p} does not  require $p \to \infty$. 
Although this condition is weaker than $\kappa_\lambda /p_\lambda \to \infty$ in part (a)  of Theorem \ref{theo: p},  it is still unlikely to be satisfied in typical MR studies, 
not to mention that the consistency of IVW estimators is not enough for assessing variability or making statistical inference on $\beta_0$.  
To complement part (b), parts (c) and (d) of Theorem \ref{theo: p} show that if the average IV strength   $\kappa_\lambda$ does not diverge 
to infinity, a common scenario in MR studies with many weak and null IVs, the IVW estimators are inconsistent and biased towards $0$. 

%In (a)-(d) of Theorem \ref{theo: p}, there is also an additional condition involving the selection threshold $\lambda$ when $\lambda >0$, which essentially controls the denominator of the IVW estimator with screening.

Finally, the last part (e) of Theorem \ref{theo: p} is for the special case of $\beta_0 = 0$, in which the  weak IV bias of $\hat{\beta}_{\rmivw} $ is not an issue because the asymptotic bias in (\ref{conda}) equals zero when $ \beta_0=0 $ and $\hat{\beta}_{\rmivw} $ is consistent and asymptotically normal under reasonable conditions. 

Result (\ref{normal}) still holds if we replace $V_{\lambda,\rmivw}$ by a plug-in consistent estimator
\begin{equation}
	\hat{V}_{\lambda,\rmivw}=\frac{\sum_{j\in S_\lambda} [ \hat w_j +\hat{\beta}_{\lambda,\rmivw}^{2}\hat \rho_j(\hat{w}_j+ \hat\rho_j )]
	}{ (\sum_{j\in S_\lambda}\hat w_j)^2},
	\label{evar2}
\end{equation}
where $\hat w_j = \hat\gamma_j^2 / \hat \sigma_{Yj}^2$, $\hat \rho_j = \hat \sigma_{Xj}^2/\hat \sigma_{Yj}^2$, and $ S_\lambda = \{ j: | \hat\gamma_j^* | > \lambda \hat\sigma_{Xj}^*  \} $.

Comparing the IVW estimator \eqref{eq: ivw} with the  IVW estimator \eqref{eq: llivw} with screening, the former requires far more stringent conditions on IV strength to guarantee its consistency or asymptotic normality, whereas  the latter requires
finding a  threshold $\lambda$ and checking whether  $\lambda$  satisfies conditions in Theorem \ref{theo: p}(a), which is cumbersome and  not always successful. 
We highlight some examples below. 
% Table \ref{tb:simu} in Section \ref{sec: simu} provides numerical illustrations. 
\begin{enumerate}
	\item Consider the common practice of  selecting IVs that pass the p-value threshold of $ 5\times10^{-8}$, which as mentioned earlier is equivalent to setting $\lambda \approx 5.45$. This $\lambda$ may or may not satisfy the conditions for consistency or asymptotic normality of $\hat{\beta}_{\lambda, \rmivw}$ in Theorem \ref{theo: p}(a)-(b).
	% $n_X$ is much larger than $\log p$, the relevant IVs have probability close to 1 to being selected and hence $p_\lambda^*\approx s$, the number of relevant IVs. Therefore, $\kappa_\lambda^*$ is of order $n_X$ and $\kappa_\lambda^*/p_\lambda^*$ is of order $n_X/s$, which can be large for the conditions in Theorem \ref{theo: three sample, IVW}(a) to be satisfied.
	\item If all IVs are very weak, for example $\gamma_j^2\leq c \sigma_{Xj}^2 $ for all $j$ and  some positive constant $c$, then $\kappa_\lambda$ is bounded regardless of the choice of $\lambda$. %In this case, the GWAS threshold may select weak IVs and bias the  IVW estimator with screening. 
	\item If every IV strength equals $\lambda$, then $q_{\lambda, j}\approx 1/2$ and $\kappa_\lambda\approx\lambda^2$. But, $\kappa_{\lambda}/p_\lambda \approx 2\lambda^2/p$ may be small, implying that the asymptotic normality in Theorem \ref{theo: p}(a) may not hold.
\end{enumerate}

\section{Debiased IVW Estimators}
\label{sec: dIVW}

%While the  IVW estimator with screening requires less stringent assumptions than the IVW estimator without screening, we believe that finding a  threshold $\lambda$ and checking whether the selected $\lambda$  satisfies conditions in Theorem \ref{theo: p}(a) is cumbersome. Furthermore, screening needs a third dataset independent of the original two datasets. In this section we develop an approach to remedy this.

\subsection{Debiased IVW Estimator}
Motivated by the stringent assumptions underlying the asymptotic normality of the IVW estimator without screening and the need of a third dataset and a
carefully chosen threshold $\lambda$ in the IVW estimator with screening, we propose a simple estimator of $\beta_0$ that  relies on neither. We name the new estimator as the debiased IVW (dIVW) estimator. It is the original IVW estimator multiplied by an explicit  bias correction factor, i.e.,
\begin{equation} \label{eq: dIVW}
	\hat{\beta}_{ \rm\rmdivw} = \hat{\beta}_{ \rmivw} \cdot \frac{\sum_{j=1}^{p} \hat w_j}{\sum_{j=1}^{p} (\hat w_j-\hat{v}_j)}
	%= \frac{\sum_{j=1}^{p} \hat{w}_j\hat{\beta}_j}{\sum_{j=1}^{p} (\hat w_j-\hat{v}_j)}
	=\frac{\sum_{j =1}^p \hat{\Gamma}_j\hat{\gamma}_j\hat\sigma_{Yj}^{-2}}{\sum_{j=1}^p (\hat{\gamma}_j^2-\hat\sigma_{Xj}^2)\hat\sigma_{Yj}^{-2} },
	%= \frac{\sum_{j=1}^p \hat w_j \hat{\beta}_j}{\sum_{j=1}^p (\hat w_j-\hat\sigma_{Xj}^2 \hat\sigma_{Yj}^{-2})} . 
\end{equation}
where $ \hat{w}_j=\hat{\gamma}_j^2/\hat{\sigma}_{Yj}^2$ and $ \hat{v}_j= \hat{\sigma}_{Xj}^2/ \hat{\sigma}_{Yj}^2$.
The bias correction factor amplifies the IVW estimator that is biased towards 0 according to Theorem \ref{theo: p}(c)-(d). 
% {\red \textbf{Not sure if the ``reality'' part should be highlighted since IVW estimator is also defined with SEs}} {\red Note that  SEs $\hat{\sigma}_{Xj}^2,  \hat{\sigma}_{Yj}$ are used in (\ref{eq: dIVW}), not the true SDs, which reflects reality}. 

Surprisingly, this simple correction makes the resulting estimator dramatically more robust to many weak IVs. To explain why, 
recall that the asymptotic normality of $\hat{\beta}_{ \rmivw} $ requires that its 
asymptotic bias tend to $0$ fast enough, i.e.,  (\ref{conda}) or 
the stringent condition $\kappa /p \to \infty$ holds. 
For the dIVW estimator, 
\[
\hat{\beta}_{\rmdivw}-\beta_0=\frac{\sum_{j  =1}^p (\hat{\Gamma}_j\hat{\gamma}_j-\beta_0\hat{\gamma}_j^2+\beta_0\hat\sigma_{Xj}^2)\hat\sigma_{Yj}^{-2}}{\sum_{j =1}^p(\hat{\gamma}_j^2-\hat\sigma_{Xj}^2)\hat\sigma_{Yj}^{-2}},
\]
the numerator has mean  zero under Assumption \ref{assump: 2}. 
%$$ E\{ (\hat{\Gamma}_j\hat{\gamma}_j-\beta_0\hat{\gamma}_j^2+\beta_0\sigma_{Xj}^2)\sigma_{Yj}^{-2} \} = ( \beta_0\gamma^2_j - \beta_0 \gamma_j^2 - \beta_0 \sigma_{Xj}^2 + \beta_0 \sigma_{Xj}^2) \sigma_{Yj}^{-2} = 0. $$
This indicates that, $\hat{\beta}_{\rmdivw}$  
has a negligible asymptotic bias,  and hence its asymptotic normality does not need a stringent condition such as $\kappa /p \to \infty$ to ensure (\ref{conda}).

As we show in the next section, the dIVW estimator (\ref{eq: dIVW}) is consistent and asymptotically normal if $\kappa\sqrt{p}\rightarrow\infty$ and  $\max_j (\gamma_j^2\sigma_{Xj}^{-2})/(\kappa p+p)\rightarrow 0$. The condition $\kappa\sqrt{p}\rightarrow\infty$ is considerably  weaker than $\kappa / p \to \infty$ required by the original IVW estimator (\ref{eq: ivw}). For example, $\kappa \sqrt{p} \to \infty$ holds even when  $\kappa \to 0$  but at a slower rate than $1/\sqrt{p}$; in contrast,  $\kappa /p \to \infty$ requires $\kappa \to \infty$. Also, when IVs are common variants, but are weak in the sense of \citet{Staiger:1997aa} (i.e., $\gamma_j$ and $\sigma_{Xj}$ are both of the order $n_X^{-1/2}$), the dIVW estimator still remains consistent and asymptotically normal if the number of such weak IVs is large. 
Finally, the condition $\kappa\sqrt{p}\rightarrow\infty$  is also related to conditions imposed by the limited information maximum likelihood (LIML) estimator in the one-sample individual-level data setting \cite{Chao:2005aa} and the robust adjusted profile score (MR-raps) estimator \cite{zhao2018statistical}.
%We remark that the latter assumes all SNPs to be common variants whereas the dIVW estimator can handle both common and rare variants}.

The quantity $\kappa \sqrt{p}$ can be interpreted as an effective sample size for the dIVW estimator and can be estimated by $\hat{\kappa}\sqrt{p}$ with $\hat\kappa$ defined in (\ref{hatkappa}). In our simulation studies (i.e., Figure \ref{figure: condition}), we provide some guidelines on what would be considered a large value of $\hat{\kappa}\sqrt{p}$ for the asymptotics promised to kick in. This is akin to qualitative guidelines on what would be a large enough sample size for a normal approximation of an estimator to hold.

\subsection{Improving Efficiency With Screening}
While the  dIVW estimator $\hat\beta_{\rmdivw}$  (\ref{eq: dIVW}) without screening is consistent and asymptotically normal even if many IVs are weak,  its asymptotic variance, 
$V_{0, \rmdivw}$  defined in (\ref{var}) with $\lambda =0$, is  larger than $V_{0, \rmivw}$, the asymptotic variance of the IVW estimator $\hat\beta_{\rmivw}$ (\ref{eq: ivw}) with $\lambda =0$; we remark that both 
$V_{0, \rmivw}$ and $V_{0, \rmdivw}$ have order $(\kappa p)^{-1}$ when 
$\hat\beta_{\rmdivw}$  and $\hat\beta_{\rmivw}$ are asymptotically normal.
The increase in variance of the dIVW estimator
is due to a bias-variance trade off between the IVW estimator and the dIVW estimator
and the bias due to weak IVs in MR studies tends to dominate the SD of the estimator. 

When summary statistics from an independent selection dataset are available, we explore how to make the dIVW estimator more efficient by screening. 
We remark here that screening in the dIVW estimator is solely for  efficiency improvement since $\hat\beta_{\rmdivw}$ without screening remains asymptotically normal under weak conditions.  In contrast, the IVW estimator uses screening to reduce bias and to achieve asymptotic normality.

Formally, consider the the dIVW estimator using only IVs selected from the selection dataset,
\begin{equation}\label{ldIVW}
	\hat{\beta}_{\lambda,\rmdivw} =\frac{\sum_{j\in S_\lambda}\hat{\Gamma}_j \hat{\gamma}_j\hat\sigma_{Yj}^{-2}}{\sum_{j\in S_\lambda}( \hat{\gamma}_j^2-\hat\sigma_{Xj}^2)\hat\sigma_{Yj}^{-2}}, \qquad S_\lambda = \{ j: | \hat\gamma_j^* | > \lambda \hat\sigma_{Xj}^*  \} 
\end{equation}

Theorem \ref{theo: dIVW} establishes the asymptotic normality of  $\hat{\beta}_{\lambda,\rmdivw}$ in (\ref{ldIVW}). The Theorem includes 
$\hat{\beta}_{ \rm\rmdivw}$ in (\ref{eq: dIVW}) as a special case of $\hat{\beta}_{\lambda,\rmdivw}$ when $\lambda =0$.
%Let $\kappa_\lambda$, $p_\lambda$ and $q_{\lambda ,j}$ be as defined in (\ref{kappa1}). 

\begin{thm} \label{theo: dIVW}
	Assume models (\ref{eq: exposure-IV})-(\ref{eq: outcome-exposure}),  Assumptions \ref{assump: 1}-\ref{assump: 2}, and that
	$\kappa_\lambda \sqrt{p_\lambda} /  \max ( 1, \lambda^2) \to \infty$ and 
	$\max_j (\gamma_j^2\sigma_{Xj}^{-2} q_{\lambda,j})/ (\kappa_\lambda p_\lambda+p_\lambda)\to 0$. Assume further that either 	$\hat\sigma_{Xj} = \sigma_{Xj}$, $\hat\sigma_{Yj} = \sigma_{Yj}$ and $ \hat{\sigma}_{Xj}^*= {\sigma}_{Xj}^* $ in  (\ref{eq: dIVW})-(\ref{ldIVW}) 
	or  $p / n_X \to 0$. Then, 
	$\hat{\beta}_{\lambda , \rmdivw}$ is consistent and asymptotically normal,  i.e., 
	$$
	V_{\lambda , \rmdivw}^{-1/2} (\hat{\beta}_{\lambda , \rmdivw}-\beta_0)\xrightarrow{D}N(0,1)
	$$
	where  
	\begin{equation}
		V_{\lambda, \rmdivw}=\frac{\sum_{j  =1}^p [  (w_j+\rho_j )+\beta_0^2\rho_j (w_j+2\rho_j )]q_{\lambda,j}}{ (\sum_{j=1}^p  w_jq_{\lambda,j})^2} \label{var}
	\end{equation} 	
	for $\lambda =0$ or $\lambda >0$, 
	$w_j = \gamma_j^2 / \sigma_{Yj}^2$, and $\rho_j= \sigma_{Xj}^2/\sigma_{Yj}^2$, $j=1,...,p$.
\end{thm}
For consistency and asymptotic normality, the stringent condition $\kappa_\lambda / p_\lambda $ $\to \infty$ required in Theorem \ref{theo: p}(a) for IVW estimators 
is not needed in Theorem \ref{theo: dIVW}. 
%For example, the condition for consistency of the pre-screened IVW estimator implies the  condition $\kappa^*_\lambda\sqrt{p^*_\lambda} / \max (1, \lambda^2) \rightarrow \infty$ which is required for the consistency of pre-screened dIVW estimator. 

Theorem \ref{theo: dIVW} shows that, if  $p / n_X \to 0$, then
using the estimates SEs leads to the same asymptotic result as using the true
SDs  in (\ref{eq: dIVW})-(\ref{ldIVW}).  
A parallel result can also be established for the IVW estimators but is omitted here. 
Thus, our result provides a theoretical justification for safely treating SEs as SDs, a commonly adopted approach in MR studies.  

%Also,  Theorem \ref{theo: dIVW} not only provides a theoretical support for the empirical findings in many MR studies, but also provides an explicit condition (i.e., $p/n_X \to 0$) for safely treating SEs as the true SDs. 

The condition $p / n_X \to 0$ typically holds since after de-correlation, the number of independent SNP is usually around a few thousands and the sample size is around 10 to 100 thousands.
Our simulation results in Section \ref{sec: simu} shows that the approximation is still very good even when $p / n_X$ is 20\%, indicating that $p / n_X \to 0$ is only sufficient rather than necessary.

The quantity $\kappa_\lambda\sqrt{p_\lambda} / \max (1, \lambda^2) $ acts like an effective sample size for the  dIVW estimator with screening and can be estimated by $\hat{\kappa}_\lambda\sqrt{\hat{p}_\lambda} /\max (1, \lambda^2)$ with 
$\hat{\kappa}_\lambda$ and $\hat{p}_\lambda$ given by (\ref{hatkappa}). In our simulation studies, specifically Figure \ref{figure: condition}, we provide some guidelines on what would be considered a large effective sample size for the asymptotic result to kick in. 

The results in Theorem \ref{theo: dIVW} still hold  if we replace the asymptotic variance $V_{\lambda , \rmdivw}$ with a consistent estimator 
\begin{equation}
	\hat{V}_{\lambda,\rmdivw}=\frac{\sum_{j\in S_\lambda} [\hat w_j +\hat{\beta}_{\lambda,\rmdivw}^{2}\hat \rho_j (\hat{w}_j+\hat \rho_j )]}{ [\sum_{j\in S_\lambda} (\hat{w}_j-\hat\rho_j)]^2},
	\label{evar}
\end{equation}
where $\hat w_j = \hat\gamma_j^2 / \hat \sigma_{Yj}^2$, $\hat \rho_i = \hat \sigma_{Xj}^2/\hat \sigma_{Yj}^2$, and $ S_\lambda = \{ j: | \hat\gamma_j^* | > \lambda \hat\sigma_{Xj}^*  \} $. 

\subsection{Choice of $\lambda$ in Screening}

We consider the choice of $\lambda$ in the  dIVW estimator (\ref{ldIVW}) with screening.  In general, the threshold $ \lambda $ should satisfy $\kappa_\lambda \sqrt{p_\lambda} /\max ( 1, \lambda^2) \to \infty$, as well as increase the efficiency of the dIVW estimator.

%The first choice is simple and sets $\lambda = \sqrt{2\log p}$. \textcolor{orange}{\textbf{Is this true? $\Rightarrow$} Not only does this $\lambda$ satisfy the rate condition $\kappa_\lambda \sqrt{p_\lambda} /\max ( 1, \lambda^2) \to \infty$ in Theorem \ref{ldIVW}}, but also as $p \to \infty$, the probability of selecting any null IV is very small under Assumptions \ref{assump: 1}-\ref{assump: 2}, 
One choice is $\lambda = \sqrt{2\log p}$ that diverges to infinity at a very slow rate. This $\lambda$ guarantees that the probability of  selecting any null IV is very small, because under Assumptions \ref{assump: 1}-\ref{assump: 2} and $p \to \infty$,
$$
P(\mbox{at least one null IV is selected})
\leq\frac{ 2(p-s)}{\lambda\sqrt{2\pi}}e^{-\lambda^2/2} 
= \frac{p-s}{p \sqrt{\pi \log p}}
\rightarrow 0 
$$
where $s$ is the number of non-null IVs. 

Another choice of $\lambda$ is motivated by directly studying the asymptotic variance $V_{\lambda, \rmdivw}$ in (\ref{var}), which has order  $(\kappa_\lambda p_\lambda)^{-1}$ when $\kappa_\lambda  \not \to 0$
and 	$(\kappa_\lambda^2 p_\lambda)^{-1}$ when $\kappa_\lambda   \to 0$.
To illustrate the idea, we focus on the situation where 
$\kappa \not \to 0$ so that the asymptotic variances of $\hat \beta_{\rmdivw}$ and $\hat \beta_{\lambda ,\rmdivw}$ have orders $(\kappa p)^{-1}$ and  $(\kappa_\lambda p_\lambda)^{-1}$, respectively. 
Since $\kappa_\lambda p_\lambda \leq \kappa p$ for any $\lambda >0$, to screen for efficiency rather for relevant IV selection,
we should not screen out too many non-null IVs (even if they are weak) to result in $\kappa_\lambda p_\lambda / \kappa p \to 0$. From this point of view, 
$\lambda = \sqrt{2\log p}$ is an improvement over the genome-wide significance p-value threshold $5\times 10 ^{-8}$ ($\lambda \approx 5.45$) because
$\sqrt{2\log p}<5.45$ if $p < 10^6$, and using
$\lambda = \sqrt{2\log p}$ eliminates null IVs with probability close to 1 as the previous discussion indicated. 
%	As an example of comparing $V_{0, \rmdivw}$ with $V_{\lambda, \rmdivw}$  with $\lambda >0$, 	consider the situation where  there are $ s_w $ weak IVs with $ \gamma_j^2 \sigma_{Xj}^{-2} = h_w $, $ s_s $ strong IVs with $ \gamma_j^2 \sigma_{Xj}^{-2} = h_s $, $ p-s_w-s_s $ null IVs, $ s_wh_w=s_sh_s $, and 	$ v_j =  v $ for all $j$. Then 	\begin{align} V_{0, \rmdivw}=\frac{p (1+2\beta_0^2v) + 2(1+\beta_0^2v) s_sh_s}{ 4vs_s^2h _s^2} \nonumber	\end{align}
%	Taking $ \lambda=\sqrt{2\log p} $ and assuming that the weak IVs have probability $ c $ of getting selected and the strong IVs have probability 1 of getting selected, we obtain
%	\begin{align} V_{\lambda , \rmdivw} = \frac{\delta + (cs_w+s_s) (1+2\beta_0^2 v) + (1+\beta_0^2v) s_sh_s(1+c)}{ v(1+c)^2s_s^2h_s^2} \nonumber	\end{align}
%	where $\delta$ is  a term $\to 0$. Therefore, if  $$	p> \frac{4(cs_w+s_s)}{(1+c)^2} + \frac{4(1+\beta_0^2v) s_s\mu_s^2}{1+2\beta_0^2v} \left( \frac{1}{1+c}-\frac{1}{2}\right)	$$
%	then $ V_{0,\rmdivw} > V_{\sqrt{2\log p}, \rmdivw}$. As an illustration, when $ s_w=100, s_s=10, h_w=10, h_s=100,  c=0.3 $, the dIVW estimator using $ \lambda=\sqrt{2\log p} $ is asymptotically more efficient than that without thresholding when} $ p>1097$. 
But, if there are many weak IVs, $\hat\beta_{\rmdivw}$ 
may be asymptotically more efficient than $\hat\beta_{\lambda , \rmdivw}$ with any $\lambda >0$ (see Case 3 of the simulation study in Section \ref{subsec: simu BMI-CAD}). In short, it is better if we can select $\lambda$ adaptively. 
%Since our purpose of screening is to improve the efficiency of  $\hat{\beta}_{ \rmdivw}$,  the asymptotic variance $V_{\lambda, \rmdivw}$ in (\ref{var}) can be used as  a basis to choose $\lambda$. 

This leads to our approach of choosing $\lambda$ that directly minimizes an estimated asymptotic variance of $\hat{\beta}_{\lambda,\rmdivw} $, 
which we call the  Mendelian Randomization Estimation-Optimization (MR-EO) algorithm. In a nutshell, MR-EO considers  $\hat{\beta}_{\lambda,\rmdivw}$ with  $\lambda$ that varies in the interval $[0,\sqrt{2\log p} \, ]$; it assumes that $\kappa_\lambda\sqrt{p_\lambda}/\max(1,\lambda^2)\rightarrow\infty$ holds for every $\lambda$ in the  range. 
%where $\lambda =0$ amounts to including every IV. 
It then tries to find the $\lambda$ in this range  that minimizes the asymptotic variance. Since we cannot directly use estimated variance in (\ref{evar}) because $\hat\beta_{\lambda , \rmdivw}$ is not available prior to the selection of $\lambda$, 
% If $\beta_0$ is known, then the most efficient estimator $\hat{\beta}_{\lambda, \rmdivw}$ among $\lambda \in [0, \sqrt{2 \log p}]$ is the one with $\lambda$ that minimizes 
%\[\hat{V}_{\lambda,\rmdivw} (\beta_0)=\frac{\sum_{j: |\hat{\mu}^*_j|>\lambda}\left[v_j^2 \hat{\mu}_j^2+\beta_0^{2}v_j^4(\hat{\mu}_j^2+1)\right]}{ \left[\sum_{j: |\hat{\mu}^*_j|>\lambda}v_j^2 (\hat{\mu}_j^2-1)\right]^2}.\]
%However, $\beta_0$ is unknown.  
MR-EO alternates between estimating  the exposure effect by $\hat\beta_t$ (i.e., the E-Step) and finding the optimal $\lambda$ given the previous estimate  $\hat\beta_t$  (i.e., the O-Step); see Algorithm \ref{lalg} for details. %where
%$\hat{V}_{\lambda,\rmdivw} (\beta) $ denotes the estimated variance in (\ref{evar}) with 
%$\hat{\beta}_{\lambda, \rmdivw}$ replaced by $\beta$. 

\begin{algorithm} 
	\caption{MR-EO algorithm to determine the optimal $\lambda$}\label{lalg}
	\SetAlgoLined 
	Initialize $t=0$, $t_{\rm max}$, $\lambda_{0} = \sqrt{2\log p}, V=\infty$\;
	\While{$t \leq  t_{\rm max}$}{
		E-Step: for a given $\lambda_t$, estimate $\beta_0$ with the dIVW estimator $\hat{\beta}_{\lambda_t, \rmdivw}$\;
		\eIf{ $V \leq \hat{V}_{\lambda_t,\rmdivw} (\hat{\beta}_{\lambda_t,\rmdivw}) $}{exit the while loop\;}{$V=\hat{V}_{\lambda_t,\rmdivw} (\hat{\beta}_{\lambda_t,\rmdivw})$\;}
		
		O-Step:  Plug $\hat{\beta}_{\lambda_t, \rmdivw}$ into the variance estimator and find 
		\[
		\lambda_{t+1} = \arg\min_{\lambda \in [0,\sqrt{2\log p}] } \hat{V}_{\lambda,\rmdivw} (\hat{\beta}_{\lambda_t,\rmdivw})
		\]
		Set $t=t+1$;
	}
	
	Output $\lambda_{t-1}$.\vspace{3mm}
\end{algorithm}

We make some comments regarding the implementation of MR-EO and its final output. First, we initialize MR-EO to $\lambda_0 =\sqrt{2\log p}$ and force the algorithm to stop at $t=t_{\rm max}$ with a reasonably large $t_{\rm max}$, mainly for computational efficiency. Second, the algorithm assumes that $\kappa_\lambda\sqrt{p_\lambda}/\max(1,\lambda^2)\rightarrow\infty$ holds for every $\lambda$ in $[0, \sqrt{2\log p} \, ]$ so that $\hat{\beta}_{\lambda_t, \rmdivw}$ is consistent for $\beta_0$. To verify this, we can empirically evaluate $\hat{\kappa}_\lambda\sqrt{\hat{p}_\lambda}/\max(1,\lambda^2)$ and check whether this quantity is reasonably large for all $\lambda$ in the range. In our simulation study in Section \ref{sec: simu}, we find that the range  $[0, \sqrt{2\log p}]$ works well. %Other ranges of $\lambda$ can also be considered. 
Third, with fixed $t_{\rm max}$ and range of $\lambda$, MR-EO produces a unique dIVW estimator. 
Finally, the estimated variance for the dIVW estimator based on MR-EO may be too optimistic due to the ``winner's curse''; however, when both $p$ and $n_X$ are large and the ratio $p/n_X$ is bounded, ideally small, this issue will be largely moot. 
In our simulation studies in Section \ref{sec: simu}, we find that the estimator chosen by MR-EO performs well and the resulting confidence interval maintains nominal coverage, 
although Theorem \ref{theo: dIVW} does not directly guarantee that the estimator chosen by MR-EO is asymptotically normal.

\subsection{Extension to Balanced Horizontal Pleiotropy} 
\label{subsec: pleiotropy}
We extend the dIVW estimator to situations under one type of pleiotropy in MR, balanced horizontal pleiotropy \cite{Hemani:2018aa, zhao2018statistical, Bowden:2017aa}. Briefly, under balanced horizontal pleiotropy, the third core IV assumption described in Section \ref{sec: notation and setup} is violated and the model 
(\ref{eq: outcome-exposure}) is extended to 
\begin{equation}
	Y=\beta_0 X +\sum_{j=1}^{p} \alpha_jZ_j
	+\eta_Y U+E_Y, \label{eq: outcome-exposure1} 
\end{equation}
where the pleiotropic effects of $p$ SNPs on $Y$, $\alpha_1,...,\alpha_p$, $\alpha_j \sim N(0,\tau_0^2)$,  are independent random effects and independent of $X$, $Z_j$'s, $U$, $E_Y$ and $E_X$.
To incorporate balanced pleiotropy, we replace Assumption \ref{assump: 2} with the following assumption \cite{Hemani:2018aa, zhao2018statistical, Zhao:2019aa}.

\begin{assumption2prime} Suppose Assumption \ref{assump: 2} holds except conditional on $ \alpha_j$, $\hat{\Gamma}_j  \sim N(\alpha_j+\beta_0\gamma_j,\sigma_{Yj}^2)$ for every $j$. 
	In addition,  for some constant $c_+$, $\tau_0\leq c_+ \sigma_{Yj}$ for all $j$.
\end{assumption2prime}
Under the same conditions in Theorem \ref{theo: dIVW} with (\ref{eq: outcome-exposure}) and Assumption \ref{assump: 2} replaced by (\ref{eq: outcome-exposure1}) and Assumption $2'$, respectively, the Supplementary Material shows that the dIVW estimators in (\ref{eq: dIVW}) and (\ref{ldIVW}) are still consistent and asymptotically normal. However, the variance of the dIVW estimators is larger due to the random effects $ \alpha_j $'s and an estimator of it under balanced pleiotropy is
\begin{equation}\label{evar5}
	\frac{\sum_{j \in S_\lambda}[  \hat{w}_j(1+\hat{\tau}^2\hat{\sigma}_{Yj}^{-2})+\hat{\beta}_{\lambda,\rmdivw}^2\hat{v}_j(\hat{w}_j+\hat{v}_j)]}{ [\sum_{j \in S_\lambda}(\hat{w}_j-\hat{v}_j)]^2},
\end{equation}
where
\[
\hat{\tau}^2=\frac{\sum_{j =1}^p[(\hat{\Gamma}_j-\hat{\beta}_{\rmdivw}\hat{\gamma}_j)^2-\hat\sigma_{Yj}^2-\hat{\beta}_{\rmdivw}^{2}\hat\sigma_{Xj}^2]\hat\sigma_{Yj}^{-2}}{\sum_{j=1}^p\hat\sigma_{Yj}^{-2}}.
\]
We remark that the estimator of $\hat{\tau}^2$ relies on $\hat{\beta}_{\rmdivw}$. Also, if $\max_k \sigma_{Yk}^{-2}$ is bounded by a constant times the average $p^{-1} \sum_{j=1}^p \sigma_{Yj}^{-2} $, as $\kappa \sqrt{ p}\rightarrow \infty$,  the variance estimator in (\ref{evar5}) is consistent. We can use the aforementioned methods (e.g., MR-EO) to choose $\lambda$ and improve efficiency. 

Finally, when balanced horizontal pleiotropy does not hold, the dIVW estimator,  like other MR estimators built upon this assumption, will be biased. In Section 2 of the Supplementary Material, we investigate the magnitude of this bias.

\section{Simulation Studies}
\label{sec: simu}
\subsection{A Simulation with the	BMI-CAD Dataset as Population}
\label{subsec: simu BMI-CAD}
We conduct a simulation study  to 
compare   the finite sample properties of several estimators under different screening thresholds.  To closely mirror what is done in practice, we adopt a real two-sample summary-level MR dataset, the BMI-CAD dataset in the \emph{mr.raps} R package (version 0.3.1) of  \citet{Zhao:2019aa}, as  the simulation population. The BMI-CAD dataset is used to make inference about the effect of $X$,  the body mass index (BMI),   on $Y$, the risk of coronary artery disease (CAD). It contains  three independent datasets:
\begin{enumerate}
	\item Exposure dataset: A GWAS for BMI in round 2 of the UK BioBank (sample size: 336,107) \citep{Abbott2018};
	\item Outcome dataset: A GWAS for CAD from the CARDIoGRAMplusC4D consortium  (sample size: $\approx$185,000), with genotype imputation using the 1000 Genome Project \citep{CARDIoGRAMplusC4D-Consortium:2015aa};
	\item Selection dataset: A GWAS for BMI in the Japanese population (sample size: 173,430) \citep{Akiyama:2017aa}. 
\end{enumerate}
The three datasets have been cleaned so that (i) SNPs appear in all three datasets and (ii)  SNPs are far apart in genetic distance; see \cite{Zhao:2019aa} for details. The initial data cleaning leads to $p=1119$ SNPs available for analysis. Each GWAS contains publicly available summary statistics that are the estimated coefficients from marginal linear regression %$\hat\gamma_j$, $\hat \Gamma_j$, $\hat\gamma_j^*$, $j=1,...,p$, 
and their SEs. We use them as population parameters in our simulation; in Section \ref{sec: real}, we use them as data.
%From the exposure dataset, the estimated average IV strength is $\hat{\kappa} = 6.8$. From the selection dataset, the estimated average IV strength is $\hat{\kappa} = 7.7$.

To begin, we construct three plausible sets of $\gamma_j$ as follows.
\begin{itemize}
	\item[ \bf Case 1](Some strong IVs, many null IVs): 
	There are $s = 20$ non-null IVs whose $\gamma_j$-values are the 20 marginal regression coefficients with the smallest p-values  in the BMI-CAD exposure dataset. The rest $p-s = 1099$ SNPs are null IVs with zero $\gamma_j$'s. Combined, we have a ``population'' with $\kappa=2.90$ and $\kappa \sqrt{p}=97.00$. 
	\item[{\bf Case 2}] (Many weak IVs, many null IVs): 
	This setting is identical to Case 1, except we use the first $s=100$ marginal regression coefficients in the BMI-CAD exposure dataset as non-null SNPs and set the rest $p-s=1019$ SNPs as null IVs. This leads to $\kappa=1.05$ and $\kappa \sqrt{p}=35.12$. 
	\item[{\bf Case 3}] (Many weak IVs, no null IVs):  This setting is identical to Case 1, except we use all $p=s=1119$  marginal regression coefficients in the BMI-CAD exposure dataset as  non-null SNPs and there are no null SNPs. This leads to $\kappa=7.78$ and $\kappa \sqrt{p}=260.25$.
	%\item[{\bf Case 4}] Case 3, with balanced horizontal pleiotropy.
\end{itemize}
Based on the $\gamma_j$'s, we set $\Gamma_j = \beta_0\gamma_j$ with $\beta_0 = 0.4$.

Next, for each simulation run  we generate summary statistics  $\{\hat{\Gamma}_j,\hat{\gamma}_j, \hat{\gamma}_j^*,$ $ j=1,...,p\}$
based on Assumption \ref{assump: 2} with $\gamma_j$'s as described for each of Cases 1-3
and the SEs in the BMI-CAD dataset as  $\sigma_{Xj}$,  
$\sigma_{Yj}$,  and $\sigma_{Xj}^*$, $j=1,...,p$.   
Since we cannot generate SEs (part of summary statistics) from this real-data setting 
for simulation, in each simulation run we set SEs to be the same as $\sigma_{Xj}$,  
$\sigma_{Yj}$,  and $\sigma_{Xj}^*$, $j=1,...,p$. 
This corresponds to treating SDs as SEs as described in the start of Section \ref{sec: IVW}, i.e., assuming that we know SD values.

We compare seven MR methods:  the IVW estimator 
introduced in Section \ref{sec: IVW}, the dIVW estimator proposed in Section \ref{sec: dIVW}, and five other methods in the literature,  
MR-Egger regression \citep{Bowden:2015aa}, weighted median estimator (MR-median) \citep{Bowden:2016aa}, weighted mode estimator (MR-mode) \cite{Hartwig:2017aa},  profile score estimator (MR-raps) \citep{zhao2018statistical}, and profile score with empirical partially Bayes shrinkage weights (MR-raps-shrink) \citep{Zhao:2019aa}. MR-Egger, MR-median and MR-mode are implemented in the \emph{MendelianRandomization} R package (version 0.4.1) \cite{Yavorska:2017aa}. To make the comparisons fair, we use the $l_2$ loss for MR-raps as implemented in the \emph{mr.raps} package. For every method except MR-raps, we also use different screening procedures, including $\lambda=0$ (no screening, all SNPs are included), $\lambda=5.45$ (p-value cutoff based on the  threshold of $5\times 10^{-8}$), and $\lambda=\sqrt{2\log p}$ ($\approx 3.75$ when $p=1119$). We also include the dIVW estimator with $\lambda$ determined by the MR-EO algorithm with the maximum number of iterations set to 
$t_{\max} = 5$
and used the \emph{optimize} function from R in the O-step. The MR-raps does not have any screening.  The default MR-raps-shrink always applies a type of screening through Bayes shrinkage with the independent selection dataset. 
\begin{table}
	\caption{Simulation results for Cases 1-3 based on 10,000 repetitions with $\beta_0 =0.4$;  $ \lambda $ for MR-EO is the simulation average;
		SD is the simulation standard deviation;  SE  is  the average of standard errors; CP is the simulation coverage probability of the 95\% confidence interval based on  normal approximation. }\label{tb:simu}
	\centering
	\begin{tabular}{llcrrrr} \hline\\[-2ex]
		\multicolumn{1}{c}{Case}                & \multicolumn{1}{c}{Method} & \multicolumn{1}{c}{$\lambda$}       & \multicolumn{1}{c}{mean} & \multicolumn{1}{c}{SD}    & \multicolumn{1}{c}{SE} & \multicolumn{1}{c}{CP}  \\ \hline \\[-2ex]
		1      & IVW & $0$                   & 0.260 & 0.069 & 0.069     & 46.9    \\
		$s=20$ & IVW&  5.45                   & 0.398 & 0.094 & 0.093     & 94.8    \\
		$p=1119 $%	 $\kappa=2.90$	
		& IVW&  $\sqrt{2\log p}=3.75$       & 0.398 & 0.087 & 0.087     & 95.1    \\ 
		&&&&&\\
		& dIVW&  $0$         & 0.402 & 0.107 & 0.107     & 95.2    \\
		& dIVW&  5.45        & 0.401 & 0.095 & 0.094     & 94.9    \\
		& dIVW&  $\sqrt{2\log p}=3.75$   & 0.401 & 0.087 & 0.088     & 95.1    \\
		& dIVW&  MR-EO $\approx 2.80 $      & 0.400 & 0.086 & 0.086    & 95.1    \\&&&&&\\
		& MR-Egger&  $0$ & 0.335	&0.082	&0.082&	87.5 \\
		& MR-Egger& 5.45 & 0.390&	0.240&	0.256&	96.0 \\
		& MR-Egger& $\sqrt{2\log p}=3.75$   & 0.389 & 0.205 & 0.214     & 95.6    \\&&&&&\\
		
		& MR-median& 0  & 0.371&	0.110&	0.122&	96.3    \\
		& MR-median& 5.45   &  0.398&	0.118	&0.128	&96.5   \\
		& MR-median& $\sqrt{2\log p}=3.75$   & 0.397 & 0.113 & 0.124    & 96.7    \\&&&&&\\
		& MR-mode& 0    & 0.033	&74	&75586&	100   \\
		& MR-mode& 5.45     & 0.395&	0.139&	0.151	&97.1   \\
		& MR-mode& $\sqrt{2\log p }=3.75$     & 0.395 & 0.142 &0.157    & 97.2    \\&&&&&\\
		& MR-raps&   0   & 0.401 & 0.105 & 0.105     & 95.2    \\
		& MR-raps-shrink&  Bayes shrinkage & 0.400 & 0.086 & 0.086     & 95.1    \\[0.5ex] \hline\\[-2ex]
		2 & IVW& $0$                  & 0.159 & 0.091 & 0.090    & 23.9    \\
		$s=100$& IVW& 5.45            & 0.397 & 0.206 & 0.206    & 95.1    \\
		$p=1119 $%	 $\kappa=1.05$ 	
		& IVW&   $\sqrt{2\log p}=3.75$   & 0.394 & 0.183 & 0.183    & 94.9    \\&&&&&\\
		& dIVW&$0$       & 0.404 & 0.233 & 0.233     & 95.4    \\
		& dIVW& 5.45            & 0.400 & 0.207 & 0.207    & 95.1    \\
		& dIVW&$\sqrt{2\log p}=3.75$         & 0.400 & 0.186 & 0.186     & 94.9    \\
		& dIVW& MR-EO $\approx 2.21 $  & 0.396 & 0.167 & 0.167     & 95.0    \\&&&&&\\
		& MR-Egger& $0$ & 0.231	&0.122	&0.123&	72.3 \\
		& MR-Egger& 5.45 & 0.388&	0.948&	0.966&	96.2\\
		& MR-Egger& $\sqrt{2\log p}=3.75$   & 0.385&	0.359&	0.385&	95.8    \\&&&&&\\
		
		& MR-median& 0  & 0.276	&0.152	&0.170&	91.3   \\
		& MR-median& 5.45   &  0.396&	0.228&	0.245	&96.6 \\
		& MR-median&$\sqrt{2\log p}=3.75$   & 0.394	&0.216&	0.236&	96.9    \\&&&&&\\
		& MR-mode& 0    &-1.062&	130	&78125 &	100   \\
		& MR-mode& 5.45     & 0.387	&0.267	&0.291	&97.0  \\
		& MR-mode& $\sqrt{2\log p }=3.75$     & 0.390	&0.237	&0.286&	97.1   \\&&&&&\\
		& MR-raps&  0    & 0.398 & 0.224 & 0.226    & 94.9    \\
		& MR-raps-shrink&   Bayes shrinkage     & 0.399 & 0.160 & 0.159    & 95.2    \\[0.5ex] \hline
	\end{tabular}
\end{table}\clearpage
\begin{table}
	\begin{tabular}{llcrrrr} \hline\\[-2ex]
		3 & IVW & $0$             & 0.352 & 0.047 & 0.047    & 82.6    \\
		$s = 1119$& IVW&  5.45     & 0.395 & 0.086 & 0.087     & 95.4    \\
		$p=1119 $% $\kappa =7.78$ 		
		& IVW&  $\sqrt{2\log p}=3.75$       & 0.392 & 0.068 & 0.069     & 95.0    \\&&&&&\\
		& dIVW&  $0$             & 0.400 & 0.054 & 0.054    & 94.7    \\
		& dIVW & 5.45              & 0.399 & 0.087 & 0.088    & 95.4    \\
		& dIVW&  $\sqrt{2\log p}=3.75$   & 0.399 & 0.070 & 0.070     & 95.4    \\
		& dIVW&   MR-EO $\approx 0.03 $  & 0.400 & 0.054 &0.054    & 94.8    \\&&&&&\\
		& MR-Egger& $0$ & 0.372&	0.066&	0.067	&93.1\\
		& MR-Egger& 5.45 & 0.383	&0.189&	0.198&	95.4\\
		& MR-Egger & $\sqrt{2\log p}=3.75$   & 0.372&	0.132&	0.136	&95.0    \\&&&&&\\
		
		& MR-median& 0  & 0.375	&0.079&	0.090&	96.6   \\
		& MR-median& 5.45   & 0.394	&0.114	&0.125&	96.8 \\
		& MR-median& $\sqrt{2\log p}=3.75$   & 0.391	&0.100	&0.111	&96.9   \\&&&&&\\
		& MR-mode&  0    &0.750	&84	&23253&	100  \\
		& MR-mode& 5.45     & 0.391&	0.125&	0.141&	96.8 \\
		& MR-mode& $\sqrt{2\log p }=3.75$     & 0.385&	0.260&	0.504&	97.6  \\&&&&&\\
		& MR-raps&      0    & 0.400 & 0.054 & 0.054    & 94.9    \\
		& MR-raps-shrink&    Bayes shrinkage   & 0.400 & 0.053 & 0.053     & 94.7    \\[0.5ex]\hline
	\end{tabular}
\end{table}

Table \ref{tb:simu} shows  (i) the simulation mean and SD of each estimator, (ii)   average of SEs, which are calculated according to (\ref{evar2}) or (\ref{evar}) for IVW or dIVW estimators, and (iii)  simulation coverage probability (CP) of  95\% confidence intervals from normal approximation. 
The simulation average $ \lambda $ determined by the MR-EO algorithm is also included. 
Under all scenarios, the IVW estimator without screening (i.e., $\lambda =0$) is biased towards zero, which agrees with our theoretical result since the average IV strength $\kappa$'s are relatively small. The  coverage probabilities based on IVW estimators are far from 95\% due to the downward bias and the inaccurate normal approximation. The IVW estimator with screening under the  threshold  $\lambda = 5.45$ 
or $\sqrt{2 \log p}$ does substantially better, which again agrees with our theory that the IVW estimator with screening requires less stringent assumptions for consistency and asymptotic normality. 

The dIVW estimators with and without screening show negligible bias and nominal coverage across all simulation scenarios. This observation agrees with our theoretical assessment that the dIVW estimator requires far less stringent conditions for consistency and asymptotic normality than the IVW estimator. Also, the dIVW estimator with screening improves the dIVW estimator by having a smaller SD in Cases 1-2 where many IVs are null. However, in Case 3 where all IVs are non-null but many are weak, screening in dIVW does not lead to any improvement. In all cases, our MR-EO algorithm adapts to the underlying data and produces the most efficient estimate of $\beta_0$ among dIVW estimators, all without losing coverage or large gains in bias. Finally, all SEs based on (\ref{evar2}) and (\ref{evar}) are close to the simulated SDs of IVW and dIVW estimators, even for the biased IVW estimator without screening. 

The MR-Egger, MR-median and MR-mode estimators without screening are biased when the average IV strength is small. In particular, the MR-mode without screening can be severely biased with unrealistically large SE. MR-Egger, MR-median and MR-mode with screening thresholds at 5.45 or $\sqrt{2\log p}$ generally have larger SDs compared to the dIVW estimators thresholded at the same level. Also, even with thresholding, these three methods (MR-Egger, MR-median and MR-mode) have larger biases than the dIVW estimator because they inherently rely on using the ratio estimator $\hat{\beta}_j$.

The performance of MR-raps is comparable to that of dIVW estimator (\ref{eq: dIVW}). Both methods do not require the independent selection dataset for screening, but the dIVW estimator without screening has a simple explicit form. The MR-raps-shrink uses an independent selection dataset to improve performance and it is comparable to the dIVW estimator with screening and $\lambda$ chosen by MR-EO. However, MR-raps-shrink  is computationally more complicated than dIVW with MR-EO and may not have a unique (or well-defined) solution as mentioned in \cite{Zhao:2019aa}.  

%Although screening is mainly for a better estimation of $\beta_0$, rather than 
%non-null IV selection,
Table \ref{tb:n_IV} presents the total number of IVs selected during screening as well as the number of non-null IVs selected.
In Case 1 where non-null IVs are strong and a good IV selection procedure should perform well,
we see that the number of non-null IVs selected based on  genome-wide significance (p-value $\leq 5\times 10^{-8}$ or $\lambda \approx 5.45$)  is much too  small compared with $s=20$, the true number of non-null IVs.
On the other hand, both $\lambda = \sqrt{2 \log p}$ and MR-EO select close to $s=20 $ non-null IVs. MR-EO selects more null IVs because it aims for efficiency instead of consistent IV selection. In Cases 2-3, there are many weak non-null IVs and all IV selection procedures via thresholding are not adequate (Table  \ref{tb:n_IV}). However, this is not surprising as screening is used in MR to ultimately improve estimation of $\beta_0$, rather than to consistently select non-null IVs. Finally, comparing the results in Tables \ref{tb:simu} and \ref{tb:n_IV} indicates that for the dIVW estimator, removing too many weak IVs may lead to an inefficient estimator of $\beta_0$.

In the Supplementary Material, we conduct a similar simulation study with balanced horizontal pleiotropy added to Case 3. The results are nearly identical to Case 3 without pleiotropy.  Also in the Supplementary Material, we conduct a simulation study where balanced horizontal pleiotropy is violated in Case 3 and we assess the bias of our estimator. We generally find that the dIVW estimator is biased, but the magnitude of the bias is often mild. This is because  the bias term is a weighted average of $ \alpha_j/\gamma_j $'s with weights $ w_jq_{\lambda, j} $'s, where large $ \alpha_j/\gamma_j $ tends to be downweighted by $ w_jq_{\lambda, j} $. 

\begin{table}
	\centering
	\caption{Average number of total IVs and non-null IVs selected from the selection dataset of $p=1119$ IVs under different  thresholds.}\label{tb:n_IV}
	\begin{tabular}{l cccccc } \hline \\[-2ex]
		&     \multicolumn{2}{c}{$\lambda=5.45$}        &  \multicolumn{2}{c}{$\lambda=\sqrt{2\log p}=3.75$}       & \multicolumn{2}{c}{MR-EO}      \\\hline \\[-2ex]
		Case& total & non-null & total & non-null & total & non-null\\
		1, $s = 20$  & 12.8& 12.8        & 18.4& 18.2      & 27.0& 19.8      \\ 
		2, $s = 100$   & 3.9 & 3.9        & 9.2& 9.0       & 63.2& 27.8       \\
		3, $s = 1119$ & 23.8& 23.8      & 84.4& 84.4      &1019.6 &1019.6       \\  \hline
	\end{tabular}
\end{table}

Overall, there are four takeaways from this simulation study.  First,  the dIVW estimator with or without screening always outperforms the IVW estimator. 
Second, without a selection dataset, the dIVW estimator and MR-raps have comparable performances and both are far better than other MR methods under consideration.
% the proposed dIVW estimator has simple explicit form. 
Third, with a selection dataset, the dIVW with screening and MR-raps-shrink have comparable performances and outperform other methods.  Finally, if a selection dataset is available and used to improve efficiency, we suggest the threshold to be  $\lambda = \sqrt{2 \log p}$ or $\lambda$ produced by the MR-EO algorithm, instead of the usual cutoff $\lambda \approx 5.45$. 

%\[
%\frac{\hat{\kappa}_\lambda^*\sqrt{\hat{p}_\lambda^*}}{\max(1,\lambda^2)}
%\]
%\[
%\hat{\beta}_{\lambda,\rmdivw}
%\]

\subsection{Empirical Guidelines for Asymptotics}
In practice, it is important to have some sense of what  is ``a large enough'' sample size for the asymptotic results in the paper to serve as  good approximations. Many researchers in MR have conducted such analysis for the IVW estimator, most notably \cite{Burgess:2013aa}. We conduct a similar simulation-based analysis for the dIVW estimator where we examine what would be a ``large'' effective sample size, as measured by $\kappa_\lambda\sqrt{p_\lambda}/\max (1, \lambda^2)$, for the asymptotics promised by Theorem \ref{theo: dIVW} to be plausible. 

The setting is identical to Case 3 in Section \ref{subsec: simu BMI-CAD}. We choose a grid of 100 equally spaced  $\lambda$'s between $0$ and $10$. For each $\lambda$, we generate 1,000 simulation datasets and calculate the corresponding $\hat{\beta}_{\lambda, \rmdivw}$ for each dataset. Figure \ref{figure: condition} plots these $\hat{\beta}_{\lambda,\rmdivw}$ values against  $\hat{\kappa}_\lambda\sqrt{\hat{p}_\lambda}/\max (1, \lambda^2)$ as well as two standard error bands centered at $\beta_0$ (shaded area). Note that the sample size $n_X$ and the number of IVs $p$ are fixed and, therefore, as $\hat{\kappa}_\lambda\sqrt{\hat{p}_\lambda}/\max (1, \lambda^2)$ grows, the confidence band first shrinks and then becomes relatively stable.

We find that for any $\lambda$, the coverage probability for the dIVW estimator with screening ranges from 93.5\% to 96.0\%. However, as $\hat{\kappa}_\lambda\sqrt{\hat{p}_\lambda}/\max (1, \lambda^2)$ grows larger, we see fewer estimates far from $\beta_0$, an indication that asymptotics have ``kicked in''. This appears to occur when $\hat{\kappa}_\lambda\sqrt{\hat{p}_\lambda}/\max (1, \lambda^2)$  is greater than $20$. Based on this, we recommend that users of dIVW check to make sure that $\hat{\kappa}_\lambda\sqrt{\hat{p}_\lambda}/\max (1, \lambda^2)$ is at least greater than $20$ as part of a diagnostic check for the dIVW estimator.

\begin{figure}[h]
	\includegraphics[scale=0.4]{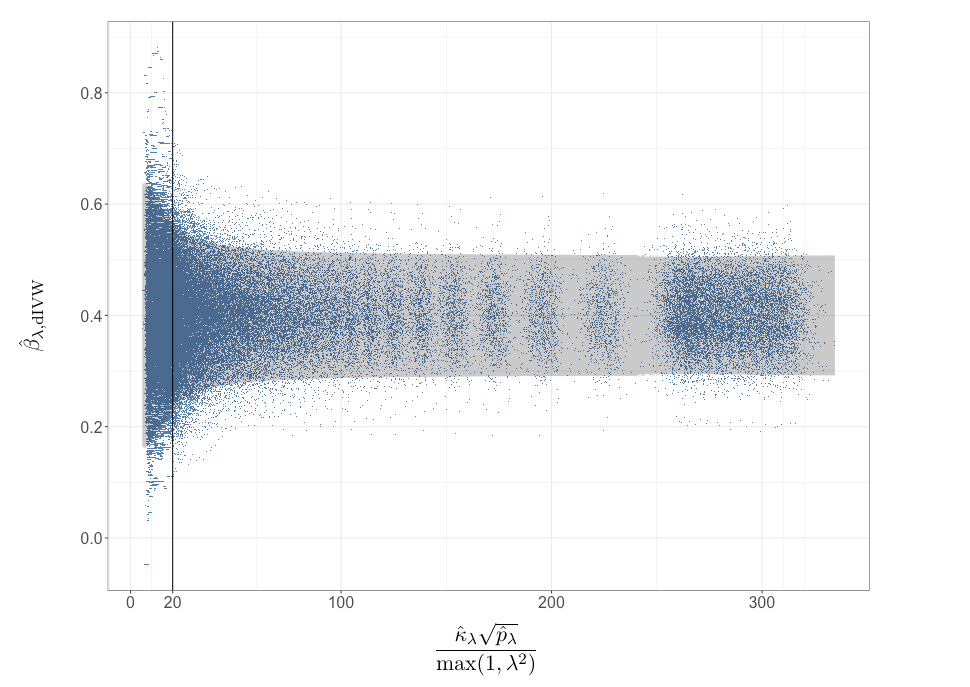}
	\caption[]{Evaluation of the consistency and asymptotic normality condition for the dIVW estimator in Theorem \ref{theo: dIVW} under Case 3. The x-axis plots the condition that governs the asymptotic rate of the dIVW estimator. The y-axis plots values of the dIVW estimator in 1000 simulations. The shaded area represents two-standard error bands centered at $\beta_0$.}
	\label{figure: condition}
\end{figure}

\subsection{Empirical Evaluation of the Effect of Using SEs not SDs} \label{sec:sim_artificial}

In this section, we evaluate the finite sample performance of the proposed dIVW estimators when we don't assume that the SDs of  summary statistics are known and instead, we use the estimated SEs, $ \hat{\sigma}_{Xj}, \hat{\sigma}_{Yj}$, $ \hat{\sigma}_{Xj}^* $'s. We construct a population with parameters 
$\gamma_j=\varphi_j \sqrt{2h^2/s}$ for $ j=1, ..., s $, and $ \gamma_j=0 $ for $ j=s+1, ..., p $, where $s$ is the number of non-null SNPs, $ h^2 $ is the total heritability, i.e., the proportion of  variance in $ X $ that is attributable to the $ s $ non-null SNPs \citep{Visscher:2008aa}, 
and $\varphi_j$'s are constants that are generated once from a standard normal distribution.

To simulate individual-level data, 
we first generate $ p $ independent SNPs, $ Z_1, ..., Z_p $,  from a multinomial distribution with $ P(Z_j=0)=0.25, P(Z_j=1)=0.5, P(Z_j=2)=0.25$, and then generate the exposure variable $ X $ and outcome  variable $ Y $ according to models (\ref{eq: exposure-IV})-(\ref{eq: outcome-exposure}) with $ \eta_X=\eta_Y=1 $, $ \beta_0=0.4 $, $ U\sim N(0, 0.6 (1-h^2)) $, and $ E_X, E_Y\sim N(0, 0.4 (1-h^2)) $.  For each simulation repetition, we generate three independent  datasets of size $n$ that represent the selection, exposure, and outcome datasets. The summary statistics $ \{\hat{\gamma}_j, \hat{\sigma}_{Xj} , j=1, ..., p\}$, $ \{\hat{\Gamma}_j, \hat{\sigma}_{Yj}, j=1,..., p\}$, and $ \{\hat{\gamma}^*_j, \hat{\sigma}^*_{Xj} , j=1, ..., p\}$  are obtained from the three datasets
through marginal linear regression. 

We consider the following combinations of $n$, $p$, $s$, and $h^2$ that reflect what may be found in practice.
\begin{itemize}
	\item[ \bf Case 4] $n= 10,000, p=2,000, s=200,  h^2= 0.1$,  $\kappa= 0.50, \kappa \sqrt{p}=21.41.$
	\item[{\bf Case 5}] $n= 10,000, p=2,000, s=1,000,  h^2= 0.2$, $\kappa= 0.93, \kappa \sqrt{p}=41.77.$
	\item[{\bf Case 6}] $n= 50,000, p=2,000, s=1,000,  h^2= 0.2$, $\kappa= 4.67, \kappa \sqrt{p}=208.84$.
	\item[{\bf Case 7}] $n= 10,000, p=2,000, s=2,000,  h^2= 0.2$, $\kappa= 0.96, \kappa \sqrt{p}=43.10.$
\end{itemize}
We consider $n=10,000$ to be a conservative sample size in modern MR studies and is much smaller than the sample sizes in the BMI-CAD dataset in Section \ref{subsec: simu BMI-CAD}.
Also, $s/p$ takes on values 10\%, 50\% and 100\%. 

Table \ref{tb:individual} presents the mean, SD, SE, and CP of IVW, dIVW, MR-raps, and MR-raps-shrink estimators, based on 10,000 replications. We omit MR-Egger, MR-median, and MR-mode for conciseness.  Overall, a similar trend appears in relation to Table \ref{tb:simu} for the case of assuming SDs $\sigma_{Xj}$, $\sigma_{Yj}$, $\sigma_{Xj}^*$  are known: the IVW estimator without screening ($\lambda =0$) is inconsistent and biased towards 0; the dIVW estimator maintains nominal coverage and outperforms the IVW estimator;  and the dIVW estimator with screening by MR-EO performs similarly with the dIVW without screening when $ s/p $ is not  small. The SEs  are generally close to the simulated SDs of point estimators, including the case where the MR-EO is applied.  The only exception is in Case 7 when $ \lambda=3.90 $. We believe this is because the consistency and asymptotic normality condition, specifically the ``effective sample size'' value $ {\kappa}_\lambda \sqrt{ p_{\lambda}} /\max (1, \lambda^2)$ is $1.11$ when $ \lambda=3.90 $ and as Figure \ref{figure: condition} illustrates, this value is too small for our asymptotic theory to kick in. Also, upon closer inspection of the numerical results in this case, there are two outlier estimates above 40 or below -20 (out of 10,000 simulation runs). Removing these two simulation runs leads to SD= 0.304 and SE=0.293, which agree more closely with each other. Overall, this observation indicates the importance of performing  diagnostic check for the dIVW estimator using $ \hat{\kappa}_\lambda \sqrt{\hat p_{\lambda}} /\max (1, \lambda^2)$ and using the MR-EO algorithm to adaptively select $ \lambda $ when needed.

We also notice the following observations that were not in Table \ref{tb:simu}. First, 
the performance of all estimators tend to improve when $n$ increases (Cases 5-6) even though $s/p =50\%$. Second, the use of genome-wide significance threshold $\lambda \approx 5.45$  often selects no SNPs in many simulation runs when $s=1,000$ or  $2,000$,  another indication that this threshold is too large in MR studies with many weak IVs.

In the Supplementary Material, we also run the same simulations using  $ \hat\sigma_{Xj} = \sigma_{Xj}, \hat\sigma_{Yj}= \sigma_{Yj}, \hat\sigma_{Xj}^* = \sigma^*_{Xj} $ and obtain almost identical results as Table \ref{tb:individual}; see  Table S3 of the Supplementary Material. This indicates that the effect of assuming known SDs  and using them as SEs is negligible, which agrees with many empirical results in the literature as well as our theoretical results in Theorem \ref{theo: dIVW}.  

%An observation which is also in Table 1 is that when $s=p$, the dIVW estimator without screening is the same as the dIVW estimator with screening by MR-EO. 
%Two phenomenons in Table 3 but not  in Table 1 are (i) 

%Finally,  the use of genome-wide significance threshold $\lambda \approx 5.45$  selects no IV over a large portion of simulation runs when $s=1,000$ or  $2,000$,  another indication that this threshold is too large in MR studies with many weak IVs. }

\begin{table}
	\caption{Simulation results for Cases 4-7 based on 10,000 repetitions with $\beta_0 =0.4$;  $ \lambda $ for  MR-EO is the simulation average;
		SD is the simulation standard deviation;  SE  is  the average of standard errors; CP is the simulation coverage probability of the 95\% confidence interval based on  normal approximation. }\label{tb:individual}
	\centering
	\begin{tabular}{llcrrrr} \hline  \\[-2ex]
		\multicolumn{1}{c}{Case}                & \multicolumn{1}{c}{Method} & \multicolumn{1}{c}{$\lambda$}       & \multicolumn{1}{c}{mean} & \multicolumn{1}{c}{SD}    & \multicolumn{1}{c}{SE} & \multicolumn{1}{c}{CP}  \\ \hline  \\[-2ex]
		4      & IVW & $0$                   & 0.129 & 0.027 & 0.027     & 0    \\
		$s=200$ & IVW&  5.45                   & 0.393 & 0.125 & 0.122     & 94.6    \\
		$p=2000$			& IVW&  $\sqrt{2\log p}=3.90$       & 0.382 & 0.075 & 0.074     & 93.8    \\ 
		$n=10000$	&&&&&\\
		& dIVW&  $0$         & 0.402 & 0.090 & 0.089    & 95.0    \\
		& dIVW&  5.45        & 0.406 & 0.131 & 0.127     & 95.0    \\
		& dIVW&  $\sqrt{2\log p}=3.90$   & 0.402 & 0.079 & 0.078     & 95.0    \\
		& dIVW&  MR-EO $\approx 2.17$      & 0.396 & 0.061 & 0.060     & 94.9
		\\&&&&&\\
		& MR-raps&   0   & 0.401 & 0.085 & 0.085    & 94.8    \\
		& MR-raps-shrink&  Bayes shrinkage & 0.399 & 0.062 & 0.062     & 95.2    \\[0.5ex] \hline  \\[-2ex]
		5 & IVW& $0$                  & 0.193 & 0.024 & 0.023     & 0    \\
		%($p=1119$)& IVW, $\calS$          & 0.366 & 0.136 & 0.536     & 94.20    \\
		$s=1000$& IVW& 5.45            & \multicolumn{4}{l}{ select no IV over 25\% of runs}   \\
		$p=2000$		& IVW&   $\sqrt{2\log p}=3.90$   & 0.364 & 0.099 & 0.099     & 92.8    \\
		$n=10000$ &&&&&\\
		& dIVW&$0$       & 0.401 & 0.051 & 0.051     & 94.7    \\
		& dIVW& 5.45            &     \multicolumn{4}{l}{ select no IV over 25\% of runs} \\
		& dIVW&$\sqrt{2\log p}=3.90 $         & 0.405& 0.112 & 0.111     & 95.3    \\
		& dIVW&MR-EO $\approx 1.10$  & 0.394 & 0.048 & 0.048     & 94.6    \\&&&&&\\
		& MR-raps&   0   & 0.400 & 0.049 & 0.049     & 94.5    \\
		& MR-raps-shrink&   Bayes shrinkage     & 0.399 & 0.047 & 0.047     & 94.8    \\ [0.5ex]\hline  \\[-2ex]
		6     & IVW & $0$                   & 0.330 & 0.014 & 0.014     & 0.1    \\
		$s=1000$ & IVW&  5.45                   & 0.390 & 0.024 & 0.024     & 93.0    \\
		$p=2000$			& IVW&  $\sqrt{2\log p}=3.90$       & 0.386 & 0.019 & 0.018    & 88.1    \\ 
		$n=50000$	&&&&&\\
		& dIVW&  $0$         & 0.400 & 0.017 & 0.017    & 94.8    \\
		& dIVW&  5.45        & 0.400 & 0.024 & 0.024     & 94.8    \\
		& dIVW&  $\sqrt{2\log p}=3.90$   & 0.400 & 0.019 & 0.019     & 94.5    \\
		& dIVW&  MR-EO $\approx 1.31$      & 0.399 & 0.017 & 0.017     & 94.9
		\\&&&&&\\
		& MR-raps&   0   & 0.400 & 0.017 & 0.017    & 94.8    \\
		& MR-raps-shrink&  Bayes shrinkage & 0.400 & 0.017 & 0.017     & 94.8    \\[0.5ex]\hline  \\[-2ex]
		7 & IVW& $0$                  & 0.196 & 0.023 & 0.023     & 0    \\
		$s=2000$& IVW& 5.45            & \multicolumn{4}{l}{ select no IV over 81\% of runs}   \\
		$p=2000$		& IVW&   $\sqrt{2\log p}=3.90$   & 0.343 & 0.197 & 0.195     & 93.5    \\
		$n=10000$ &&&&&\\
		& dIVW&$0$       & 0.400 & 0.050 & 0.049     & 94.5    \\
		& dIVW& 5.45            &     \multicolumn{4}{l}{ select no IV over 81\% of runs} \\
		& dIVW&$\sqrt{2\log p}=3.90 $         & 0.423& 0.622 & 0.428     & 96.5    \\
		& dIVW&MR-EO $\approx 0.44$  & 0.395 & 0.051 & 0.050     & 94.4    \\&&&&&\\
		& MR-raps&    0  & 0.400 & 0.048 & 0.047     & 94.8    \\
		& MR-raps-shrink&     Bayes shrinkage   & 0.399 & 0.048 & 0.047     & 94.8    \\ [0.5ex]\hline
	\end{tabular}
\end{table}

\section{Real Data Example}
\label{sec: real}
We apply our methods to the BMI-CAD example described in Section \ref{subsec: simu BMI-CAD}. Table \ref{tb: real} summarizes the results, where ${\rm dIVW}_\alpha$ denotes the dIVW estimator developed under balanced horizontal pleiotropy, MR-${\rm raps}_\alpha$ and MR-raps-${\rm shrink}_\alpha$ are MR-raps estimators that account for balanced  horizontal  pleiotropy by setting the over.dispersion parameter in the \emph{mr.raps} R package to be \emph{TRUE}. 

We make the following comments. First, we see that the MR-mode estimator with or without screening is very unstable. 
Second, in light of our simulation result under Case 3, we suspect that the IVW estimate 0.315 without screening ($\lambda = 0$)  is slightly biased towards zero, compared with the dIVW estimate 0.365, although the difference is not statistically significant. Third, except for MR-Egger and MR-mode, selecting IVs based on genome-wide significance (i.e., $\lambda=5.45$)
produces point estimates between 0.278 and 0.287 and larger SEs across all methods, most likely because too many IVs are screened out.  Fourth, except for the IVW estimator without screening, the dIVW estimator with MR-EO achieves the smallest SE among all estimates. But, since the dIVW estimate based on MR-EO has the same SE as the dIVW estimate without screening, screening is probably not necessary for this dataset with many weak IVs; this is also supported by the fact that there is not much difference between MR-raps and MR-raps-shrink. Fifth, the estimators accounting for balanced horizontal pleiotropy are similar to those without it, except for an expected increase in SEs due to the random effect terms.

%Comparing thresholds, we generally recommend using $\lambda=\sqrt{2\log p}$ for all  methods that are not sensitive to weak IVs because this threshold adapts to the dimension $p$ and includes more relevant IVs than a static threshold of $\lambda=5.45$ based on the genome-wide significance level. 

%for $0 \leq \lambda \leq \sqrt{2\log p}$, the results are not sensitive to the choice of $\lambda$. 

\begin{table}[t]
	\centering
	\caption{Point estimates of exposure effect and their SEs (in parentheses) from different MR methods in the BMI-CAD example.} %Each column represents results under a $\lambda$ threshold for screening. %The top block summarizes how many IVs were selected for a given $\lambda$ threshold along with calculations of effective sample sizes $\hat{\kappa}_\lambda^*$ and $\hat{\kappa}^*_\lambda\sqrt{\hat{p}^*_\lambda}/\max (1, \lambda^2)$ that govern the asymptotic properties of the pre-screened IVW and dIVW estimators. 
	%	MR-EO selects  slightly different $\lambda$ depending on whether the dIVW or $\mbox{dIVW}_\alpha$ are used; the left-hand side of the semicolon refers to the dIVW estimator and the right-hand size of the semicolon refers to the $\mbox{dIVW}_\alpha$ estimator.}
	
	%The bottom block represents each point estimate in log odd ratios. Estimated standard errors are in parentheses. 
	%The subscript $\alpha$ denotes methods tailored to handle balanced horizontal pleiotropy.
	\label{tb: real}
	\begin{tabular}{lccccc} \hline
		&       &   &   & {\tiny MR-EO} & {\tiny MR-EO$_\alpha$} \\ \cline{5-6}  \\[-2ex]
		$\lambda$        & 0      & 5.45   & $\sqrt{2\log p}\! =\! 3.75$   & 0.57 & 0.59 \\  \hline   \\[-2ex]
		{\# of IVs selected }
		& 1119          & 44            & 165           & 1029 & 1023       \\
		%		 $\hat{\kappa}_\lambda$               & 6.8           & 72.9          & 28.1          & 7.2 & 7.3        \\
		$\hat{\kappa}_\lambda\sqrt{\hat{p}_\lambda} \! /\! \max (1, \! \lambda^2)$ & 226.8 & 16.3 &25.7 &232.4&  233.1\\ \hline  \\[-2ex]
		IVW                      & 0.315 (0.050) & 0.282 (0.084) & 0.319 (0.068) &               \\
		dIVW                     & 0.365 (0.058) & 0.287 (0.085) & 0.331 (0.071) & 
		\multicolumn{2}{c}{0.345  (0.058)} \\
		$\mbox{dIVW}_\alpha$ & 0.365 (0.067) & 0.287 (0.100) &0.331 (0.082) & \multicolumn{2}{c}{0.345 (0.067)}\\
		MR-Egger                 & 0.386 (0.077) & 0.513 (0.184) & 0.390 (0.129)                \\
		MR-median                & 0.322 (0.097) & 0.278 (0.124) & 0.304 (0.116) &               \\
		MR-mode                  & 0.739 (402.9) & 0.499 (0.402) & 0.488 (4.241) &               \\  MR-raps                    & 0.382 (0.061) &   \\ 
		%0.291 (0.086) & 0.339 (0.072) &    
		$\mbox{MR-raps}_\alpha$ 	   &  0.367 (0.067)    & \\
		\multicolumn{2}{l}{  MR-raps-shrink (Bayes shrinkage) }
		& 0.388  (0.060) \\
		\multicolumn{2}{l}{$\mbox{MR-raps-shrink}_\alpha$(Bayes shrinkage)} & 0.374  (0.067)   \\[0.5ex]
		%0.297 (0.120)&0.337 (0.090)
		%0.292 (0.086) & 0.341 (0.072) &   
		%0.298 (0.120)& 0.339 (0.090)
		\hline  \\[-2ex]
		\multicolumn{6}{l}{The subscript $\alpha$ indicates application under balanced horizontal pleiotropy }
		%	 &&&&\\
		%	 ${\rm IVW}_\alpha$   & 0.315 (0.058) & 0.282 (0.096) & 0.319 (0.079) & \\              \hline
	\end{tabular}
\end{table}

Following \cite{zhao2018statistical}, 
we run a diagnostic to assess the plausibility of Assumption \ref{assump: 2}
by constructing a Quantile-Quantile  plot of the standardized residuals, 
$
(\hat{\Gamma}_j-\hat{\beta}_{ \rmdivw}\hat{\gamma}_j)/
(\hat\sigma_{Yj}^2+\hat{\beta}_{ \rmdivw}^2\hat\sigma_{Xj}^2)^{1/2}$, $j=1,...,p$. 
Figure \ref{figure: qq} shows the result. Since the residuals line up close to the 45-degree line, Assumption \ref{assump: 2} is likely to hold for this example.
In the Supplementary Material, a similar figure is obtained for assessing the plausibility of Assumption $2'$.

\begin{figure}[h]
	\includegraphics[scale=0.7]{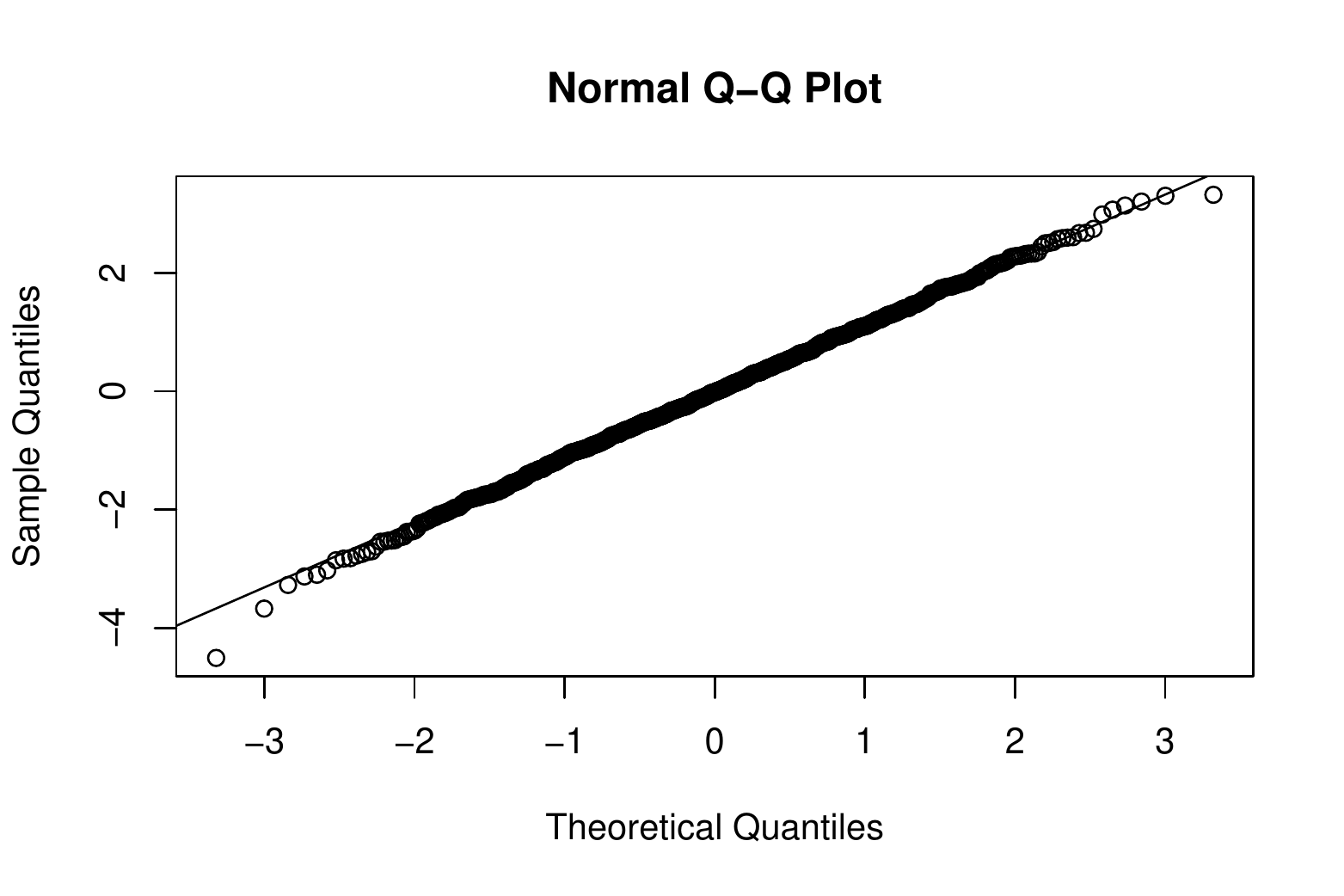}
	\caption[]{Quantile-Quantile plot of the standardized residuals  against a standard normal.}
	\label{figure: qq}
\end{figure}

\section{Summary and Discussion}
\label{sec: discussion} 
In  two-sample summary-data MR studies, we show that the IVW estimator requires stringent conditions on the average strength of IVs for consistency and asymptotic normality. The IVW estimator with screening relaxes these conditions somewhat, but requires carefully choosing a threshold $\lambda$ and a third independent dataset. We then propose a simple modification of the IVW estimator,  
called the debiased IVW (dIVW) estimator. The dIVW estimator, with or without screening, is shown to be consistent and asymptotically normal under conditions that are much weaker than those required by the IVW estimator, with or without screening. Finally, we provide some theoretical and numerical results on assuming the commonly invoked known-variance condition.
%that the common pre-screening practice of selecting strong instruments based on a p-value cutoff at $5\times 10^{-8}$ may eliminate majority of relevant IVs. In order to increase estimation efficiency of the dIVW estimator, we then propose two approaches of choosing SNPs via $\lambda$, one based on $\lambda=\sqrt{2\log p}$ that effectively eliminates null IVs and the other that uses the MR-EO algorithm to directly maximize efficiency.

While our work primarily focuses on the ``standard'' IVW estimator commonly used in practice, as suggested by the anonymous referees and the editor, the standard IVW estimator, with or without screening, can be viewed as instances of the generalized IVW estimator
\begin{align*}
	\frac{	\sum_{j=1}^{p} \hat\beta_j f(\hat w_j)}{\sum_{j=1}^p f(\hat w_j )}, \
\end{align*}
where $ \hat \beta_j = \hat\Gamma_j /\hat\gamma_j$, and $ f $ is a general weighting function. This general weighting function can encompass soft thresholding and other IV selection procedures. However, this class of estimators does not include the proposed dIVW estimator since the weights from the dIVW estimator will not sum to 1. Nevertheless, extending the current theory to better understand this broader class of IVW estimators under many weak IVs is an important direction for future research.

Finally, based on our theoretical and simulation work, we make three recommendations for practice. % for two-sample summary-data Mendelian randomization studies. 
First, we argue that the dIVW estimator without screening should be the default baseline estimator for two-sample summary-data  MR studies instead of the IVW estimator. It is as simple as the IVW estimator, %requiring one to multiply the IVW estimator with a bias-correction factor, 
and has provable robustness against many weak instruments and balanced horizontal pleiotropy, lending itself as the baseline estimator for investigating more complex pleiotropy.
%If an independent selection dataset is not available, including every SNP in dIVW can still lead to a consistent and asymptotically normal dIVW estimator. 
Second,   if there are many irrelevant IVs and summary statistics from a third independent selection dataset are available, we may improve the efficiency of the dIVW estimator by screening with threshold  $\lambda$ produced by the MR-EO algorithm; we discourage the use of the genome-wide significance p-value threshold $\lambda \approx 5.45$  as it tends to screen out too many IVs. Third, for the promised theoretical properties of the proposed dIVW estimator to hold, it is important to perform diagnostics by constructing Quantile-Quantile  plot of the standardized residuals and also checking that $ \hat\kappa_{\lambda}\sqrt{\hat p_\lambda} /\max (1,\lambda^2)$ is at least greater than 20.

%%%%%%%%%%%%%%%%%%%%%%%%%%%%%%%%%%%%%%%%%%%%%%
%% Support information (funding), if any,   %%
%% should be provided in the                %%
%% Acknowledgements section.                %%
%%%%%%%%%%%%%%%%%%%%%%%%%%%%%%%%%%%%%%%%%%%%%%
\section*{Acknowledgements}
The authors would like to thank the anonymous referees, an Associate
Editor and the Editor for their constructive comments that led to a much improved  paper.

The research of Jun Shao was partially supported by the National Natural Science Foundation of China Grant 11831008 and the U.S. National Science Foundation Grant DMS-1914411.

The research of Hyunseung Kang was supported in part by the U.S. National Science Foundation Grant DMS-1811414.

%
%The first author was supported by NSF Grant DMS-??-??????.
%
%The second author was supported in part by NIH Grant ???????????.

\section*{Software and Reproducibility} R code for the methods proposed in this paper can be found in the R package mr.divw, which is posed at \url{https://github.com/tye27/mr.divw}. Numerical examples in this article can be reproduced by running examples in the R package.

%%%%%%%%%%%%%%%%%%%%%%%%%%%%%%%%%%%%%%%%%%%%%%
%% Supplementary Material, if any, should   %%
%% be provided in {supplement} environment  %%
%% with title and short description.        %%
%%%%%%%%%%%%%%%%%%%%%%%%%%%%%%%%%%%%%%%%%%%%%%
\begin{supplement}
\textbf{Supplementary Material: Debiased Inverse-Variance Weighted Estimator in Two-Sample Summary-Data Mendelian Randomization}.
We provide additional numerical results and theoretical proofs for the theorems in the paper.
\end{supplement}
%\begin{supplement}
%\textbf{Title of Supplement B}.
%Short description of Supplement B.
%\end{supplement}

%%  MR numbers will be added by VTeX.                      %%
%%                                                         %%
%%  Use \cite{...} to cite references in text.             %%
%%                                                         %%
%%%%%%%%%%%%%%%%%%%%%%%%%%%%%%%%%%%%%%%%%%%%%%%%%%%%%%%%%%%%%

%% if your bibliography is in bibtex format, uncomment commands:
\bibliographystyle{imsart-nameyear}% Style BST file (imsart-number.bst or imsart-nameyear.bst)
\bibliography{reference}       % Bibliography file (usually '*.bib')

%% or include bibliography directly:
%\begin{thebibliography}{4}
%%%
%\bibitem{r1}
%\textsc{Billingsley, P.} (1999). \textit{Convergence of
%Probability Measures}, 2nd ed.
%Wiley, New York.
%
%\bibitem{r2}
%\textsc{Bourbaki, N.}  (1966). \textit{General Topology}  \textbf{1}.
%Addison--Wesley, Reading, MA.
%
%\bibitem{r3}
%\textsc{Ethier, S. N.} and \textsc{Kurtz, T. G.} (1985).
%\textit{Markov Processes: Characterization and Convergence}.
%Wiley, New York.
%
%\bibitem{r4}
%\textsc{Prokhorov, Yu.} (1956).
%Convergence of random processes and limit theorems in probability
%theory. \textit{Theory  Probab.  Appl.}
%\textbf{1} 157--214.
%\end{thebibliography}

\clearpage
\setcounter{equation}{0}
\setcounter{figure}{1}
\setcounter{table}{0}
\setcounter{page}{1}
\setcounter{section}{0}
\makeatletter
\renewcommand{\theequation}{S\arabic{equation}}
\renewcommand{\thefigure}{S\arabic{figure}}
\renewcommand{\thetable}{S\arabic{table}}

\begin{center}
	\huge Supplementary Material: Debiased Inverse-Variance Weighted Estimator in Two-Sample Summary-Data Mendelian Randomization
\end{center}

\section{Additional Simulation and Real Data Results} 
\subsection{Comparison of Estimators under Balanced Horizontal Pleiotropy}\label{app: simu}
We conduct an additional simulation study when there is balanced horizontal pleiotropy. The simulation setting is identical to Case 3, except now $\Gamma_j=\beta_0\gamma_j+\alpha_j$ for $j=1,\dots,p$, where $\alpha_j\sim N(0,\tau_0^2)$ with $\tau_0=2 p^{-1}\sum_{j=1}^p \sigma_{Yj}$. $\mbox{dIVW}_\alpha$ is to denote the dIVW estimator developed in Section \ref{subsec: pleiotropy} under Assumption $ 2' $. MR-$\mbox{raps}_\alpha$ and MR-raps-$\mbox{shrink}_\alpha$ are implemented via the \emph{mr.raps} package setting the over.dispersion parameter to be \emph{TRUE}.

From Table \ref{app: tb}, we observe that MR-Egger, MR-median, MR-mode with different choices of $\lambda$ all have undesirable performance. For MR-Egger, it shows large bias and wide confidence interval. For MR-median, the constructed confidence interval does not have correct coverage probability. For MR-mode, it is sensitive to many weak IVs if without thresholding, it has large bias and incorrect coverage probability if with thresholding. Overall, the dIVW estimators and MR-raps have good statistical properties under balanced horizontal pleiotropy.

\begin{table}[h]
	\caption{Simulation results for Case 3 under balanced horizontal pleiotropy based on 10,000 repetitions with $\beta_0 =0.4$;  $ \lambda $ for MR-EO is the simulation average;
		SD is simulation standard deviation;  SE  is  the average of standard errors; CP is the simulation coverage probability of the 95\% confidence interval based on  normal approximation. }\label{app: tb}
	\centering
	\begin{tabular}{llcrrrr} \hline\\[-2ex]
		\multicolumn{1}{c}{Case}                & \multicolumn{1}{c}{Method} & \multicolumn{1}{c}{$\lambda$}       & \multicolumn{1}{c}{mean} & \multicolumn{1}{c}{SD}    & \multicolumn{1}{c}{SE} & \multicolumn{1}{c}{CP}  \\ \hline \\[-2ex]
		$s = 1119$	& $\text{dIVW}_\alpha$ & $0$             & 0.401&	0.139&	0.138	&95.0    \\
		%& dIVW, $\calS$         & 0.400 & 0.054 & 0.210     & 94.70    \\
		$p=1119$ 	& $\text{dIVW}_\alpha$&5.45              & 0.398&	0.223&	0.224&	95.0    \\
		& $\text{dIVW}_\alpha$& $\sqrt{2\log p}=3.75$   & 0.399&	0.180&	0.182&	95.2   \\
		& $\text{dIVW}_\alpha$& MR-EO$ \approx 0.03 $ & 0.402&	0.139&	0.138&	94.9  \\&&&&&&\\
		& MR-Egger & $0$ & 0.372	&0.171&	0.163&	93.5\\
		& MR-Egger & 5.45 & 0.378&	0.495	&0.448	&91.3\\
		& MR-Egger& $\sqrt{2\log p}=3.75$ & 0.369	&0.346&	0.327&	93.0   \\&&&&&&\\
		
		& MR-median &0  & 0.372&	0.206	&0.115&	72.1  \\
		& MR-median & 5.45   & 0.392	&0.298	&0.151&	68.0\\
		& MR-median & $\sqrt{2\log p}=3.75$  & 0.388	&0.260&	0.138	&70.4  \\&&&&&&\\
		& MR-mode & 0    &4.783&	376	&59530&	100  \\
		& MR-mode & 5.45     & 0.373	&0.329	&0.171&	63.8 \\
		& MR-mode & $\sqrt{2\log p}=3.75$     & 0.370&	0.587&	0.913	&72.8  \\&&&&&&\\
		& MR-$\text{raps}_\alpha$ &     0      & 0.402	&0.135	&0.135&	94.9    \\
		& MR-raps-$\text{shrink}_\alpha$ & Bayes shrinkage     & 0.401	&0.134	&0.134	&95.0   \\[0.5ex] \hline 
	\end{tabular}
\end{table}

\subsection{Empirical Evaluation of the Effect of Using SEs not SDs} 
In this section, we include simulation results of the MR estimators using the true SDs $ (\sigma_{Xj}, \sigma_{Yj}, \sigma_{Xj}^*)$ under Cases 4-7, where the true SDs can be calculated via
\begin{align*}
	\sigma_{Xj}^2= \sigma_{Xj}^{*2}= \frac{\var(X)- \gamma_j^2\var(Z_j)}{n_X\var (Z_j)}, \qquad \sigma_{Yj}^2= \frac{\var(Y)- \beta_0^2\gamma_j^2\var(Z_j)}{n_Y\var (Z_j)}.
\end{align*}

Simulation results are in Table \ref{apptb:individual}. By comparing the ``feasible'' dIVW estimators using $ \hat{\sigma}_{Xj}, \hat{\sigma}_{Xj}^*, \hat{\sigma}_{Yj} $ (Table \ref{tb:individual}) with the ``infeasible'' dIVW estimators using $ \sigma_{Xj}, \sigma_{Xj}^*, \sigma_{Yj} $ (Table \ref{apptb:individual}), we find that the point estimators and the standard deviations are all very similar, supporting our theoretical results in Theorem \ref{theo: dIVW} that the ''feasible'' and the ''infeasible'' dIVW estimators are asymptotically equivalent when $ p/n_X\rightarrow 0 $.

\begin{table}
	\setcounter{table}{2}
	\renewcommand{\thetable}{S\arabic{table}}
	\caption{Simulation results of MR estimators using true SDs for Case 4-Case 7 based on 10,000 repetitions with $\beta_0 =0.4$;  $ \lambda $ for the MR-EO is the average $ \lambda $ determined by the algorithm;
		SD is simulation standard deviation;  SE  is  the average of standard errors; CP is the simulation coverage probability of the 95\% confidence interval based on  normal approximation. }\label{apptb:individual}
	\centering
	\begin{tabular}{llcrrrr} \hline\\[-2ex]
		\multicolumn{1}{c}{Case}                & \multicolumn{1}{c}{Method} & \multicolumn{1}{c}{$\lambda$}       & \multicolumn{1}{c}{mean} & \multicolumn{1}{c}{SD}    & \multicolumn{1}{c}{SE} & \multicolumn{1}{c}{CP}  \\ \hline\\[-2ex]
		4      & IVW & $0$                   & 0.129 & 0.027 & 0.027     & 0    \\
		$s=200$ & IVW&  5.45                   & 0.393 & 0.126 & 0.122     & 94.7    \\
		$p=2000$			& IVW&  $\sqrt{2\log p}=3.90$       & 0.382 & 0.075 & 0.074     & 93.6    \\ 
		$n=10000$	&&&&&\\
		& dIVW&  $0$         & 0.403 & 0.092 & 0.089    & 94.6    \\
		& dIVW&  5.45        & 0.406 & 0.131 & 0.127     & 95.1    \\
		& dIVW&  $\sqrt{2\log p}=3.90$   & 0.402 & 0.079 & 0.078     & 94.9    \\
		& dIVW&  MR-EO $\approx 2.17$      & 0.396 & 0.061 & 0.060     & 94.8
		\\&&&&&\\
		& MR-raps&     0 & 0.401 & 0.088 & 0.085    & 94.2    \\
		& MR-raps-shrink& Bayes shrinkage & 0.399 & 0.063 & 0.062     & 95.0    \\[0.5ex] \hline\\[-2ex]
		5 & IVW& $0$                  & 0.193 & 0.024 & 0.023     & 0    \\
		%($p=1119$)& IVW, $\calS$          & 0.366 & 0.136 & 0.536     & 94.20    \\
		$s=1000$& IVW& 5.45            & \multicolumn{4}{l}{ \ select no IV over 25\% of runs}   \\
		$p=2000$		& IVW&   $\sqrt{2\log p}=3.90$   & 0.364 & 0.099 & 0.099     & 92.8    \\
		$n=10000$ &&&&&\\
		& dIVW&$0$       & 0.401 & 0.052 & 0.051     & 94.1    \\
		& dIVW& 5.45            &     \multicolumn{4}{l}{ \ select no IV over 25\% of runs} \\
		& dIVW&$\sqrt{2\log p}=3.90 $         & 0.404& 0.112 & 0.111     & 95.2    \\
		& dIVW&MR-EO $\approx 1.10$  & 0.394 & 0.048 & 0.048     & 94.6    \\&&&&&\\
		& MR-raps&   0   & 0.401 & 0.051 & 0.049     & 93.6    \\
		& MR-raps-shrink&  Bayes shrinkage     & 0.400 & 0.048 & 0.047     & 93.9    \\[0.5ex] \hline\\[-2ex]
		6     & IVW & $0$                   & 0.330 & 0.014 & 0.014     & 0.1    \\
		$s=1000$ & IVW&  5.45                   & 0.390 & 0.024 & 0.024     & 93.1    \\
		$p=2000$			& IVW&  $\sqrt{2\log p}=3.90$       & 0.386 & 0.019 & 0.018    & 88.1    \\ 
		$n=50000$	&&&&&\\
		& dIVW&  $0$         & 0.400 & 0.017 & 0.017    & 94.7    \\
		& dIVW&  5.45        & 0.400 & 0.024 & 0.024     & 94.9    \\
		& dIVW&  $\sqrt{2\log p}=3.90$   & 0.400 & 0.019 & 0.019     & 94.5    \\
		& dIVW&  MR-EO $\approx 1.31$      & 0.399 & 0.017 & 0.017     & 94.7
		\\&&&&&\\
		& MR-raps&      0& 0.400 & 0.017 & 0.017    & 94.6    \\
		& MR-raps-shrink&  Bayes shrinkage& 0.400 & 0.017 & 0.017     & 94.6    \\ [0.5ex]\hline\\[-2ex]
		7 & IVW& $0$                  & 0.196 & 0.023 & 0.023     & 0    \\
		$s=2000$& IVW& 5.45            & \multicolumn{4}{l}{ \ select no IV over 81\% of runs}   \\
		$p=2000$		& IVW&   $\sqrt{2\log p}=3.90$   & 0.344 & 0.198 & 0.196     & 93.6    \\
		$n=10000$ &&&&&\\
		& dIVW&$0$       & 0.401 & 0.051 & 0.049     & 94.0    \\
		& dIVW& 5.45            &     \multicolumn{4}{l}{ \ select no IV over 81\% of runs} \\
		& dIVW&$\sqrt{2\log p}=3.90 $         & 0.424& 0.602 & 0.408     & 96.4    \\
		& dIVW&MR-EO $\approx 0.43$  & 0.395 & 0.052 & 0.050     & 94.0    \\&&&&&\\
		& MR-raps& 0     & 0.400 & 0.050 & 0.048     & 93.8    \\
		& MR-raps-shrink&   Bayes shrinkage    & 0.400 & 0.050 & 0.047     & 93.7    \\[0.5ex] \hline
	\end{tabular}
\end{table}

\subsection{Q-Q Plot of BMI-CAD under Balanced Horizontal Pleiotropy} \label{appsec: qq}
Under Assumption $ 2' $ when $\kappa\sqrt{p}\rightarrow \infty$, $\hat{\beta}_{\rmdivw}$ is close to $\beta_0$ and the standardized residuals 
\begin{eqnarray}
	\	\frac{\hat{\Gamma}_j-\hat{\beta}_{ \rmdivw}\hat{\gamma}_j}{\sqrt{\hat{\sigma}_{Yj}^2+\hat{\tau}^2+\hat{\beta}_{ \rmdivw}^2\hat{\sigma}_{Xj}^2}}, ~ j=1,\dots,p. \label{appeq: qq}
\end{eqnarray}
should follow a standard normal distribution. We  use this to assess the plausibility of Assumption $ 2' $  by making a Quantile-Quantile (Q-Q) plot of the standardized residuals as in Figure \ref{appfigure: qq}. Since the residuals line up close to the 45-degree line, Assumption $ 2' $  is likely to hold for the BMI-CAD dataset.
\begin{figure}[!h]
	\includegraphics[scale=0.6]{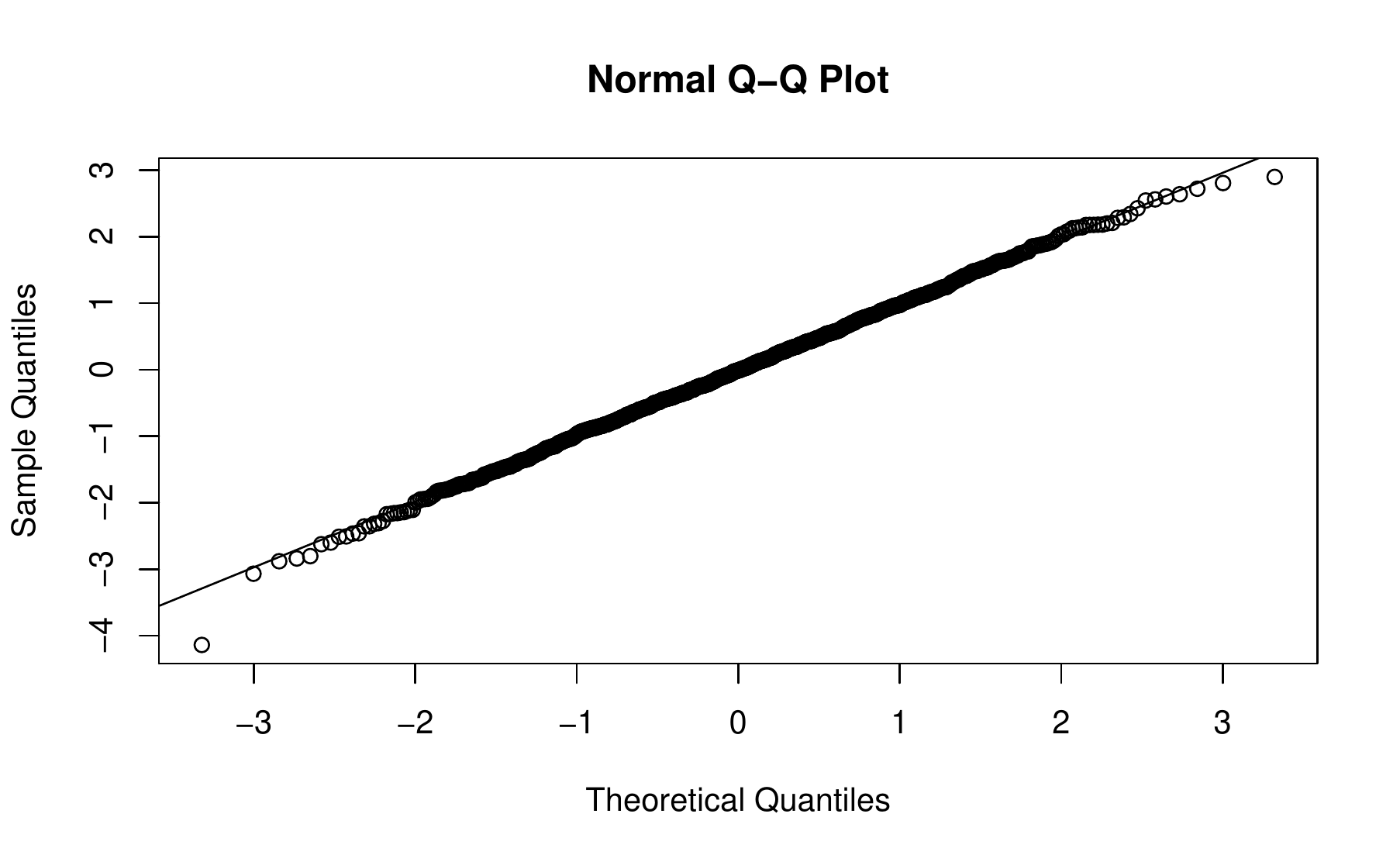}
	\caption[]{Quantile-Quantile plot of the standardized residuals in (\ref{appeq: qq}) against a standard normal.}
	\label{appfigure: qq}
\end{figure}

\clearpage

\section{The dIVW Estimator Beyond Balanced Horizontal Pleiotropy}
When Assumption 2$ ' $ (balanced horizontal pleiotropy) does not hold, the dIVW estimator,  like other MR estimators built upon this assumption, will be biased. Here, we provide heuristics about the magnitude of bias for the dIVW estimator.

Suppose that Assumption 2 holds except $ \hat\Gamma_j \sim N(\alpha_j +\beta_0 \gamma_j, \sigma_{Yj}^2) $ for $ j=1,\dots, p $, where $ \alpha_j $'s are viewed in this section only as fixed parameters. To avoid unnecessary technical details, we assume that the SDs $ \sigma_{Xj}, \sigma_{Yj}, \sigma_{Xj}^* $ are known so that $ \hat\sigma_{Xj}= \sigma_{Xj},  \hat  \sigma_{Yj}=  \sigma_{Yj},  \hat \sigma_{Xj}^*= \sigma_{Xj}^*  $.  

From the proof of Theorem 4.1, we have that as $ \kappa_\lambda \sqrt{p_\lambda } / \max(1, \lambda^2)\rightarrow\infty $ and $ p\rightarrow\infty $, 
\begin{eqnarray}
	&&\hat\beta_{\lambda, \rm dIVW}- \beta_0 - \frac{\sum_{j=1}^p (\alpha_j /\gamma_j) w_j q_{\lambda, j}}{\sum_{j\in S_\lambda} (\hat{\gamma}_j^2- \sigma_{Xj}^2)\sigma_{Yj}^{-2}} \nonumber\\
	&&=    \frac{\sum_{j=1}^{p}(\hat{\Gamma}_j \hat{\gamma}_j-\beta_0\hat{\gamma}_j^2+\beta_0\sigma_{Xj}^2)\sigma_{Yj}^{-2}\indpit- \sum_{j=1}^p (\alpha_j /\gamma_j) w_j q_{\lambda, j}}{\sum_{j\in S_\lambda} (\hat{\gamma}_j^2- \sigma_{Xj}^2)\sigma_{Yj}^{-2}}\nonumber\\
	&&= \frac{\sum_{j=1}^{p}(\hat{\Gamma}_j \hat{\gamma}_j-\beta_0\hat{\gamma}_j^2+\beta_0\sigma_{Xj}^2)\sigma_{Yj}^{-2}\indpit- \sum_{j=1}^p (\alpha_j /\gamma_j) w_j q_{\lambda, j}}{ \sum_{j=1}^p w_jq_{\lambda,j}} \bigg/ (1+o_P(1))\nonumber.
\end{eqnarray}
Then, for any $\epsilon > 0$,
\begin{align}
	&P\left\{ \bigg| \frac{\sum_{j=1}^{p}\left(\hat{\Gamma}_j \hat{\gamma}_j-\beta_0\hat{\gamma}_j^2+\beta_0\sigma_{Xj}^2\right)\sigma_{Yj}^{-2}\indpit-  \sum_{j=1}^p (\alpha_j/ \gamma_j) w_jq_{\lambda, j}}{\sum_{j=1}^p w_j q_{\lambda, j} }\bigg|>\epsilon\right\}\nonumber\\ 
	&\leq \frac{ \Theta (\kappa_{\lambda} p_\lambda+p_\lambda)}{\epsilon^2 \Theta (\kappa_{\lambda} p_\lambda)^2}\rightarrow 0.  \nonumber
\end{align}
In consequence,
\begin{align}
	\hat\beta_{\lambda, \rm dIVW}= \beta_0 + \frac{\sum_{j=1}^p (\alpha_j /\gamma_j) w_j q_{\lambda, j}}{\sum_{j=1}^{p} w_jq_{\lambda, j}} (1+ o_P(1)) + o_P(1) \nonumber,
\end{align}
where $  w_j =  \gamma_j^2/\sigma_{Yj}^2 $, $ v_j = \sigma_{Xj}^2/ \sigma_{Yj}^2 $ and $ q_{\lambda, j}= P(|\hat\gamma_j^*|>\lambda\sigma_{Xj}^*) $. In other words, $ \hat{\beta}_{\lambda, \rm dIVW} $ converges in probability to 
\begin{align}
	\beta_0 + {\rm bias}_{\lambda}= \beta_0 +  \frac{\sum_{j=1}^p (\alpha_j /\gamma_j) w_j q_{\lambda, j}}{\sum_{j=1}^{p} w_jq_{\lambda, j}},  \label{appeq: dIVW}
\end{align}
provided that $   {\rm bias}_{\lambda}$ is bounded. The bias term   ${\rm bias}_{\lambda}$ is a weighted average of $ \alpha_j/\gamma_j $, with weight $ w_jq_{\lambda, j} $. The magnitude of bias will be small, for example,  if (i) only a small proportion of IVs has $ \alpha_j \neq 0$; (ii) the positive $ (\alpha_j /\gamma_j)$'s and negative $ (\alpha_j/ \gamma_j)$'s cancel out; or (iii) $ \alpha_j/\gamma_j $ is small for IVs with large weight $ w_jq_{\lambda, j} $. In particular, if weak IVs have more pleiotropy than strong IVs,  point (iii) suggests that using the weight $ \hat w_j $ in our dIVW estimator not only makes the estimator more robust against weak IVs, but also gives protection against bias due to horizontal pleiotropy.

We conduct a small simulation study to numerically assess the above derivations. The simulation setting is identical to Case 3 in Section \ref{app: simu} of the Supplementary Material, except now we set $ \alpha_j $ differently.	Specifically, we set $ \alpha_j =0.01$ (two times the mean of $ \{|\gamma_j| \}_{j=1}^p$) for $ j=1,\dots, {\rm int}(p\xi) $, and the rest $ \alpha_j=0 $, where $ {\rm int}(x) $ denotes the integer closest to $ x $. In Table \ref{app: tb2}, we vary $ \xi $ to be $ 0.25, 0.5, 0.75 $, meaning that respectively $ 25\%, 50\%, 75\% $ of the $ \alpha_j $'s are non-zero.  

Overall, the performance of the dIVW estimator in Table \ref{app: tb2} agrees with our derivations.  When balanced horizontal pleiotropy does not hold, like all the other methods in Table  \ref{app: tb2},  the dIVW estimator is biased. 
We find that compared to most estimators in the literature, the empirical bias of the dIVW estimator is not too severe even when 75\% of the IVs violate the balanced horizontal pleiotropy assumption ($ \xi=0.75 $). This observation agrees with point (iii) above as the large $\alpha_j/ \gamma_j $'s are downweighted by $ w_jq_{\lambda, j}$ and thus, the bias of the dIVW estimator is small. However, we acknowledge that the simulation study is not comprehensive in terms of all possible violations of the balanced horizontal pleiotropy assumption and our estimator may be more biased in certain settings.

\begin{table}[]
	\caption{Simulation results for Case 3 when balanced horizontal pleiotropy does not hold based on 1,000 repetitions with $\beta_0 =0.4$; $ \xi $ is fraction of non-zero $ \alpha_j $'s, $\beta_0 + {\rm bias}_{\lambda} $ is the theoretically derived limit of $ \hat\beta_{\lambda, \rm dIVW} $ in \eqref{appeq: dIVW};  $ \lambda $ for MR-EO is the simulation average;
		SD is simulation standard deviation;  SE  is  the average of standard errors; CP is the simulation coverage probability of the 95\% confidence interval based on  normal approximation. }\label{app: tb2} \centering
	\begin{tabular}{llccrrrrr} \hline\\[-2ex]
		\multicolumn{1}{c}{$ \xi $}                & \multicolumn{1}{c}{Method} & \multicolumn{1}{c}{$\lambda$}  &		$ \beta_0 + {\rm bias}_{\lambda}$           & \multicolumn{1}{c}{mean} & \multicolumn{1}{c}{SD}    & \multicolumn{1}{c}{SE} & \multicolumn{1}{c}{CP}  \\ \hline\\[-2ex]
		$  0.25    $ & $\text{dIVW}_\alpha$ &     0  & 0.383& 0.384   & 0.056   & 0.057     & 95.3 \\
		& 	$\text{dIVW}_\alpha$   &     5.45  & 0.373& 0.374   & 0.090    & 0.093     & 95.7 \\
		&$\text{dIVW}_\alpha$	  &   $ \sqrt{2\log p} =3.75$     &0.375& 0.375   & 0.072   & 0.075     & 94.9 \\
		& $\text{dIVW}_\alpha$      &     MR-EO$ \approx0.18 $   && 0.383   & 0.056   & 0.057     & 95.0   \\[1.5ex]
		& MR-Egger     &   0     && 0.357   & 0.070    & 0.070      & 91.1 \\
		& MR-Egger &     5.45   && 0.366   & 0.196   & 0.205     & 95.2 \\
		& MR-Egger   &      $ \sqrt{2\log p} =3.75$  && 0.353   & 0.138   & 0.142     & 95.2 \\[1.5ex]
		& MR-median   &    0    && 0.361   & 0.084   & 0.091     & 95.8 \\
		& MR-median &      5.45  && 0.376   & 0.120    & 0.126     & 96.2 \\
		& MR-median  &      $ \sqrt{2\log p} =3.75$  && 0.374   & 0.106   & 0.112     & 95.5 \\[1.5ex]
		& MR-mode      &    0    && 7.181   & 227 & 32173  & 100  \\
		& MR-mode   &    5.45    && 0.375   & 0.129   & 0.154     & 96.5 \\
		& MR-mode   &     $ \sqrt{2\log p} =3.75$   && 0.360    & 0.287   & 0.446     & 96.4 \\[1.5ex]
		& MR-$\text{raps}_\alpha$           &   0    & & 0.385   & 0.056   & 0.057     & 95.5 \\
		& MR-raps-$\text{shrink}_\alpha$    &     Bayes shrinkage   && 0.384   & 0.056   & 0.057     & 95.4 \\[0.5ex] \hline\\[-2ex]
		$  0.5    $	&  $\text{dIVW}_\alpha$         &     0  &		0.421   & 0.422   & 0.056   & 0.061     & 95.7 \\
		&   $\text{dIVW}_\alpha$     &    5.45    &0.408& 0.408   & 0.090    & 0.099     & 97.3 \\
		&   $\text{dIVW}_\alpha$       &     $ \sqrt{2\log p} =3.75$   &0.408& 0.408   & 0.074   & 0.080      & 96.9 \\
		&   $\text{dIVW}_\alpha$      &     MR-EO $ \approx 0.18 $ & & 0.422   & 0.057   & 0.061     & 95.6 \\[1.5ex]
		& MR-Egger     &    0    && 0.397   & 0.071   & 0.074     & 96.3 \\
		& MR-Egger  &    5.45    && 0.503   & 0.202   & 0.219     & 93.3 \\
		& MR-Egger   &    $ \sqrt{2\log p} =3.75$   & & 0.405   & 0.139   & 0.152     & 96.6 \\[1.5ex]
		& MR-median    &   0     && 0.415   & 0.083   & 0.092     & 96.8 \\
		& MR-median&    5.45    && 0.444   & 0.121   & 0.129     & 95.1 \\
		& MR-median  &   $ \sqrt{2\log p} =3.75$     && 0.434   & 0.108   & 0.114     & 95.4 \\[1.5ex]
		& MR-mode      &  0      && 7.410    & 227  & 32433 & 100  \\
		& MR-mode  &    5.45    && 0.451   & 0.134   & 0.157     & 94.9 \\
		& MR-mode    &     $ \sqrt{2\log p} =3.75$   && 0.422   & 0.232   & 0.460      & 96.4 \\[1.5ex]
		& MR-$\text{raps}_\alpha$            &     0  & & 0.424   & 0.057   & 0.061     & 95.7 \\
		& MR-raps-$\text{shrink}_\alpha$      &    Bayes shrinkage   && 0.424   & 0.056   & 0.060      & 95.8 \\[0.5ex]\hline \\[-2ex]
		$ 0.75 $	&   $\text{dIVW}_\alpha$       &   0     &0.453& 0.455   & 0.057   & 0.064     & 89.4 \\
		&   $\text{dIVW}_\alpha$     &    5.45    &0.451& 0.451   & 0.091   & 0.105     & 95.5 \\
		&  $\text{dIVW}_\alpha$     &    $ \sqrt{2\log p} =3.75$    &0.439& 0.439   & 0.074   & 0.084     & 95.6 \\
		&   $\text{dIVW}_\alpha$     &    MR-EO $ \approx0.18 $   && 0.455   & 0.057   & 0.064     & 88.7 \\[1.5ex]
		& MR-Egger     &   0     && 0.426   & 0.073   & 0.078     & 95.3 \\
		& MR-Egger &    5.45    && 0.591   & 0.206   & 0.224     & 86.7 \\
		& MR-Egger   &    $ \sqrt{2\log p} =3.75$    && 0.468   & 0.142   & 0.159     & 94.8 \\[1.5ex]
		& MR-median   &    0    && 0.460    & 0.092   & 0.095     & 90.9 \\
		& MR-median &   5.45     && 0.495   & 0.126   & 0.129     & 88.5 \\
		& MR-median  &    $ \sqrt{2\log p} =3.75$    && 0.480    & 0.114   & 0.116     & 89.9 \\	[1.5ex]
		& MR-mode      &   0     && -10.549 & 315 & 34036 & 100  \\
		& MR-mode   &   5.45     && 0.505   & 0.131   & 0.156     & 89.0   \\
		& MR-mode    &  $ \sqrt{2\log p} =3.75$      && 0.489   & 0.277   & 0.434     & 92.7 \\[1.5ex]
		& MR-$\text{raps}_\alpha$             &     0  && 0.455   & 0.057   & 0.064     & 88.2 \\
		& MR-raps-$\text{shrink}_\alpha$      &  Bayes shrinkage      && 0.455   & 0.057   & 0.064     & 88.6\\[0.5ex]\hline
	\end{tabular}
\end{table}

\section{Proofs}\label{app}
In this section, we provide theoretical proofs for the theorems in the main paper. Throughout this section, we use $c$ to denote a generic constant. For two sequences of  real numbers $a_n$ and $b_n$ indexed by $n$, we write $a_n=O(b_n)$ if $|a_n|\leq cb_n$ for all $n$ and a constant $c$, $a_n=o(b_n)$ if $a_n/b_n\rightarrow 0$ as $n\rightarrow \infty$, $a_n=\Theta(b_n)$ if $c^{-1}b_n\leq |a_n|\leq cb_n$ for all $n$ and a constant $c$. We use  $\xrightarrow{P}$ to denote convergence in probability, $\xrightarrow{D}$ to denote convergence in distribution. For random variables $X$ and $Y$, we denote $X=o_P(Y)$ if $X/Y\xrightarrow{P} 0$,  $X=O_P(Y)$ if $X/Y$ is bounded in probability.

\subsection{ Proof of Theorem \ref{theo: p}}
We first provide a proof for the IVW estimator when $ \lambda=0 $. Then, we provide a full proof for the IVW estimator with general $ \lambda\geq 0 $.  In the  proof, we also derive the asymptotic bias and standard deviation ratio for $ \hat{\beta}_{\rmivw} $ and $ \hat{\beta}_{\lambda, \rmivw} $, respectively. In this section, we assume  $ \hat{\sigma}^*_{Xj}= \sigma^*_{Xj}, \hat{\sigma}_{Yj}= \sigma_{Yj}$ for every $ j $.

\subsubsection{When $ \lambda=0$}
From the definition of the IVW estimator, we have
\begin{eqnarray}
	\hat{\beta}_{\rmivw}-\beta_0&=&\frac{\sum_{j=1}^p\left(\hat{\Gamma}_j \hat{\gamma}_j-\beta_0\hat{\gamma}_j^2\right)\sigma_{Yj}^{-2}}{\sum_{j=1}^p \hat{\gamma}_j^2\sigma_{Yj}^{-2}}\nonumber
	%&=&\left\{\sum_{j\in \calS} \gamma_j^2\sigma_{Yj}^{-2}\right\}^{-1}\left\{\sum_{j\in \calS} \hat{\Gamma}_j \hat{\gamma}_j\sigma_{Yj}^{-2}-\beta_0\left(\sum_{j\in \calS} \hat{\gamma}_j^2\sigma_{Yj}^{-2}-\sum_{j\in \calS} {\gamma}_j^2\sigma_{Yj}^{-2}\right)\right\}+O_P\left(\frac{1}{\kappa s}\right) \nonumber \\
	%&=& \beta_0+\frac{\sum_{j\in \calS}\left(\hat{\Gamma}_j \hat{\gamma}_j-\beta_0 \hat{\gamma}_j^2\right)\sigma_{Yj}^{-2}}{\sum_{j\in \calS} {\gamma}_j^2\sigma_{Yj}^{-2}}+O_P\left(\frac{1}{\kappa s}\right)\nonumber 
\end{eqnarray}
From Assumption \ref{assump: 2}, we have for every $j=1,\dots,p$,
\begin{displaymath}
	\left[ \left(
	\begin{array}{c}
		\hat{\Gamma}_j\\
		\hat{\gamma}_j
	\end{array}\right)-\left(
	\begin{array}{c}
		\beta_0\gamma_j\\
		\gamma_j
	\end{array}\right)
	\right]\sim N \left( \left( 
	\begin{array}{c}
		0\\
		0
	\end{array}
	\right),  \left( 
	\begin{array}{cc}
		\sigma_{Yj}^{2} &0\\
		0 & \sigma_{Xj}^2
	\end{array}
	\right)
	\right)
\end{displaymath}
It is easy to show that for every $j$,
\begin{align}
	&E\left\{\left(\hat{\Gamma}_j \hat{\gamma}_j-\beta_0 \hat{\gamma}_j^2\right)\sigma_{Yj}^{-2}\right\}=-\beta_0v_j,\nonumber\\
	&\var \left\{\left(\hat{\Gamma}_j \hat{\gamma}_j-\beta_0 \hat{\gamma}_j^2\right)\sigma_{Yj}^{-2}\right\}=(w_j+v_j)+\beta_0^2v_j(w_j+2v_j)\nonumber, \\
	&E\left(\hat{\gamma}_j^2\sigma_{Yj}^{-2}\right)=w_j+v_j\nonumber,\\
	&\var\left(\hat{\gamma}_j^2\sigma_{Yj}^{-2}\right)=(4w_j+2v_j)v_j,\nonumber
\end{align}
where $w_j=\gamma_j^2/\sigma_{Yj}^2$, $v_j=\sigma_{Xj}^2/\sigma_{Yj}^2$. First, we want to show
\begin{eqnarray}
	\frac{\sum_{j=1}^p \hat{\gamma}_j^2\sigma_{Yj}^{-2}}{\sum_{j=1}^p (w_j+v_j)}\xrightarrow{P} 1, \label{app:1}
\end{eqnarray}
and it suffices to show that 
\[
\var(\sum_{j=1}^p \hat{\gamma}_j^2\sigma_{Yj}^{-2})/\{\sum_{j=1}^p (w_j+v_j)\}^2=O( \frac{1}{\kappa p+p})=o(1)
\]
when $\kappa p+ p\rightarrow\infty$, then (\ref{app:1}) follows from Markov Inequality. Note that (\ref{app:1}) holds under all four regimes considered in this theorem.

Next, we use Lindeberg Central Limit Theorem to show that as $p\rightarrow \infty$,
\begin{eqnarray}
	\frac{\sum_{j=1}^p\left(\hat{\Gamma}_j \hat{\gamma}_j-\beta_0 \hat{\gamma}_j^2\right)\sigma_{Yj}^{-2}+\sum_{j=1}^p \beta_0v_j}{[\sum_{j=1}^p (w_j+v_j)+\beta_0^2v_j(w_j+2v_j)]^{1/2}} \xrightarrow{D} N(0,1). \label{eq: app, CLT}
\end{eqnarray}

To verify the Lindeberg's condition \cite{Shao}, we define the normalized version of  $\left(\hat{\Gamma}_j \hat{\gamma}_j-\beta_0 \hat{\gamma}_j^2\right)\sigma_{Yj}^{-2}$ as 
\[
K_j=\frac{\left(\hat{\Gamma}_j \hat{\gamma}_j-\beta_0 \hat{\gamma}_j^2\right)\sigma_{Yj}^{-2}+\beta_0v_j}{a_j}.
\]
where $ a_j^2=(w_j+v_j)+\beta_0^2v_j(w_j+2v_j)$. Define  $ \sigma^2_p=\sum_{j=1}^p a_j^2$, the Lindeberg's condition holds because for  any $\epsilon>0$,
\begin{align}
	&\sum_{j=1}^p E\left[\frac{ a_j^2K_j^2}{\sigma^2_{p}}I \left\{a_j|K_j|>\epsilon \sigma_{p}\right\}\right]\leq \sum_{j=1}^{p} \frac{a_j^2}{\sigma_p^2} \max_j E\left[  K_j^2  I \left\{a_j|K_j|>\epsilon \sigma_{p}\right\}  \right]\nonumber\\
	&= \max_j E\left[  K_j^2  I \left\{a_j|K_j|>\epsilon \sigma_{p}\right\}  \right] = o(1)\nonumber
\end{align}
which is a direct result from $E(K_j^2)=1$ and $ \max_j a_j^2/\sigma_p^2=o(1) $ from the assumption that $\max_j(\gamma_j^2\sigma_{Xj}^{-2})/(\kappa p)=o(1)$. 

From equations (\ref{app:1})-(\ref{eq: app, CLT}), and Slutsky's theorem, we directly have  
$$V_{ 0, \rmivw}^{-1/2}\{\hat{\beta}_{\rmivw} -\beta_0  - \mbox{abias} (\hat{\beta}_{\rmivw} ) \}
\xrightarrow{D}N(0,1), $$
where $V_{0, \rmivw} $ 
%\[V_{0, \rmivw} =\frac{\sum_{j=1}^p [ (w_j+\rho_j )+\beta_0^2\rho_j(w_j+2\rho_j )]}{ [\sum_{j=1}^p (w_j+\rho_j)]^2} \]
is the asymptotic variance of $\hat\beta_{\rmivw}$ given in (\ref{normal}) with $\lambda =0$, $\mbox{abias} (\hat{\beta}_{\rmivw} )= 
- \sum_{j=1}^p \beta_0 v_j/ \sum_{j=1}^p (w_j+v_j). $

(a) From the assumption that $\kappa/p\rightarrow \infty$, \eqref{eq: app, CLT} implies 
\[
\frac{\sum_{j=1}^p\left(\hat{\Gamma}_j \hat{\gamma}_j-\beta_0 \hat{\gamma}_j^2\right)\sigma_{Yj}^{-2}}{[\sum_{j=1}^p (w_j+v_j)+\beta_0^2v_j(w_j+2v_j)]^{1/2}} \xrightarrow{D} N(0,1). 
\]

Combining these results and using Slutsky's theorem, we have thus proved 
\[
V^{-1/2}_{0,\rmivw}(\hat{\beta}_{\rmivw}-\beta_0)\xrightarrow{D} N(0,1).
\]
Consistency follows directly from $V_{0,\rmivw}=O(\frac{1}{\kappa p+p})=o(1).$

Then, we prove that the same result still holds when replacing $V_{0,\rmivw}$ by $\hat{V}_{0,\rmivw}$ by showing 
\[
\frac{\hat{V}_{0,\rmivw}}{V_{0,\rmivw}}\xrightarrow{P}1.
\]

From (\ref{app:1}), we have 
\begin{eqnarray}
	\frac{\sum_{j=1}^p \hat{w}_j}{\sum_{j=1}^p (w_j+v_j)}\xrightarrow{P} 1,
\end{eqnarray}
where $ \hat{w}_j=\hat{\gamma}^2_j\sigma_{Yj}^{-2} $. 
Then, it remains to show that 
\[
\frac{\sum_{j=1}^p \hat{\beta}_{ \rmivw}^2v_j(\hat{w}_j+v_j)+\hat{w}_j}{\sum_{j=1}^p \beta_0^2 v_j(w_j+2v_j)+(w_j+v_j)}\xrightarrow{P}1,
\]
where the denominator in the above formula  is of the order $\Theta(\kappa p+p)$. Additionally, the difference between the numerator and the denominator can be written as
\begin{align}
	&\underbrace{(\hat{\beta}_{ \rmivw}^2-\beta_0^2)\sum_{j=1}^pv_j(\hat{w}_j+v_j)}_{A1}+\underbrace{\sum_{j=1}^p (\beta_0^2v_j+1)(\hat{w}_j-w_j-v_j)}_{A2}.\nonumber
\end{align}
Because the estimator is consistent, i.e.,  $\hat{\beta}^2_{ \rmivw}-\beta_0^2=o_P(1)$, and  $\sum_{j=1}^p\hat{w}_j=\sum_{j=1}^p (w_j+v_j) +o_P( \kappa p+p)= \Theta (\kappa p+p)+o_P(\kappa p+p)$ from (\ref{app:1}), we arrive at $A1/\Theta(\kappa p+p)=o_P(1)$. Similarly, because $A2$ is of order $o_P(\kappa p+p)$, we have $A2/\Theta(\kappa p+p)=o_P(1)$ and thus complete the proof. 

(b) We show that $\hat{\beta}_{\rmivw}=\beta_0+o_P(1)$ as $\kappa\rightarrow\infty$. Notice that 
\begin{align}
	&\hat{\beta}_{\rmivw}-\beta_0=\frac{\sum_{j=1}^p\left(\hat{\Gamma}_j \hat{\gamma}_j-\beta_0\hat{\gamma}_j^2\right)\sigma_{Yj}^{-2}}{\sum_{j=1}^p \hat{\gamma}_j^2\sigma_{Yj}^{-2}}\nonumber\\
	&=\underbrace{\frac{\sum_{j=1}^p\left(\hat{\Gamma}_j \hat{\gamma}_j-\beta_0\hat{\gamma}_j^2\right)\sigma_{Yj}^{-2}}{\sum_{j=1}^p (w_j+v_j)}}_{o_P(1)}\biggr/ \underbrace{\frac{\sum_{j=1}^p \hat{\gamma}_j^2\sigma_{Yj}^{-2}}{\sum_{j=1}^p (w_j+v_j) } }_{1+o_P(1)} \nonumber.
\end{align}
The first term is $o_P(1)$ is because of Markov inequality, for any $\epsilon > 0$, as $\kappa\rightarrow \infty$, 
\begin{align}
	&P\left\{\left|\frac{\sum_{j=1}^p\left(\hat{\Gamma}_j \hat{\gamma}_j-\beta_0\hat{\gamma}_j^2\right)\sigma_{Yj}^{-2}}{\sum_{j=1}^p (w_j+v_j)} \right|>\epsilon\right\}\nonumber\\
	& \leq  \frac{(\sum_{j=1}^p \beta_0v_j)^2+\sum_{j=1}^p (w_j+v_j)+\beta_0^2v_j(w_j+2v_j)}{\epsilon^2(\sum_{j=1}^p (w_j+v_j))^2}\nonumber\\
	&=\Theta(\frac{p}{\kappa p+p})^2+\Theta(\frac{1}{\kappa p +p})=o(1). \nonumber
\end{align}
The second term is $1+o_P(1)$ from (\ref{app:1}).

(c) Following the proofs in (b), we can show that as $p\rightarrow \infty$, if $\kappa \rightarrow c>0$, 
\begin{align}
	&\hat{\beta}_{\rmivw}-\beta_0+\frac{\beta_0\sum_{j=1}^p v_j}{\sum_{j=1}^p \hat{\gamma}_j^2\sigma_{Yj}^{-2}}\nonumber\\
	&=\frac{\sum_{j=1}^p\left(\hat{\Gamma}_j \hat{\gamma}_j-\beta_0\hat{\gamma}_j^2\right)\sigma_{Yj}^{-2}+\beta_0\sum_{j=1}^p v_j}{\sum_{j=1}^p \hat{\gamma}_j^2\sigma_{Yj}^{-2}} \nonumber\\
	&=\underbrace{\frac{\sum_{j=1}^p\left(\hat{\Gamma}_j \hat{\gamma}_j-\beta_0\hat{\gamma}_j^2\right)\sigma_{Yj}^{-2}+\beta_0\sum_{j=1}^pv_j}{\sum_{j=1}^p (w_j+v_j) }}_{o_P(1)}\biggr/ \underbrace{\frac{\sum_{j=1}^p \hat{\gamma}_j^2\sigma_{Yj}^{-2}}{\sum_{j=1}^p (w_j+v_j) }}_{1+o_P(1)} =o_P(1) \nonumber
\end{align}
Then the result follows from Slutsky's theorem and (\ref{app:1}).

(d) The statements in the proof of (c) remain true if $\kappa\rightarrow 0$ and $p\rightarrow \infty$. In this regime, \[
\frac{\beta_0\sum_{j=1}^p v_j}{\sum_{j=1}^p \hat{\gamma}_j^2\sigma_{Yj}^{-2}}=\frac{\beta_0\sum_{j=1}^p v_j}{\sum_{j=1}^p (w_j+v_j)}(1+o_P(1))=\beta_0+o_P(1),
\]
where the first equality is from (\ref{app:1}), and thus $\hat{\beta}_{\rmivw}=o_P(1)$.

(e)  When $ \beta_0=0 $,  if $ \max_j(\gamma_j^2\sigma_{Xj}^{-2})/(\kappa p+p) \rightarrow 0$ and $ p\rightarrow \infty $,  then
$ \hat{\beta}_{\rmivw}  $ is consistent and  asymptotically normal. This is a direct result from (\ref{app:1})- (\ref{eq: app, CLT}).

\subsubsection{ When  $ \lambda\geq 0 $} 
(a) We have
\[
\hat{\beta}_{\lambda, \rmivw}-\beta_0=\frac{\sum_{j=1}^{p}\left(\hat{\Gamma}_j \hat{\gamma}_j-\beta_0\hat{\gamma}_j^2\right)\sigma_{Yj}^{-2}\indpit}{\sum_{j=1}^{p} \hat{\gamma}_j^2\sigma_{Yj}^{-2}\indpit}:=\frac{\psi(\beta_0)}{-\psi'},
\]
where $\psi'=\partial \psi(\beta)/\partial \beta$. To prove the asymptotic normality for $\hat{\beta}_{\lambda, \rmivw}-\beta_0$, we first show in Lemma \ref{lemma: ivw three sample 1} that the numerator $\psi(\beta_0)$ is asymptotically normal with mean zero after standardization. We then show in Lemma \ref{lemma: ivw three sample 2} that the denominator $-\psi '$ converges in probability to a positive number.
\begin{lemma}
	\label{lemma: ivw three sample 1} For a given threshold $\lambda\geq 0$, if  $\max_j(\gamma_j^2\sigma_{Xj}^{-2}q_{\lambda,j})/(\kappa_\lambda p_\lambda)\rightarrow 0$, $\kappa_\lambda/p_\lambda\rightarrow\infty$, as $p\rightarrow \infty$,
	\[
	V_1^{-1/2}\psi(\beta_0) \xrightarrow{D} N(0,1),\]
	where
	\[ V_1=\sum_{j=1}^p \left[ (w_j+v_j)q_{\lambda, j} +\beta_0^2v_j(w_j+3v_j)q_{\lambda, j} - \beta_0^2v_j^2q_{\lambda, j}^2\right]
	\]
\end{lemma}
\begin{lemma}
	\label{lemma: ivw three sample 2} For a given threshold $\lambda\geq  0$, if  $(\kappa_\lambda\sqrt{p_\lambda}+\sqrt{p_\lambda})/\max(1,\lambda^2)\rightarrow\infty$,
	\[
	-V_2^{-1}\psi'\xrightarrow{P}1,\]
	where \[ V_2=\sum_{j=1}^p (w_j+v_j)q_{\lambda, j}
	\]
\end{lemma}
Combining Lemmas \ref{lemma: ivw three sample 1} and \ref{lemma: ivw three sample 2}, we have from  Slutsky's theorem that
\[
\left(\frac{V_1}{V_2^2}\right)^{-1/2}(\hat{\beta}_{\lambda,\rmivw}-\beta_0) \xrightarrow{D} N(0,1).
\]
Consistency follows from $V_1/V_2^2=\Theta(\kappa_\lambda p_\lambda+p_\lambda)\rightarrow \infty$ as $\kappa_\lambda/p_\lambda\rightarrow \infty$.

To show the asymptotic normality still holds when replacing $V_{\lambda, \rmivw}$ by $\hat{V}_{\lambda, \rmivw}$, it suffices to show that 
\[
\frac{\hat{V}_{\lambda,\rmivw}}{V_{\lambda,\rmivw}}\xrightarrow{P}1.
\]
From Lemma \ref{lemma: ivw three sample 2}, we have 
\[
\left\{\frac{\sum_{j\in S_\lambda} \hat{w}_j}{\sum_{j=1}^p (w_j+v_j)q_{\lambda, j}}\right\}^2\xrightarrow{P}1.
\]
Then it remains to show that
\[
\frac{\sum_{j\in S_\lambda}\left[ \hat{w}_j+\hat{\beta}_{\lambda,\rmivw}^{2} v_j(\hat{w}_j+v_j)\right]}{\sum_{j=1}^p \left[(w_j+v_j)q_{\lambda,j}+\beta_0^2v_j(w_j+3v_j)q_{\lambda,j}-\beta_0^2v_j^2q_{\lambda,j}^{2}\right]}\xrightarrow{P}1,
\]
where the denominator in the above formula is of order $\Theta(\kappa_\lambda p_\lambda+p_\lambda)$. Also,  the difference between the numerator and the denominator is 
\begin{eqnarray}
	&\underbrace{\sum_{j\in S_\lambda}\hat{w}_j -\sum_{j=1}^{p} (w_j+v_j)q_{\lambda,j}}_{A1} +\underbrace{(\hat{\beta}^2_{\lambda,  \rmivw}-\beta_0^2) \sum_{j\in S_\lambda}v_j(\hat{w}_j+v_j)}_{A2}  \nonumber\\
	&+\underbrace{\beta_0^2\left[  \sum_{j\in S_\lambda}v_j(\hat{w}_j+v_j)-\sum_{j=1}^p v_j(w_j+2v_j)q_{\lambda, j}\right]}_{A3}-\underbrace{\beta_0^2\sum_{j=1}^pv_j^2 q_{\lambda,j}(1-q_{\lambda,j}) }_{A4} \nonumber
\end{eqnarray}
From Lemma \ref{lemma: ivw three sample 2}, we have $A1/\Theta(\kappa_\lambda p_\lambda+p_\lambda)=o_P(1)$ and $A3/\Theta(\kappa_\lambda p_\lambda+p_\lambda)=o_P(1)$. Also, the consistency of $\hat{\beta}_{\lambda,\rmivw}$ and $\sum_{j\in S_\lambda}v_j(\hat{w}_j+v_j)=\sum_{j=1}^p v_j(w_j+2v_j)q_{\lambda,j}+o_P(\kappa_\lambda p_\lambda+p_\lambda)$ from Lemma \ref{lemma: ivw three sample 2} implies $\hat{\beta}_{\lambda,\rmivw}^2-\beta_0^2=o_P(1)$ and $A2/\Theta(\kappa_\lambda p_\lambda+p_\lambda)=o_P(1)$. Finally, we see that $0\leq A4\leq \beta_0^2 \Theta(p_\lambda)$ and thus $A4/\Theta(\kappa_\lambda p_\lambda+p_\lambda)=o(1)$ as $\kappa_\lambda /p_\lambda\rightarrow \infty$. Hence, $\hat{V}_{\lambda,\rmivw}/V_{\lambda,\rmivw}\xrightarrow{P}1$.

In what follows, we prove Lemma \ref{lemma: ivw three sample 1} and Lemma \ref{lemma: ivw three sample 2}, completing the proof of Theorem  \ref{theo: p}(a).

{\bf Proof of Lemma \ref{lemma: ivw three sample 1}.} 
Denote  $\psi_j(\beta_0)=\left(\hat{\Gamma}_j \hat{\gamma}_j-\beta_0\hat{\gamma}_j^2\right)\sigma_{Yj}^{-2}\indpit$, it is easy to show that under Assumptions \ref{assump: 1}-\ref{assump: 2}, for every $j=1,\dots, p$,
\begin{align}
	& E[\psi_j(\beta_0)]=-\beta_0v_jq_{\lambda,j}, \nonumber\\
	& E[\psi_j^2(\beta_0)]=\left[(w_j+v_j) + \beta_0^2 v_j(w_j+3v_j)\right]q_{\lambda,j}, \nonumber \\
	&\var[\psi_j(\beta_0)]= \left[(w_j+v_j) + \beta_0^2 v_j(w_j+3v_j)\right]q_{\lambda,j}-\beta_0^2v_j^2q_{\lambda,j}^{2}=\Theta((w_j+v_j)q_{\lambda,j}).\nonumber
\end{align}
As $p\rightarrow \infty$,  we will use  Lindeberg Central Limit Theorem to show that the properly normalized sum of $\psi_{j}(\beta_0)$ is asymptotically normal, i.e., $V_1^{-1/2}\psi(\beta_0) \xrightarrow{D} N(0,1)$. To do this, we verify the Lindeberg's condition
\[
V_1^{-1}\sum_{j=1}^p \var[\psi_j(\beta_0)] E\left[ K_j^{*2}I \left\{|K_j^*|>\epsilon \left(\frac{V_1}{\var[\psi_j(\beta_0)]}\right)^{1/2}\right\}\right]\rightarrow 0,
\]
where 
\[
K^*_j=\frac{\psi_j(\beta_0)-E\left[\psi_j(\beta_0)\right]}{\{\var [\psi_j(\beta_0)]\}^{1/2}}.
\]  By the conditions  in the theorem, specifically $\max_j(\gamma_j^2\sigma_{Xj}^{-2}q_{\lambda, j})/\kappa_\lambda p_\lambda\rightarrow 0$ and $\kappa_\lambda p_\lambda\rightarrow \infty$, we arrive at $\max_j\{\var[\psi_j(\beta_0)]\}/V_1=o(1)$. Also, because $E(K_j^{*2})=1$ by definition, we have
\begin{align}
	&V_1^{-1}\sum_{j=1}^p \var[\psi_j(\beta_0)] E\left[ K_j^{*2}I \left\{|K_j^*|>\epsilon \left(\frac{V_1}{\var[\psi_j(\beta_0)]}\right)^{1/2}\right\}\right]\nonumber \\
	& \leq \max_j E\left[ K_j^{*2}I \left\{|K_j^*|>\epsilon \left(\frac{V_1}{\var[\psi_j(\beta_0)]}\right)^{1/2}\right\}\right]=o(1)\nonumber 
\end{align}
and the Lindeberg's condition holds. 
Then Lemma \ref{lemma: ivw three sample 1} follows from
\[
\frac{\psi(\beta_0)}{[\var\{\psi(\beta_0)\}]^{1/2}}-\underbrace{\frac{ E[\psi(\beta_0)]}{[\var\{\psi(\beta_0)\}]^{1/2}}}_{\Theta(p_\lambda/(\kappa_\lambda p_\lambda+p_\lambda)^{1/2})}\xrightarrow{D} N(0,1)
\]
where $p_\lambda/(\kappa_\lambda p_\lambda+p_\lambda)^{1/2}=o(1)$ from the assumption that $\kappa_\lambda /p_\lambda\rightarrow\infty$. It also shows the asymptotic bias and standard deviation ratio for $ \hat{\beta}_{\lambda, \rmivw} $ is of order $ p_\lambda^{1/2}/(1+\kappa_\lambda)^{1/2} $.

{\bf Proof of Lemma \ref{lemma: ivw three sample 2}.}
It suffices to prove that $E(-\psi'/V_2)=1$ and $\var (-\psi'/V_2)=o(1)$. Then Lemma \ref{lemma: ivw three sample 2} follows by Markov Inequality, 
\[
P\left(\left|\frac{-\psi'}{V_2}-1\right|>\epsilon\right)\leq \frac{E\left(-\psi'/V_2-1\right)^2}{\epsilon^2}=\frac{\var (-\psi'/V_2)}{\epsilon^2}\rightarrow 0
\]
Some algebra reveals that
%Some details are similar as in the proof of Lemma \ref{lemma: ivw three sample 1}, thus will not be repeated in this proof. It is easy to show\
\begin{align}
	E(-\psi')&=\sum_{j=1}^p(w_j+v_j)q_{\lambda,j}=V_2\nonumber
\end{align}
Next, we show that $\var(-\psi'/V_2)=o(1)$. Notice that
\begin{align}
	\var(-\psi')&=\sum_{j=1}^p \left[(w_j^2+6w_jv_j+3v_j^2)q_{\lambda,j} - (w_j^2+2w_jv_j+v_j^2)q_{\lambda,j}^{2}\right] \nonumber\\
	&=\underbrace{\sum_{j=1}^p w_j^2 q_{\lambda,j}(1-q_{\lambda,j})}_{A1} +\underbrace{\sum_{j=1}^p v_jq_{\lambda,j}\left[6w_j+3v_j-(2w_j+v_j)q_{\lambda,j}\right]}_{A2}\nonumber
\end{align}
where $A1$ involves $w_j^2= \gamma_j^4/\sigma_{Yj}^{4}$, $A2=\Theta(\kappa_\lambda p_\lambda+p_\lambda)$. We will calculate the order of $A1$ in the following. Define $ \delta_j= |\gamma_j|/\sigma_{Xj}^{*} - \lambda$, then $ q_{\lambda, j}= \Phi( |\gamma_j|/\sigma_{Xj}^*-\lambda) + \Phi(- |\gamma_j|/\sigma_{Xj}^*-\lambda)=
\Phi(\delta_j) +\Phi(-\delta_j-2\lambda)$, where $ \Phi(\cdot) $ is the cumulative distribution function for standard normal distribution. It is easy to see that $ A1=0 $ when $ \lambda=0 $. When $ \lambda>0 $, consider the following two cases.
\begin{itemize}
	\item If $\delta_j\geq \lambda$, then 
	\begin{align*}
		w_j^2q_{\lambda,j} (1-q_{\lambda,j}) &\leq c(\gamma_j/ \sigma_{Xj}^{*} )^4 q_{\lambda, j} (1-q_{\lambda, j})\\
		& \leq  c(\gamma_j/ \sigma_{Xj}^{*} )^4 q_{\lambda, j}  \Phi(-\delta_j)\\
		&= c (\delta_j+\lambda)^4q_{\lambda,j} \Phi(-\delta_j)\\
		&\leq c(\delta_j+\lambda)^4q_{\lambda,j}\frac{1}{2}e^{-\delta_j^2/2} \\
		&\leq  c\delta_j^4e^{-\delta_j^2/2}q_{\lambda,j}\\
		&=  O(q_{\lambda,j}) 
	\end{align*}
	where $c$ is a generic constant. The first inequality is from Assumption \ref{assump: 2}, the second inequality is from $1-q_{\lambda,j}\leq \Phi(-\delta_j)$. The fourth line is from the tail bound of  normal distribution. The fifth line is because $\delta_j\geq \lambda$. The last equality is because   $\delta_j^4e^{-\delta_j^2/2}\leq 16e^{-2}$ for $\delta_j>0$. 
	\item If $\delta_j< \lambda$, then 
	\[
	w_j^2 q_{\lambda,j}(1-q_{\lambda,j})\leq   c(\gamma_j/ \sigma_{Xj}^{*} )^4 q_{\lambda,j}=c(\delta_j+\lambda)^4q_{\lambda,j}\leq c\lambda^4 q_{\lambda,j}
	\]
	where the first inequality is because of Assumption \ref{assump: 2} and $1-q_{\lambda,j}\leq 1$, the last inequality is because $0\leq \delta_j+\lambda< 2\lambda $.
\end{itemize}
Combining the above cases, $A1=O(\lambda^4p_\lambda+p_\lambda)$ when $ \lambda>0 $, and $ A1=0 $ when $ \lambda=0 $. Hence,
\[
\var(-\psi')= O\left(\lambda^4 p_\lambda+p_\lambda\right)+\Theta\left( \kappa_\lambda p_\lambda +p_\lambda\right).
\]
and $\var(-\psi'/V_2)\rightarrow 0$ from the condition $(\kappa_\lambda \sqrt{p_\lambda}+\sqrt{p_\lambda})/\max(1,\lambda^2)\rightarrow\infty$, which completes the proof of Lemma \ref{lemma: ivw three sample 2}. \\

Next, we prove (b)-(d). Since the conditions in Lemma \ref{lemma: ivw three sample 2} are implied throughout this theorem, we have under regimes (b)-(d),
\begin{align}
	\hat{\beta}_{\lambda, \rmivw}-\beta_0&=\frac{\sum_{j=1}^{p}\left(\hat{\Gamma}_j \hat{\gamma}_j-\beta_0\hat{\gamma}_j^2\right)\sigma_{Yj}^{-2}\indpit}{\sum_{j=1}^{p} \hat{\gamma}_j^2\sigma_{Yj}^{-2}\indpit},  \nonumber\\
	=& \frac{\sum_{j=1}^{p}\left(\hat{\Gamma}_j \hat{\gamma}_j-\beta_0\hat{\gamma}_j^2\right)\sigma_{Yj}^{-2}\indpit}{\sum_{j=1}^p (w_j+v_j)q_{\lambda, j}}\bigg/(1+o_P(1)).\nonumber
\end{align}

For (b), from the Markov inequality, for $\epsilon >0$ and $\kappa_\lambda \rightarrow \infty$, we have
\begin{align}
	&P\left\{ \bigg| \frac{\sum_{j=1}^{p}\left(\hat{\Gamma}_j \hat{\gamma}_j-\beta_0\hat{\gamma}_j^2\right)\sigma_{Yj}^{-2}\indpit}{\sum_{j=1}^p (w_j+v_j)q_{\lambda, j}}\bigg|>\epsilon\right\}\nonumber\\ 
	&\leq \frac{ \beta_0^2\Theta(p_\lambda^{2})+\Theta (\kappa_{\lambda} p_\lambda+p_\lambda)}{\epsilon^2 \Theta ((\kappa_{\lambda} p_\lambda+p_\lambda)^2)}\rightarrow 0.  \nonumber
\end{align}
Therefore, $\hat{\beta}_{\lambda, \rmivw}=\beta_0+o_P(1)$. \\

For (c)-(d), by Markov's inequality, we have
\begin{align}
	&P\left\{ \bigg| \frac{\sum_{j=1}^{p}\left(\hat{\Gamma}_j \hat{\gamma}_j-\beta_0\hat{\gamma}_j^2+\beta_0\sigma_{Xj}^2\right)\sigma_{Yj}^{-2}\indpit}{\sum_{j=1}^p (w_j+v_j)q_{\lambda, j}}\bigg|>\epsilon\right\}\nonumber\\ 
	&\leq \frac{ \Theta (\kappa_{\lambda} p_\lambda+p_\lambda)}{\epsilon^2 \Theta ((\kappa_{\lambda} p_\lambda+p_\lambda)^2)}\rightarrow 0  \nonumber
\end{align}
as $p_\lambda\rightarrow \infty$. This is equivalent to
\[
\underbrace{\frac{\sum_{j=1}^{p}\left(\hat{\Gamma}_j \hat{\gamma}_j-\beta_0\hat{\gamma}_j^2\right)\sigma_{Yj}^{-2}\indpit}{\sum_{j=1}^p (w_j+v_j)q_{\lambda, j}}}_{\hat{\beta}_{\lambda,\rmivw}-\beta_0+o_P(1)}+ \underbrace{\frac{\beta_0\sum_{j=1}^{p}v_j\indpit}{\sum_{j=1}^p (w_j+v_j)q_{\lambda, j}}}_{\frac{\beta_0\sum_{j=1}^pv_jq_{\lambda,j}}{\sum_{j=1}^p (w_j+v_j)q_{\lambda, j}}+o_P(1)}=o_P(1),
\]
where the second term has the stated probability limit due to Chebyshev's inequality and 
\[
E\left[\frac{\beta_0\sum_{j=1}^{p}v_j\indpit}{\sum_{j=1}^p (w_j+v_j)q_{\lambda, j}}\right]=\frac{\beta_0\sum_{j=1}^pv_jq_{\lambda,j}}{\sum_{j=1}^p (w_j+v_j)q_{\lambda, j}}
\]
\begin{align}
	\var\left\{ \frac{\beta_0\sum_{j=1}^{p}v_j\indpit}{\sum_{j=1}^p (w_j+v_j)q_{\lambda, j}}\right\}&=\frac{\beta_0^2\sum_{j=1}^p v_j^2(q_{\lambda,j}-q_{\lambda,j}^{2})}{\left(\sum_{j=1}^p (w_j+v_j)q_{\lambda, j}\right)^2} \nonumber\\
	&=\frac{O(p_\lambda)}{\Theta(\kappa_\lambda p_\lambda+p_\lambda)^2}\rightarrow 0 \nonumber
\end{align}
as $p_\lambda\rightarrow \infty$. Hence,
\[
\hat{\beta}_{\lambda, \rmivw}-\beta_0+ \frac{\beta_0\sum_{j=1}^{p}v_jq_{\lambda,j}}{\sum_{j=1}^p (w_j+v_j)q_{\lambda, j}}=o_P(1).
\]
When the conditions in (c) are true ($\kappa_\lambda \rightarrow c>0$), 
\[
\hat{\beta}_{\lambda, \rmivw}- \frac{\beta_0\sum_{j=1}^{p}w_jq_{\lambda,j}}{\sum_{j=1}^p (w_j+v_j)q_{\lambda, j}}=o_P(1).
\]
When the conditions in (d) are true ($\kappa_\lambda \rightarrow 0$), $\hat{\beta}_{\lambda, \rmivw}=o_P(1)$.

When $ \beta_0=0 $ and the conditions in (e) are true, the results follow from $ V_1^{-1/2} \psi(\beta_0)\xrightarrow{D} N(0,1) $ and Lemma \ref{lemma: ivw three sample 2}.

\subsection{ Proof of $ \kappa /p $  approximately upper bounded by  $ n_X / p^2  $}
From the definition, \[
\kappa p= n_X\sum_{j=1}^{p} \frac{\gamma_j^2\var(Z_j)}{ \var(X)- \gamma_j^2\var(Z_j)} \approx n_X\sum_{j=1}^{p} \frac{\gamma_j^2\var(Z_j)}{ \var(X)} \leq n_X
\] 
The approximation is reasonable because each $ \gamma_j^2\var(Z_j) $ is small compared with $ \var(X) $.  The last inequality is because $ \sum_{j=1}^{p}\gamma_j^2\var(Z_j)\leq  \var(X) $.

\subsection{Proof of $ \kappa_\lambda $ is approximately increasing in $ \lambda $}
Let   $ \mu_j^2= \gamma_j^2/\sigma_{Xj}^2 $. Without loss of generality, assume $ \gamma_j\geq 0 $, as we can simply replace negative $ \gamma_j $ with their negations.
From the definition 
\begin{align}
	\kappa_\lambda= \frac{\sum_{j =1}^p \mu_j^2q_{\lambda, j} }{\sum_{j=1}^p q_{\lambda,j}}  \approx \frac{\sum_{j  =1}^p\mu_j^2 \Phi(\mu_j- \lambda) }{\sum_{j  =1}^p  \Phi( \mu_j- \lambda)} 
\end{align}
where   $ \Phi(\cdot) $ is the cumulative distribution function for standard normal distribution, the approximation is because $ q_{\lambda, j}\approx\Phi( \mu_j-\lambda)  $.  In what follows, we will prove that 
\begin{align}   
	\frac{f(\lambda)}{g(\lambda)}=\frac{\sum_{j  =1}^p\mu_j^2 \Phi(\mu_j- \lambda) }{\sum_{j  =1}^p  \Phi( \mu_j- \lambda)} \label{appeq: kappa increasing}
\end{align}
is an increasing function in $ \lambda $ for $ \lambda\in (-\infty, \infty) $. The key idea is applying L'H\^opital Monotone Rule \citep{Anderson:2006aa}. 

Note that $ f $ and $ g $ in (\ref{appeq: kappa increasing}) are both continuously differentiable functions of $ \lambda $ on $ (-\infty, +\infty) $ with $ f(\infty) = g(\infty)=0$, and $ g'(x)\neq 0 $ for each $ x\in(\-\infty, \infty) $. It remains to prove that $ f'/g' $ is increasing in ($-\infty, +\infty$), then applying L'H\^opital Monotone Rule, $ f/g $ is also increasing in $ (-\infty, \infty)  $.

In the sequel, we prove that $ f'/g' $ is increasing in $ \lambda $. Notice that 
\begin{align}
	&\left(\frac{f'}{g'}\right)'= \left( \frac{\sum_{j  =1}^p \mu_j^2 \phi (\lambda- \mu_j)}{\sum_{j  =1}^p\phi (\lambda- \mu_j)}\right)' \nonumber\\
	&\propto \left[ \sum_{j=1}^{p} \mu_j^2 \phi(\lambda-\mu_j)\right] \left[ \lambda \sum_{j=1}^{p} \phi(\lambda-\mu_j) - \sum_{j=1}^{p} \mu_j \phi(\lambda-\mu_j)\right] \nonumber\\
	&\qquad- \left[ \sum_{j=1}^{p}  \phi(\lambda-\mu_j)\right] \left[\lambda \sum_{j=1}^{p}\mu_j^2 \phi(\lambda-\mu_j) - \sum_{j=1}^{p}\mu_j^3 \phi(\lambda-\mu_j)\right] \nonumber\\ 
	&=  - \left[ \sum_{j=1}^{p} \mu_j^2\phi(\lambda-\mu_j)\right]
	\left[ \sum_{j=1}^{p} \mu_j\phi(\lambda-\mu_j)\right]+  \left[ \sum_{j=1}^{p} \phi(\lambda-\mu_j)\right]
	\left[ \sum_{j=1}^{p} \mu_j^3\phi(\lambda-\mu_j)\right]\nonumber \\
	&\propto \frac{\sum_{j=1}^{p} \mu_j^3 \phi(\lambda-\mu_j)}{\sum_{j=1}^{p} \mu_j^2  \phi(\lambda-\mu_j)}   -  \frac{\sum_{j=1}^{p} \mu_j \phi(\lambda-\mu_j)}{\sum_{j=1}^{p}  \phi(\lambda-\mu_j)}   \nonumber \\
	&=  \frac{\sum_{j=1}^{p} \mu_j^3 \phi(\lambda-\mu_j)}{\sum_{j=1}^{p} \mu_j^2  \phi(\lambda-\mu_j)}  -  \frac{\sum_{j=1}^{p} \mu_j^2 \phi(\lambda-\mu_j)}{\sum_{j=1}^{p} \mu_j \phi(\lambda-\mu_j)}+  \frac{\sum_{j=1}^{p} \mu_j^2 \phi(\lambda-\mu_j)}{\sum_{j=1}^{p} \mu_j \phi(\lambda-\mu_j)}  -  \frac{\sum_{j=1}^{p} \mu_j \phi(\lambda-\mu_j)}{\sum_{j=1}^{p}  \phi(\lambda-\mu_j)}   \nonumber \\
	&=  \frac{[\sum_{j=1}^{p} \mu_j^3 \phi(\lambda-\mu_j)] [\sum_{j=1}^{p} \mu_j \phi(\lambda-\mu_j) ]- [\sum_{j=1}^{p} \mu_j^2  \phi(\lambda-\mu_j)]^2}{[\sum_{j=1}^{p} \mu_j^2  \phi(\lambda-\mu_j)]  [\sum_{j=1}^{p} \mu_j  \phi(\lambda-\mu_j)]  } \nonumber\\
	&\qquad +  \frac{[\sum_{j=1}^{p} \mu_j^2 \phi(\lambda-\mu_j)] [\sum_{j=1}^{p}  \phi(\lambda-\mu_j) ]- [\sum_{j=1}^{p} \mu_j  \phi(\lambda-\mu_j)]^2}{[\sum_{j=1}^{p} \mu_j \phi(\lambda-\mu_j)] [ \sum_{j=1}^{p}  \phi(\lambda-\mu_j)  ]}  >0\nonumber
\end{align}
where $ \phi(\cdot) $ is the density function of standard normal distribution, and the last line is from Cauchy-Schwartz Inequality, and  equality holds if and only if all the $ \mu_j $'s are equal. Therefore, $ f'/g' $ is an increasing function in $ \lambda $, completing the proof.

\subsection{Proof of Theorem \ref{theo: dIVW}} We divide the proof of Theorem \ref{theo: dIVW} into four parts. We first provide a proof for the dIVW estimator when assuming $ \hat{\sigma}_{Xj}= \sigma_{Xj}, \hat{\sigma}^*_{Xj}= \sigma^*_{Xj}, \hat{\sigma}_{Yj}= \sigma_{Yj}$ and with $ \lambda=0 $. Then, we provide a proof for the dIVW estimator when assuming $ \hat{\sigma}_{Xj}= \sigma_{Xj}, \hat{\sigma}^*_{Xj}= \sigma^*_{Xj}, \hat{\sigma}_{Yj}= \sigma_{Yj}$ and with general $ \lambda\geq 0 $. Next, we  provide a  proof when $ \lambda=0 $ but without assuming knowing $ \sigma_{Xj}, \sigma_{Xj}^*, \sigma_{Yj} $. Finally, we complete the proof by showing the result for general  $ \lambda\geq 0 $ and without assuming knowing $ \sigma_{Xj}, \sigma_{Xj}^*, \sigma_{Yj} $. The proof that $ \hat{V}_{\lambda, \rmdivw} $ is consistent for $ V_{\lambda, \rmdivw} $ is similar to the proof of Theorem \ref{theo: p}(a), thus is omitted.

\subsubsection{When $ \hat{\sigma}_{Xj}= \sigma_{Xj}, \hat{\sigma}^*_{Xj}= \sigma^*_{Xj}, \hat{\sigma}_{Yj}= \sigma_{Yj}$ and  $ \lambda=0 $}
\label{subsubsec: divw, known SDs, lambda=0}
(a) Consistency follows from the proof of (b), so we only prove asymptotic normality here. Following a similar argument in the proof of Theorem \ref{theo: p}(a), we use Lindeberg Central Limit Theorem to show that if $\max_j (\gamma_j^2\sigma_{Xj}^{-2})/(\kappa p+p)\rightarrow 0$, as $p\rightarrow \infty$, 
\begin{align*}
	&\frac{\sum_{j=1}^{p} (\hat{\gamma}_j^2 - \sigma_{X_j}^2) \sigma_{Y_j}^{-2}}{\sum_{j=1}^{p} w_j} \cdot  V_{0, \rm dIVW}^{-1/2} (\hat{\beta}_{\rm dIVW} - \beta_0) \\
	=& \frac{\sum_{j=1}^p\left(\hat{\Gamma}_j \hat{\gamma}_j-\beta_0 \hat{\gamma}_j^2+\beta_0\sigma_{Xj}^2\right)\sigma_{Yj}^{-2}}{[\sum_{j=1}^p (w_j+v_j)+\beta_0^2v_j(w_j+2v_j)]^{1/2}} \xrightarrow{D} N(0,1).\nonumber
\end{align*} 
Then, it suffices to show that if $\kappa\sqrt{p}\rightarrow\infty$,
\begin{eqnarray}
	\frac{\sum_{j=1}^p (\hat{\gamma}_j^2-\sigma_{Xj}^2)\sigma_{Yj}^{-2}}{\sum_{j=1}^p w_j}\xrightarrow{P} 1. \label{app: divwD}
\end{eqnarray}
Equation (\ref{app: divwD}) follows from $E[\sum_{j=1}^p (\hat{\gamma}_j^2-\sigma_{Xj}^2)\sigma_{Yj}^{-2}]=\sum_{j=1}^pw_j$ and 
\[
\frac{\var[\sum_{j=1}^p \hat{\gamma}_j^2\sigma_{Yj}^{-2}] }{\left(\sum_{j=1}^p w_j\right)^2}=O(\frac{1}{\kappa p}+\frac{1}{\kappa^2p})\rightarrow 0
\]
whenever $\kappa\sqrt{p}\rightarrow \infty$. %The proof when replacing $V_{\rmdivw}$ by $\hat{V}_{\rmdivw}$ is similar to the proof of Theorem \ref{theo: p}(a). 

(b) We show that $\hat{\beta}_{\rmdivw}=\beta_0+o_P(1)$ as $\kappa\sqrt{p}\rightarrow\infty$. Notice that 
\begin{align}
	\hat{\beta}_{\rmdivw}-\beta_0&=\frac{\sum_{j=1}^p\left(\hat{\Gamma}_j \hat{\gamma}_j-\beta_0\hat{\gamma}_j^2+\beta_0\sigma_{Xj}^2\right)\sigma_{Yj}^{-2}}{\sum_{j=1}^p( \hat{\gamma}_j^2-\sigma_{Xj}^2)\sigma_{Yj}^{-2}}\nonumber\\
	&=\underbrace{\frac{\sum_{j=1}^p\left(\hat{\Gamma}_j \hat{\gamma}_j-\beta_0\hat{\gamma}_j^2+\beta_0\sigma_{Xj}^2\right)\sigma_{Yj}^{-2}}{\sum_{j=1}^p w_j}}_{o_P(1)}\biggr/ \underbrace{\frac{\sum_{j=1}^p (\hat{\gamma}_j^2-\sigma_{Xj}^2)\sigma_{Yj}^{-2}}{\sum_{j=1}^p w_j}}_{1+o_P(1)} \nonumber,
\end{align}
where the first term is $o_P(1)$ because as $\kappa\sqrt{p}\rightarrow \infty$, we have $\kappa p\rightarrow\infty$, $\kappa^2p\rightarrow \infty$, and by Markov inequality with any $\epsilon > 0$,
\begin{align}
	&P\left\{\left|\frac{\sum_{j=1}^p\left(\hat{\Gamma}_j \hat{\gamma}_j-\beta_0\hat{\gamma}_j^2+\beta_0\sigma_{Xj}^2\right)\sigma_{Yj}^{-2}}{\sum_{j=1}^p w_j} \right|>\epsilon\right\}\nonumber\\
	& \leq  \frac{\sum_{j=1}^p (w_j+v_j) + \beta_0^2v_j(w_j+2v_j)}{\epsilon^2(\sum_{j=1}^p w_j)^2}\nonumber\\
	&=O(\frac{1}{\kappa p}+\frac{1}{\kappa^2p})\rightarrow 0. \nonumber
\end{align}
The second term equals $1+o_P(1)$ from (\ref{app: divwD}).

\subsubsection{ When $ \hat{\sigma}_{Xj}= \sigma_{Xj}, \hat{\sigma}^*_{Xj}= \sigma^*_{Xj}, \hat{\sigma}_{Yj}= \sigma_{Yj}$ and  $ \lambda\geq 0 $}
(a) Here, with slightly abuse of notation, we borrow the notations $\psi(\beta_0)$ and $\psi'$ from the proof of Theorem 3.1. Notice that
\begin{align}
	&\hat{\beta}_{\lambda,\rmdivw}-\beta_0=\frac{\sum_{j=1}^{p}(\hat{\Gamma}_j \hat{\gamma}_j-\beta_0\hat{\gamma}_j^2+\beta_0\sigma_{Xj}^2)\sigma_{Yj}^{-2}\indpit}{\sum_{j=1}^{p}( \hat{\gamma}_j^2-\sigma_{Xj}^2)\sigma_{Yj}^{-2}\indpit}
	:=\frac{\psi(\beta_0)}{-\psi'} \nonumber
\end{align}
where $\psi'=\partial \psi(\beta)/\partial \beta$. To prove the asymptotic normality for $\hat{\beta}_{\lambda,\rmdivw}-\beta_0$, we first show that the numerator $\psi(\beta_0)$ is asymptotically normal with mean zero when properly normalized (Lemma \ref{lemma: three sample 1}), we then show the denominator $-\psi'$ converges in probability to a positive number (Lemma \ref{lemma: three sample 2}).
\begin{lemma}
	\label{lemma: three sample 1}
	For a given threshold $\lambda\geq 0$, if $\max_j(\gamma_j^2\sigma_{Xj}^{-2}q_{\lambda,j})/(\kappa_\lambda p_\lambda+p_\lambda)\rightarrow 0$, as $p\rightarrow \infty$,
	\[
	V_3^{-1/2}\psi(\beta_0) \xrightarrow{D} N(0,1),\qquad V_3=\sum_{j=1}^p [  (w_j+v_j) + \beta_0^2v_j(w_j+2v_j)] q_{\lambda, j}
	\]
\end{lemma}
\begin{lemma} For a given threshold $\lambda\geq 0$, if $\kappa_\lambda  \sqrt{p_\lambda}/\max(1,\lambda^2)\rightarrow \infty$,
	\label{lemma: three sample 2}
	\[
	-V_4^{-1}\psi'\xrightarrow{P}1,\qquad V_4=\sum_{j=1}^p w_jq_{\lambda,j}
	\]
\end{lemma}
Combining Lemmas \ref{lemma: three sample 1} and \ref{lemma: three sample 2}, we have from Slutsky's theorem that 
\[
\left(\frac{V_3}{V_4^2}\right)^{-1/2} (\hat{\beta}_{\lambda,\rmdivw}-\beta_0)\xrightarrow{D} N(0,1).
\]
%The proof when replacing $V_{\lambda,\rmdivw}$ by $\hat{V}_{\lambda, \rmdivw}$ is similar to the proof of Theorem \ref{theo: p}(a). 
Consistency follows from the proof in part (b). In what follows, we prove Lemma \ref{lemma: three sample 1} and Lemma \ref{lemma: three sample 2}, completing the proof of Theorem 4.1 (a).

{\bf Proof of Lemma \ref{lemma: three sample 1}.}
Let  $\psi_j(\beta_0)=\left(\hat{\Gamma}_j \hat{\gamma}_j-\beta_0\hat{\gamma}_j^2+\beta_0\sigma_{Xj}^2\right)\sigma_{Yj}^{-2}\indpit$. Some algebra reveals that
\begin{align}
	& E[\psi_j(\beta_0)]=0, \nonumber\\
	& \var[\psi_j(\beta_0)]=[  (w_j+v_j) + \beta_0^2v_j(w_j+2v_j)] q_{\lambda, j}. \nonumber 
\end{align}
Then, similar to the proof of Lemma  \ref{lemma: ivw three sample 1}, as $p\rightarrow \infty$, applying the Lindeberg Central Limit Theorem with assumption $\max_j(\gamma_j^2\sigma_{Xj}^{-2}q_{\lambda, j})/(\kappa_\lambda p_\lambda+p_\lambda)\rightarrow 0$ arrives at the desired result.

{\bf Proof of Lemma \ref{lemma: three sample 2}.}
The proof is similar to the proof of Lemma \ref{lemma: ivw three sample 2}. In particular, some algebra reveals that
\begin{align}
	E(-\psi')&=\sum_{j=1}^pw_jq_{\lambda,j}=V_4\nonumber\\
	\var(-\psi')&=\sum_{j=1}^p (w_j^2+4w_jv_j+2v_j^2)q_{\lambda,j} - w_j^2q_{\lambda,j}^{2} \nonumber\\
	&=\sum_{j=1}^p w_j^2q_{\lambda,j}(1-q_{\lambda,j}) +\sum_{j=1}^p \left\{4w_jv_j+2v_j^2\right\}q_{\lambda,j}\nonumber
\end{align}
Now, it remains to show that $\var(-\psi'/V_4)=o(1)$ where $V_4=\Theta(\kappa_\lambda p_\lambda)$. Similar to the proof of Lemma \ref{lemma: ivw three sample 2}, we have
\[
\var(-\psi')= O\left(\lambda^4 p_\lambda+p_\lambda\right)+\Theta\left( \kappa_\lambda p_\lambda +p_\lambda\right).
\]
Hence, by the Markov Inequality and the condition $\kappa_\lambda \sqrt{p_\lambda}/\max(1,\lambda^2)\rightarrow\infty$,  we have $\var(-\psi'/V_4)=o(1)$.\\

(b) As $\kappa_\lambda \sqrt{p_\lambda}/\max(1,\lambda^2)\rightarrow\infty$, we have from Lemma \ref{lemma: three sample 2} that
\begin{align}
	&\hat{\beta}_{\lambda,\rmdivw}-\beta_0 \nonumber = \frac{\sum_{j=1}^{p}(\hat{\Gamma}_j \hat{\gamma}_j-\beta_0\hat{\gamma}_j^2+\beta_0\sigma_{Xj}^2)\sigma_{Yj}^{-2}\indpit}{\sum_{j=1}^p w_jq_{\lambda,j} }\biggr/(1+o_P(1))\nonumber
\end{align}
Then, for any $\epsilon > 0$, as $\kappa_\lambda \sqrt{p_\lambda}\rightarrow\infty$, 
\begin{align}
	&P\left\{ \bigg| \frac{\sum_{j=1}^{p}\left(\hat{\Gamma}_j \hat{\gamma}_j-\beta_0\hat{\gamma}_j^2+\beta_0\sigma_{Xj}^2\right)\sigma_{Yj}^{-2}\indpit}{\sum_{j=1}^p w_j q_{\lambda,j}}\bigg|>\epsilon\right\}\nonumber\\ 
	&\leq \frac{ \Theta (\kappa_{\lambda} p_\lambda+p_\lambda)}{\epsilon^2 \Theta (\kappa_{\lambda} p_\lambda)^2}\rightarrow 0,  \nonumber
\end{align}
where the variance of the numerator is calculated in the proof of Lemma \ref{lemma: three sample 1}. Therefore, $\hat{\beta}_{\lambda, \rmdivw}=\beta_0+o_P(1)$.

\subsubsection{Without assuming knowing SDs and $ \lambda=0 $}
Before diving into the investigations about the effect of using SEs $ \hat{\sigma}_{Xj}, \hat{\sigma}_{Yj}, \hat{\sigma}_{Xj}^* $ instead of SDs  $ \sigma_{Xj}, \sigma_{Yj}, \sigma_{Xj}^* $, we introduce some notations and useful results that will be used in the proof.

Assume models (2.1)-(2.2). Let $ (X_i, Z_{X1i}, \dots, Z_{Xpi})_{i=1,\dots, n_X}$ be the i.i.d. individual-level data from the exposure dataset, $ (Y_i, Z_{Y1i}, \dots, Z_{Ypi})_{i=1,\dots, n_Y} $  be the  i.i.d. individual-level data from the outcome dataset. The summary statistics $ \{\hat{\gamma}_j, \hat{\sigma}_{Xj}\}_{j=1,\dots, p} $ are calculated from the individual-level data from the exposure data, $ \{\hat{\Gamma}_j, \hat{\sigma}_{Yj}\}_{j=1,\dots, p} $ are calculated from the individual-level data from the outcome data, but the individual-data are not available to us. Due to the two-sample design, $ (X_i, Z_{X1i}, \dots, Z_{Xpi})_{i=1,\dots, n_X}$ are independent with $ (Y_i, Z_{Y1i}, \dots, Z_{Ypi})_{i=1, \dots, n_Y} $.

We use $ H_Q= Q(Q^T Q)^{-1} Q^T $	to denote the hat matrix, where $ Q^T $ represents the transpose of $ Q $. For example, $ H_e= e(e^Te)^{-1} e^T $ denotes the hat matrix when the regression only includes the intercept, where $ e $ is a column vector of 1's; $ H_{Yj} = \Z_{Yj} (\Z_{Yj}^T \Z_{Yj})^{-1} \Z_{Yj}^T$, where $ \Z_{Yj} $ is the column vector of the $ j $th IV from the outcome dataset after \emph{centering}, i.e., $ e^T\Z_{Yj}=0 $; $ H_{Xj} = \Z_{Xj} (\Z_{Xj}^T \Z_{Xj})^{-1} \Z_{Xj}^T$, where $ \Z_{Xj} $ is the column vector of the $ j $th IV from the exposure dataset after centering. 

From marginal regression, 
\begin{align}
	&\sigma_{Yj}^2=  \frac{\var (Y)- \Gamma_j^2 \var(Z_{j})}{n_Y\var(Z_{j})}, \qquad \hat{\sigma}_{Yj}^2= \frac{1}{\Z_{Yj}^T\Z_{Yj} } \frac{RSS_{Yj}}{n_Y-2} \nonumber\\
	&\sigma_{Xj}^2=  \frac{\var (X)- \gamma_j^2 \var(Z_{j})}{n_X\var(Z_{j})}, \qquad \hat{\sigma}_{Xj}^2= \frac{1}{\Z_{Xj}^T\Z_{Xj} } \frac{RSS_{Xj}}{n_X-2} \nonumber
\end{align}
where $ RSS_{Yj} = \Y^T (I- H_e- H_{Yj}) \Y$ is the residual sum of squares from the marginal regression of $ Y $ on $ Z_{Yj} $, $ \Y^T = (Y_1,\dots, Y_{n_Y})$, $ RSS_{Xj} = \X^T (I- H_e- H_{Xj}) \X$ is the residual sum of squares from the marginal regression of $ X $ on $ Z_{Xj} $, $ \X^T = (X_1,\dots, X_{n_X})$.

We will use the following useful facts in the proof.
\begin{enumerate}
	\item The residual sum of squares from the marginal regression can be decomposed as 
	\begin{align}
		RSS_{Yj}&= \Y^T (I- H_e- H_{Yj})\Y = \Y^T(I-H_e)\Y - \Y^T H_{Yj} \Y \nonumber\\
		& = \Y^T(I-H_e)\Y- \hat{\Gamma}_j^2 \Z_{Yj}^T\Z_{Yj} \label{appeq: RSS_j} \\
		RSS_{Xj}&=\X^T (I- H_e- H_{Xj})\X = \X^T(I-H_e)\X - \X^T H_{Xj} \X \nonumber\\
		& = \X^T(I-H_e)\X- \hat{\gamma}_j^2 \Z_{Xj}^T\Z_{Xj}  \nonumber
	\end{align}
	\item From Taylor expansion,
	\begin{align}
		\frac{\Y^T(I-H_e)\Y}{ \Z_{Yj}^T\Z_{Yj}} &= \frac{\var(Y)}{\var(Z_{j})} \biggr[  1+\left( \frac{\Y^T(I-H_e)\Y}{(n_Y-1)\var(Y)}- 1\right) \nonumber\\
		&\qquad \qquad	- \left( \frac{\Z_{Yj}^T\Z_{Yj}}{(n_Y-1)\var(Z_{j})}-1\right)+O_P(n_Y^{-1})\biggr] \nonumber
	\end{align}
	Thus, from (\ref{appeq: RSS_j}) and the above equation, 
	\begin{align}
		\frac{\hat{\sigma}_{Yj}^2}{\sigma_{Yj}^2}-1&=\sigma_{Yj}^{-2} \left( \frac{\Y^T(I-H_e)\Y}{(n_Y-2)\Z_{Yj}^T\Z_{Yj}}- \frac{\hat{\Gamma}_j^2}{n_Y-2}- \sigma_{Yj}^2\right) \nonumber\\
		&= \psi_{Yj}^2\left[ \left( \frac{\Y^T(I-H_e)\Y}{(n_Y-1)\var(Y)}- 1\right) - \left( \frac{\Z_{Yj}^T\Z_{Yj}}{(n_Y-1)\var(Z_{j})}-1\right)         \right] \nonumber\\
		& -\sigma_{Yj}^{-2} \left(\frac{\hat{\Gamma}_j^2-\Gamma_j^2-\sigma_{Yj}^2}{n_Y-2}\right) + \frac{1}{n_Y} \psi_{Yj}^2- \frac{\beta_0^2w_j}{n_Y(n_Y-2)}+ O_P(n_Y^{-1}) \label{eq: sigma_{Yj}}
	\end{align}
	where $\psi_{Yj}^2= \frac{n_Y \text{var}(Y)}{(n_Y-2)(\text{var}(Y)-\Gamma_j^2\text{var}(Z_{j}))} $, which by assumption is  bounded for every $ j $. 	Similarly,
	\begin{align}
		\frac{\hat{\sigma}_{Xj}^2}{\sigma_{Xj}^2}-1&=\sigma_{Xj}^{-2} \left( \frac{\X^T(I-H_e)\X}{(n_X-2)\Z_{Xj}^T\Z_{Xj}}- \frac{\hat{\gamma}_j^2}{n_Y-2}- \sigma_{Xj}^2\right) \nonumber\\
		&= \psi_{Xj}^2\left[ \left( \frac{\X^T(I-H_e)\X}{(n_X-1)\var(X)}- 1\right) - \left( \frac{\Z_{Xj}^T\Z_{Xj}}{(n_X-1)\var(Z_{j})}-1\right)         \right]  \nonumber\\
		&-\sigma_{Xj}^{-2} \left(\frac{\hat{\gamma}_j^2-\gamma_j^2-\sigma_{Xj}^2}{n_X-2}\right) +\frac{1}{n_X} \psi_{Xj}^2- \frac{w_j/v_j}{n_X(n_X-2)}+ O_P(n_X^{-1}) 
	\end{align}
	where $\psi_{Xj}^2= \frac{n_X\text{var}(X)}{(n_X-2)(\text{var}(X)-\gamma_j^2\text{var}(Z_{j}))} $, which by assumption is  bounded for every $ j $.
	
	\item  From Taylor expansion, we have 
	\begin{align}
		\frac{1}{\hat{\sigma}_{Yj}^2} = \frac{1}{{\sigma}_{Yj}^2} - \frac{1}{{\sigma}_{Yj}^4} (\hat{\sigma}_{Yj}^2- {\sigma}_{Yj}^2) + \frac{1}{\sigma_{Yj}^6} (\hat{\sigma}_{Yj}^2- {\sigma}_{Yj}^2)^2 + \dots \nonumber
	\end{align}
	thus,
	\begin{align}
		\frac{\hat{\sigma}^{-2}_{Yj}}{{\sigma}_{Yj}^{-2}}-1 = - \left(\frac{\hat{\sigma}_{Yj}^2}{{\sigma}_{Yj}^2}-1\right) + \underbrace{\left( \frac{\hat{\sigma}_{Yj}^2}{{\sigma}_{Yj}^2}-1 \right)^2}_{O_p(n_Y^{-1})}+ o_P(n_Y^{-1}) \label{appeq: sigma_Yj}
	\end{align}

	\item From the above derivations,  we have
	\begin{align}
		&-	\sum_{j=1}^p\beta_0v_j\left(\frac{\hat{\sigma}_{Yj}^{-2}}{\sigma_{Yj}^{-2}}- 1\right)= \sum_{j=1}^{p} \beta_0 v_j \left(\frac{\hat{\sigma}_{Yj}^2}{{\sigma}_{Yj}^2}-1\right) + O_P(p/n_Y) \nonumber\\
		&=\left( \sum_{j=1}^p\beta_0v_j\psi_{Yj}^2 \right) \left(\frac{\Y^T (I-H_e)\Y}{(n_Y-1) \var(Y)}-1\right)- \underbrace{\sum_{j=1}^p \beta_0v_j\sigma_{Yj}^{-2} \left( \frac{\hat{\Gamma}_j^2- \Gamma_j^2-\sigma_{Yj}^2}{n_Y-2}\right) }_{O_p\left(\frac{\sqrt{\kappa p+p}}{n_Y}\right)}\nonumber\\
		&-\underbrace{\sum_{j=1}^p\beta_0v_j\psi_{Yj}^2 \left(\frac{\Z_{Yj}^T\Z_{Yj}}{(n_Y-1) \var(Z_{j})}-1\right)}_{O_p(\sqrt{p/n_Y})}   +O_P(p/n_Y) + \Theta(p/n_Y+ \kappa p/n_Y^2) \nonumber
	\end{align}
	and 
	\begin{align}
		&\sum_{j=1}^p\beta_0v_j\left(\frac{\hat{\sigma}_{Xj}^{2}}{\sigma_{Xj}^{2}}- 1\right) \label{eq: B2}\\
		&=\left( \sum_{j=1}^p\beta_0v_j\psi_{Xj}^2 \right) \left(\frac{\X^T (I-H_e)\X}{(n_X-1) \var(X)}-1\right)
		- \underbrace{\sum_{j=1}^p \beta_0v_j\sigma_{Xj}^{-2} \left( \frac{\hat{\gamma}_j^2- \gamma_j^2-\sigma_{Xj}^2}{n_X-2}\right) }_{O_p\left(\frac{\sqrt{\kappa p+p}}{n_X}\right)}\nonumber\\
		& -\underbrace{\sum_{j=1}^p\beta_0v_j\psi_{Xj}^2 \left(\frac{\Z_{Xj}^T\Z_{Xj}}{(n_X-1) \var(Z_{j})}-1\right)}_{O_p(\sqrt{p/n_X})} +O_P(p/n_X) + \Theta(p/n_X+ \kappa p/n_X^2) \nonumber
	\end{align}
	
	\item By central limit theorem, we have the conditional distribution of $ \hat{\Gamma}_j, \hat{\gamma}_j$ as 
	\begin{align}
		&	\hat{\Gamma}_j |\Z_{Yj} \sim N\left(\Gamma_j, \frac{ \var (Y)- \Gamma_j^2\var(Z_{j}) }{\Z_{Yj}^T\Z_{Yj}} \right)  \nonumber\\
		&	\hat{\gamma}_j |\Z_{Xj} \sim N\left(\gamma_j, \frac{ \var (X)- \gamma_j^2\var(Z_{j}) }{\Z_{Xj}^T\Z_{Xj}} \right) \nonumber
	\end{align}
	Thus,
	\begin{align}
		&	\cov(\hat{\Gamma}_j, \hat{\Gamma}_j^2\Z_{Yj}^T\Z_{Yj}) = 2\Gamma_j\{\var(Y)- \Gamma_j^2\var(Z_{j})\} \label{appeq:cov}\\
		&	\cov(\hat{\gamma}_j, \hat{\gamma}_j^2\Z_{Xj}^T\Z_{Xj}) = 2\gamma_j\{\var(X)- \gamma_j^2\var(Z_{j})\} \label{appeq: cov, X}
	\end{align}
	\item   Cauchy-Schwartz inequality:  for random variables $ X_j, j=1,\dots, p $ with finite second moments, it is true that 
	\begin{align}
		\var(\sum_{j=1}^p X_i) \leq p \sum_{j=1}^{p} \var(X_j) \label{appeq: cauchy-schwartz}
	\end{align}
	with equality holds when all the pairwise correlations among $ X_j $'s are equal to 1.
\end{enumerate}

Lemma \ref{lemma: dIVW} studies the numerator of $ \hat{\beta}_{\rmdivw}- \beta_0 $ without assuming $ \hat{\sigma}_{Xj}= \sigma_{Xj}, \hat{\sigma}^*_{Xj}= \sigma^*_{Xj}, \hat{\sigma}_{Yj}= \sigma_{Yj}$.
\begin{lemma} \label{lemma: dIVW}
	Assume models (\ref{eq: exposure-IV})-(\ref{eq: outcome-exposure}),  Assumptions \ref{assump: 1}-\ref{assump: 2},   and that $ p/n_X \rightarrow 0$,
	$\kappa \sqrt{p} \to \infty$ and 
	$\max_j (\gamma_j^2\sigma_{Xj}^{-2} )/ (\kappa p+p)\to 0$. Then, as $ p, n_X \rightarrow \infty $,
	\begin{align}
		\frac{\sum_{j=1}^p\left(\hat{\Gamma}_j \hat{\gamma}_j-\beta_0\hat{\gamma}_j^2+\beta_0\hat{\sigma}_{Xj}^2\right)\hat{\sigma}_{Yj}^{-2}}{ \sigma_{0,\rmdivw} }\xrightarrow{D} N(0,1)\nonumber 
	\end{align}
	where $  	\sigma^2_{0,\rmdivw}=\sum_{j=1}^p (w_j+v_j)+\beta_0^2v_j(w_j+2v_j), 
	v_j=\sigma_{Xj}^2/\sigma_{Yj}^2,  w_j=\gamma_j^2/\sigma_{Yj}^2  $.
\end{lemma}

\begin{proof}
	Consider the numerator 
	\begin{align}
		&\sum_{j=1}^p\left(\hat{\Gamma}_j \hat{\gamma}_j-\beta_0\hat{\gamma}_j^2+\beta_0\hat{\sigma}_{Xj}^2\right)\hat{\sigma}_{Yj}^{-2} \nonumber\\
		&= \sum_{j=1}^p\left(\hat{\Gamma}_j \hat{\gamma}_j-\beta_0\hat{\gamma}_j^2+ \beta_0\sigma_{Xj}^2\right){\sigma}_{Yj}^{-2} + \sum_{j=1}^p \beta_0(\hat{\sigma}_{Xj}^2- \sigma_{Xj}^2)\sigma_{Yj}^{-2}\nonumber\\
		&+ 	\sum_{j=1}^p\left(\hat{\Gamma}_j \hat{\gamma}_j-\beta_0\hat{\gamma}_j^2+\beta_0\sigma_{Xj}^2\right)(\hat{\sigma}_{Yj}^{-2}- {\sigma}_{Yj}^{-2}) + \sum_{j=1}^{p} \beta_0(\hat{\sigma}_{Xj}^2- \sigma_{Xj}^2)(\hat{\sigma}_{Yj}^{-2}- \sigma_{Yj}^{-2}) \nonumber\\
		&=\underbrace{ \sum_{j=1}^p\left(\hat{\Gamma}_j \hat{\gamma}_j-\beta_0\hat{\gamma}_j^2+ \beta_0 \sigma_{Xj}^2\right){\sigma}_{Yj}^{-2}}_{B_1} +\Psi_X \underbrace{ \left(\frac{\X^T (I-H_e)\X}{(n_X-1) \var(X)}-1\right) }_{B_2}\qquad \text{From (\ref{eq: B2})} \nonumber\\ 
		&  +	\underbrace{ \sum_{j=1}^p\left(\hat{\Gamma}_j \hat{\gamma}_j-\beta_0\hat{\gamma}_j^2+ \beta_0 \sigma_{Xj}^2\right)(\hat{\sigma}_{Yj}^{-2}- {\sigma}_{Yj}^{-2})}_{B_3}+ \underbrace{ \sum_{j=1}^{p} \beta_0(\hat{\sigma}_{Xj}^2- \sigma_{Xj}^2)(\hat{\sigma}_{Yj}^{-2}- \sigma_{Yj}^{-2})}_{B_4}  \nonumber\\
		&\qquad +  \underbrace{  O_P\left( \frac{p}{n_X}\right)   + O_P \left( \sqrt{\frac{p}{n_X}}\right) +O_P\left(\frac{\sqrt{\kappa p+p}}{n_X}\right)+\Theta\left( \frac{p}{n_X}+\frac{\kappa p}{n_X^2}\right)}_{o_P(\sqrt{\kappa p+p})+o(\sqrt{\kappa p+p}) } \nonumber\\ 
		&= B_1+\Psi_X B_2+B_3+B_4+  o_P(\sqrt{\kappa p+p})+o(\sqrt{\kappa p+p}) \nonumber
	\end{align}
	where $ 	\Psi_X=  \sum_{j=1}^p \beta_0v_j\psi_{Xj}^2,  \psi_{Xj}^2= \frac{n_X \var(X)}{(n_X-2) (\var(X)- \gamma_j^2\var(Z_j))} $.

	In the sequel, we will deal with $ B_1 $ and $ \Psi_X B_2 $ separately, and show that $ B_3+B_4 =o_P(\sqrt{\kappa p+p})$.

	The first term $ B_1 $ is studied in Section \ref{subsubsec: divw, known SDs, lambda=0}, as $ p\rightarrow \infty $
	\begin{eqnarray}
		\frac{B_1}{[\sum_{j=1}^p (w_j+v_j)+\beta_0^2(w_j+2v_j)v_j]^{1/2}} \xrightarrow{D} N(0,1).\nonumber
	\end{eqnarray}
	For the second term $ B_2 $, by central limit theorem for i.i.d. random variables, as $ n_X\rightarrow \infty $,
	\begin{align}
		\sqrt{n_X}B_2  \xrightarrow{D} N\left( 0, \frac{E[(X-E(X))^4]}{(\var(X))^2}-1\right) \nonumber
	\end{align} 
	Consider the covariance 
	\begin{align}
		&	\cov(B_1, B_2)\nonumber\\
		&=  \sum_{j=1}^{p} \cov\left(  \left(\hat{\Gamma}_j \hat{\gamma}_j-\beta_0\hat{\gamma}_j^2+ \beta_0 \sigma_{Xj}^2\right){\sigma}_{Yj}^{-2},    \frac{\X^T (I-H_e)\X }{(n_X-1) \var(X)}    \right) \nonumber \\
		& = \sum_{j=1}^{p} \cov\left( ( \beta_0\gamma_j\hat{\gamma}_j- \beta_0\hat{\gamma}_j^2)\sigma_{Yj}^{-2},    \frac{\X^T (I-H_e-H_{Xj})\X + \hat{\gamma}_j^2 \Z_{Xj}^T\Z_{Xj}}{(n_X-1) \var(X)}    \right) \nonumber \\
		&= \sum_{j=1}^{p} \cov\left( ( \beta_0\gamma_j\hat{\gamma}_j- \beta_0\hat{\gamma}_j^2)\sigma_{Yj}^{-2},    \frac{ \hat{\gamma}_j^2 \Z_{Xj}^T\Z_{Xj}}{(n_X-1) \var(X)}    \right) \nonumber \\
		&=\frac{-2\beta_0n_X}{(n_X-2)(n_X-1)} \sum_{j=1}^{p} (w_j+v_j)\psi_{Xj}^{-2}\nonumber
	\end{align}
	where the second equality is because $\hat{\Gamma}_j  $'s are independent with $ \X $, the third equality is because $ \hat{\gamma}_j \perp \X^T(I-H_e-H_{Xj})\X $, the last equality is from  (\ref{appeq: cov, X}) and $ \cov (\hat{\gamma}_j^2, \hat{\gamma}_j^2\Z_{Xj}^T\Z_{Xj}) = (4\gamma_j^2+2\sigma_{Xj}^2) \var(X-\gamma_jZ_j)$.
	
	Therefore, 
	\begin{align}
		\var(B_1+\Psi_X B_2)&= \var(B_1) + \frac{\Psi_X^2}{n_X} \left( \frac{E[(X-E(X))^4]}{(\var(X))^2}-1\right)\nonumber\\ &-\frac{4\beta_0 n_X\Psi_X}{(n_X-2)(n_X-1)}\sum_{j=1}^{p} (w_j+v_j)\psi_{Xj}^{-2} \nonumber\\
		&= \Theta(\kappa p+p) +\Theta(\frac{p^2}{n_X} ) -\Theta( \frac{p}{n_X} (\kappa p+p))\nonumber\\
		&=\var(B_1)+ o( \kappa p+p) \nonumber
	\end{align}
	where the last equality is from $ p/n_X=o(1) $.
	
	Besides, given that the conditional distribution  $ B_1|\{ \Z_{X1}, \dots, \Z_{Xp}, \X\} $ is asymptotically normal,  it implies that any linear combination of $ B_1 $ and $ B_2 $ is asymptotically normal conditional on $  \Z_{X1}, \dots, \Z_{Xp}, \X$. By bounded convergence theorem, it is also true that any linear combination of $ B_1 $ and $ B_2 $ is asymptotically normal unconditionally. Therefore,  as $ p, n_X\rightarrow \infty $  and $ p/n_X\rightarrow 0$, 
	\begin{align}
		\frac{B_1+\Psi_X B_2}{ \sigma_{0, \rmdivw}}\xrightarrow{D}N(0,1)
	\end{align}

	Next, we prove that $ B_3=o_P(\sqrt{\kappa p+p}) $. Note that
	\begin{align}
		B_3 \nonumber&= \sum_{j=1}^p\left(\hat{\Gamma}_j \hat{\gamma}_j-\beta_0\hat{\gamma}_j^2+ \beta_0 \sigma_{Xj}^2\right) (\hat{\sigma}_{Yj}^{-2} - {\sigma}_{Yj}^{-2}) \nonumber\\
		&=-\underbrace{\sum_{j=1}^{p} \left(\hat{\Gamma}_j \hat{\gamma}_j-\beta_0\hat{\gamma}_j^2+ \beta_0 \sigma_{Xj}^2\right) {\sigma}_{Yj}^{-2} \left( \frac{\hat{\sigma}_{Yj} ^{2}}{{\sigma}_{Yj}^{2}} -1 \right)}_{B_{31}} \nonumber\\
		&+ \underbrace{\sum_{j=1}^{p} \left(\hat{\Gamma}_j \hat{\gamma}_j-\beta_0\hat{\gamma}_j^2+ \beta_0 \sigma_{Xj}^2\right) {\sigma}_{Yj}^{-2} \left( \frac{\hat{\sigma}_{Yj} ^{2}}{{\sigma}_{Yj}^{2}} -1 \right)^2}_{B_{32}} + \dots\nonumber
	\end{align}
	Applying Cauchy-Schwartz inequality in (\ref{appeq: cauchy-schwartz}) to $ B_{32} $, we have
	that 
	\begin{align}
		\var(B_{32})& \leq p \sum_{j=1}^p \var \left\{   \left(\hat{\Gamma}_j \hat{\gamma}_j-\beta_0\hat{\gamma}_j^2+ \beta_0 \sigma_{Xj}^2\right) {\sigma}_{Yj}^{-2} \left( \frac{\hat{\sigma}_{Yj} ^{2}}{{\sigma}_{Yj}^{2}} -1 \right)^2  \right\} \nonumber\\
		&=p\sum_{j=1}^{p} \left\{(w_j+v_j)+ \beta_0^2 (w_j+2v_j) v_j\right\} \Theta(n_Y^{-2}) = \Theta\left( \frac{p (\kappa p+p)}{n^2_Y}\right) \nonumber
	\end{align}
	where the  second line is due to $( \hat{\Gamma}_j, \hat{\gamma}_j) \perp \hat{\sigma}_{Yj} $ and (\ref{appeq: sigma_Yj}). Similarly, the remaining terms' variance is of order $ o( p (\kappa p+p)/ n^2_Y) $. Therefore, $ B_{32} $ together with the remaining terms are $ o_P( \sqrt{\kappa p+p}) $ when $ p/n^2_Y \rightarrow 0$.
	
	For $ B_{31}$, 
	\begin{align}
		B_{31}&=\underbrace{\left[ \sum_{j=1}^p (\hat{\Gamma}_j \hat{\gamma}_j-\beta_0\hat{\gamma}_j^2+ \beta_0 \sigma_{Xj}^2)\sigma_{Yj}^{-2} \psi_{Yj}^2 \right]\left( \frac{\Y^T(I-H_e)\Y}{(n_Y-1)\var(Y)}- 1\right)}_{A_1}  \nonumber\\
		&- \underbrace{\sum_{j=1}^p (\hat{\Gamma}_j \hat{\gamma}_j-\beta_0\hat{\gamma}_j^2+ \beta_0 \sigma_{Xj}^2)\sigma_{Yj}^{-2} \psi_{Yj}^2   \left( \frac{\Z_{Yj}^T\Z_{Yj}}{(n_Y-1)\var(Z_{j})}-1\right)  }_{A_2}       \nonumber\\
		&- \underbrace{ \sum_{j=1}^p (\hat{\Gamma}_j \hat{\gamma}_j-\beta_0\hat{\gamma}_j^2+ \beta_0 \sigma_{Xj}^2)\sigma_{Yj}^{-2} \psi_{Yj}^2    \left(\frac{\hat{\Gamma}_j^2-\Gamma_j^2-\sigma_{Yj}^2}{n_Y-2}\right)\sigma_{Yj}^{-2}}_{A_3}  +  o_P(\sqrt{\kappa p+p})\nonumber
	\end{align} 
	The first term $ A_1 =o_P(\sqrt{\kappa p+p})$ because $ \Y^T(I-H_e)\Y/\{(n_Y-1) \var(Y)\})-1=o_P(1)$ from the central limit theorem. Consider the second term. From $ E(\hat{\Gamma}_j|\Z_{Yj})=\Gamma_j $, $ \hat{\gamma}_j\perp \Z_{Yj} $, we have $ E(A_2)=E[E(A_2|\Z_{Yj})]=0 $.	For its variance, 
	\begin{align}
		&\var \left\{   (\hat{\Gamma}_j \hat{\gamma}_j-\beta_0\hat{\gamma}_j^2+ \beta_0 \sigma_{Xj}^2)\sigma_{Yj}^{-2} \psi_j^2   \left( \frac{\Z_{Yj}^T\Z_{Yj}}{(n_Y-1)\var(Z_{j})}-1\right) \right\} \nonumber\\
		&= E\left\{    \var( \hat{\Gamma}_j \hat{\gamma}_j-\beta_0\hat{\gamma}_j^2+ \beta_0 \sigma_{Xj}^2|\Z_{Yj}   ) \sigma_{Yj}^{-4} \psi_j^4    \left( \frac{\Z_{Yj}^T\Z_{Yj}}{(n_Y-1)\var(Z_{j})}-1\right)^2  \right\} \nonumber\\
		&= E\left[ \left\{   \beta_0^2(w_j+2v_j)v_j+\frac{w_j+v_j}{\Z_{Yj}^T\Z_{Yj}} \var(Y-\Gamma_j Z_{j})\sigma_{Yj}^{-2}\right\} \left( \frac{\Z_{Yj}^T\Z_{Yj}}{(n_Y-1)\var(Z_{j})}-1\right)^2  \right]\psi_j^4 \nonumber\\
		&= \beta_0^2(w_j+2v_j) v_j\psi_j^4 E\left(\frac{\Z_{Yj}^T\Z_{Yj}}{(n_Y-1)\var(Z_{j})}-1\right)^2 \nonumber\\
		&\qquad + (w_j+v_j)\psi_j^4  \underbrace{\sigma_{Yj}^{-2}\var(Y-\Gamma_j Z_{j}) }_{n_Y\var(Z_{j})} E\left\{  -\frac{1}{(n_Y-1)\var(Z_{j})} +\frac{1}{\Z_{Yj}^T\Z_{Yj}}\right\} \nonumber\\
		&=   \underbrace{ \beta_0^2(w_j+2v_j) v_j\psi_j^4 E\left(\frac{\Z_{Yj}^T\Z_{j}}{(n_Y-1)\var(Z_{j})}-1\right)^2}_{\Theta((w_j+1)/n_Y)} \nonumber\\
		&\qquad + \underbrace{(w_j+v_j)v_j\psi_j^4  E\left\{ \frac{\var(Z_{j})}{\Z_{j}^T\Z_{Yj}/(n_Y-1)}-1\right\}}_{{o(w_j+1)}} \nonumber
	\end{align}
	Hence, $ A_2=o_P(\sqrt{\kappa p+p}) $.	Finally, consider the third term. One can show that $ E(A_3)= \frac{2\beta_0}{n_Y-2}\sum_{j=1}^{p} w_j\psi_j^2  = \Theta( \kappa p/n_Y)$, and 
	\begin{align}
		&	\var\left\{   (\hat{\Gamma}_j \hat{\gamma}_j-\beta_0\hat{\gamma}_j^2+ \beta_0 \sigma_{Xj}^2)\sigma_{Yj}^{-2} \psi_j^2    \left(\frac{\hat{\Gamma}_j^2-\Gamma_j^2-\sigma_{Yj}^2}{n_Y-2}\right)\sigma_{Yj}^{-2}  \right\} \nonumber\\
		&=\frac{ C_1(\gamma^4_j\sigma_{Xj}^{-4})+C_2(\gamma^2_j\sigma_{Xj}^{-2})+C_3}{(n_Y-2)^2} \nonumber
	\end{align}
	where $ C_1, C_2, C_3 $ are bounded constants. 	Hence, $ \var (A_3)= \frac{\Theta(\sum_{j=1}^{p}\gamma^4_j\sigma_{Xj}^{-4})+\Theta( \kappa p +p) }{n_Y^2}$ and $ A_3=\Theta(\kappa p/n_Y)+o_P(\sqrt{\kappa p+p} )$ when $ \sum_{j=1}^{p} \gamma^4_j\sigma_{Xj}^{-4}/(n_Y^2 (\kappa p+p)) \rightarrow 0$. This condition holds because 
	\begin{align}
		&	\sum_{j=1}^{p} \gamma^4_j\sigma_{Xj}^{-4}/(n_Y^2 (\kappa p+p)) \leq \max_j( \gamma_j^2\sigma_{Xj}^{-2})	\sum_{j=1}^{p} \gamma^2_j\sigma_{Xj}^{-2}/(n_Y^2 (\kappa p+p))\nonumber\\
		&= \frac{\max_j( \gamma_j^2\sigma_{Xj}^{-2}) }{\kappa p+p} \frac{\Theta(\kappa p)}{n_Y^2} \rightarrow 0. \nonumber
	\end{align} 
	Finally,  $ B_4=o_P(\sqrt{\kappa p+p})$ because
	\begin{align}
		B_4&=\sum_{j=1}^{p}\beta_0(\hat{\sigma}_{Xj}^2- \sigma_{Xj}^2)(\hat{\sigma}_{Yj}^{-2}- \sigma_{Yj}^{-2})\nonumber\\
		&=\underbrace{-\beta_0\left(\sum_{j=1}^{p} v_j^2\psi_{Xj}^2\psi_{Yj}^2 \right) \left( \frac{\X^T (I-H_e)\X}{(n_X-1)\var(X)}-1\right)\left( \frac{\Y^T (I-H_e)\Y}{(n_Y-1)\var(Y)}-1\right)}_{ O_P (p/\sqrt{n_Xn_Y})} \nonumber\\
		&+  o_P( p/n_X) \nonumber
	\end{align}
	This completes the proof.
\end{proof}

Lemma \ref{lemma: dIVW denominator} studies the denominator of $ \hat{\beta}_{\rmdivw}- \beta_0 $ without assuming $ \hat{\sigma}_{Xj}= \sigma_{Xj}, \hat{\sigma}^*_{Xj}= \sigma^*_{Xj}, \hat{\sigma}_{Yj}= \sigma_{Yj}$.
\begin{lemma} \label{lemma: dIVW denominator}
	Assume models (\ref{eq: exposure-IV})-(\ref{eq: outcome-exposure}),  Assumptions \ref{assump: 1}-\ref{assump: 2},   $ \kappa \sqrt{p}\rightarrow \infty $ and $ p/n_X\rightarrow 0 $, as $ p, n_X \rightarrow \infty $,
	\begin{align}
		\frac{\sum_{j=1}^p\left( \hat{\gamma}_j^2-\hat{\sigma}_{Xj}^2 \right) \hat{\sigma}_{Yj}^{-2}}{\sum_{j=1}^{p} w_j}\xrightarrow{P} 1\nonumber 
	\end{align}
\end{lemma}

\begin{proof}
	\begin{align}
		&\sum_{j=1}^p\left( \hat{\gamma}_j^2-\hat{\sigma}_{Xj}^2 \right) \hat{\sigma}_{Yj}^{-2} = \sum_{j=1}^p \gamma_j^2\sigma_{Yj}^{-2} + \underbrace{\sum_{j=1}^{p} (\hat{\gamma}_j^2-\sigma_{Xj}^2-\gamma_j^2)\sigma_{Yj}^{-2} }_{A_1}\nonumber\\
		&\qquad +\underbrace{ \sum_{j=1}^{p} (\hat{\gamma}_j^2- \sigma_{Xj}^2-\gamma_j^2 ) (\hat{\sigma}_{Yj}^{-2}- \sigma_{Yj}^{-2}) }_{A_2} + \underbrace{\sum_{j=1}^{p} \gamma_j^2(\hat{\sigma}_{Yj}^{-2}- \sigma_{Yj}^{-2}) }_{A_3}  \nonumber\\
		& \qquad    -\underbrace{ \sum_{j=1}^{p} (\hat{\sigma}_{Xj}^2-\sigma_{Xj}^2) \sigma_{Yj}^{-2} }_{A_4}-\underbrace{ \sum_{j=1}^{p} (\hat{\sigma}_{Xj}^2-\sigma_{Xj}^2) (\hat{\sigma}_{Yj}^{-2}- \sigma_{Yj}^{-2}) }_{A_5}\nonumber
	\end{align}
	In Section \ref{subsubsec: divw, known SDs, lambda=0}, we have shown that $ A_1=o_P( \kappa p) $. The rest of the terms can also be shown to be $ o_P(\kappa p) $. The details are similar to the proof of Lemma \ref{lemma: dIVW}, thus are omitted.
\end{proof}

Combining Lemmas \ref{lemma: dIVW} and \ref{lemma: dIVW denominator}, we have the consistency and asymptotic normality of $ \hat{\beta}_{\rmdivw}. $

\subsubsection{Without assuming knowing SDs and $ \lambda\geq 0 $}

In this section, we show the results for general  $ \lambda\geq 0 $ and without assuming knowing $ \sigma_{Xj}, \sigma_{Xj}^*, \sigma_{Yj} $.   
\begin{lemma} \label{lemma: dIVW, threshold}
	Assume models (\ref{eq: exposure-IV})-(\ref{eq: outcome-exposure}),  Assumptions \ref{assump: 1}-\ref{assump: 2},   and that $ \kappa_\lambda \sqrt{p_\lambda}/ \max (1, \lambda^2)\rightarrow \infty $,  $
	\max_j (\gamma_j^2\sigma_{Xj}^{-2} q_{\lambda, j})/ (\kappa_\lambda p_\lambda+p_\lambda)\rightarrow 0 $, $ p/n_X \rightarrow 0$, when $ p, n_X \rightarrow \infty $,
	\begin{align}
		\frac{\sum_{j=1}^p\left(\hat{\Gamma}_j \hat{\gamma}_j-\beta_0\hat{\gamma}_j^2+\beta_0\hat{\sigma}_{Xj}^2\right)\hat{\sigma}_{Yj}^{-2}I(|\hat{\gamma}_j^*|>\lambda\hat{\sigma}_{Xj}^* )}{ {\sigma}_{\lambda, dIVW} }\xrightarrow{D} N(0,1)\nonumber 
	\end{align}
	where $ 	{\sigma}^2_{\lambda, \rmdivw}=\sum_{j=1}^p[ (w_j+v_j)+\beta_0^2(w_j+2v_j)v_j]q_{\lambda, j} $
\end{lemma}

\begin{proof}
	Consider the numerator 
	\begin{align}
		&\sum_{j=1}^p\left(\hat{\Gamma}_j \hat{\gamma}_j-\beta_0\hat{\gamma}_j^2+\beta_0\hat{\sigma}_{Xj}^2\right)\hat{\sigma}_{Yj}^{-2}  I(|\hat{\gamma}_j^*|>\lambda\hat{\sigma}_{Xj}^* )\nonumber\\
		&= \sum_{j=1}^p\left(\hat{\Gamma}_j \hat{\gamma}_j-\beta_0\hat{\gamma}_j^2+\beta_0\hat{\sigma}_{Xj}^2\right)\hat{\sigma}_{Yj}^{-2}  I(|\hat{\gamma}_j^*|>\lambda{\sigma}_{Xj}^* ) \nonumber\\
		&\qquad + \sum_{j=1}^p\left(\hat{\Gamma}_j \hat{\gamma}_j-\beta_0\hat{\gamma}_j^2+\beta_0\hat{\sigma}_{Xj}^2\right)\hat{\sigma}_{Yj}^{-2}  \{I(|\hat{\gamma}_j^*|>\lambda\hat{\sigma}_{Xj}^* )-I(|\hat{\gamma}_j^*|>\lambda{\sigma}_{Xj}^* )\} \nonumber
	\end{align}
	We will deal with the first term, and then show the second term is negligible. Notice that 
	\begin{align}
		& \sum_{j=1}^p\left(\hat{\Gamma}_j \hat{\gamma}_j-\beta_0\hat{\gamma}_j^2+\beta_0\hat{\sigma}_{Xj}^2\right)\hat{\sigma}_{Yj}^{-2}  I(|\hat{\gamma}_j^*|>\lambda{\sigma}_{Xj}^* )\label{eq: threshold}\\
		&= \sum_{j=1}^p\left(\hat{\Gamma}_j \hat{\gamma}_j-\beta_0\hat{\gamma}_j^2+ \beta_0\sigma_{Xj}^2\right){\sigma}_{Yj}^{-2}I(|\hat{\gamma}_j^*|>\lambda{\sigma}_{Xj}^* ) + \sum_{j=1}^p \beta_0(\hat{\sigma}_{Xj}^2- \sigma_{Xj}^2)\sigma_{Yj}^{-2} q_{\lambda, j}\nonumber\\
		&+ 	\sum_{j=1}^p\left(\hat{\Gamma}_j \hat{\gamma}_j-\beta_0\hat{\gamma}_j^2+\beta_0\sigma_{Xj}^2\right)(\hat{\sigma}_{Yj}^{-2}- {\sigma}_{Yj}^{-2}) I(|\hat{\gamma}_j^*|>\lambda{\sigma}_{Xj}^* )\nonumber\\
		&	+ \sum_{j=1}^{p} \beta_0(\hat{\sigma}_{Xj}^2- \sigma_{Xj}^2)(\hat{\sigma}_{Yj}^{-2}- \sigma_{Yj}^{-2})I(|\hat{\gamma}_j^*|>\lambda{\sigma}_{Xj}^* ) \nonumber\\
		&+ \sum_{j=1}^p  \beta_0(\hat{\sigma}_{Xj}^2- \sigma_{Xj}^2) \sigma_{Yj}^{-2}  \{    I(|\hat{\gamma}_j^*|>\lambda{\sigma}_{Xj}^* )- q_{\lambda, j} \}\nonumber\\
		&=\underbrace{ \sum_{j=1}^p\left(\hat{\Gamma}_j \hat{\gamma}_j-\beta_0\hat{\gamma}_j^2+ \beta_0 \sigma_{Xj}^2\right){\sigma}_{Yj}^{-2} I(|\hat{\gamma}_j^*|>\lambda{\sigma}_{Xj}^* )}_{B_1} +\Psi_{\lambda,X} \underbrace{ \left(\frac{\X^T (I-H_e)\X}{(n_X-1) \var(X)}-1\right) }_{B_2}\nonumber\\ 
		&+	\underbrace{ \sum_{j=1}^p\left(\hat{\Gamma}_j \hat{\gamma}_j-\beta_0\hat{\gamma}_j^2+ \beta_0 \sigma_{Xj}^2\right)(\hat{\sigma}_{Yj}^{-2}- {\sigma}_{Yj}^{-2})I(|\hat{\gamma}_j^*|>\lambda{\sigma}_{Xj}^* )}_{B_3}\nonumber\\
		&+ \underbrace{ \sum_{j=1}^{p} \beta_0(\hat{\sigma}_{Xj}^2- \sigma_{Xj}^2)(\hat{\sigma}_{Yj}^{-2}- \sigma_{Yj}^{-2})I(|\hat{\gamma}_j^*|>\lambda{\sigma}_{Xj}^* )}_{B_4}  \nonumber\\
		& +  \underbrace{  O_P\left( \frac{p_\lambda}{n_X}\right)   + O_P \left( \sqrt{\frac{p_\lambda}{n_X}}\right) +O_P\left(\frac{\sqrt{\kappa_\lambda p_\lambda+p_\lambda}}{n_X}\right)+\Theta\left( \frac{p_\lambda}{n_X}+\frac{\kappa_\lambda p_\lambda}{n_X^2}\right)}_{o_P(\sqrt{\kappa_\lambda p_\lambda+p_\lambda}) +  o(\sqrt{ \kappa_\lambda p_\lambda+p_\lambda} )} \nonumber\\ 
		&= B_1+\Psi_{\lambda, X} B_2+B_3+B_4+ o_P(\sqrt{\kappa_\lambda p_\lambda+p_\lambda})+ o(\sqrt{ \kappa_\lambda p_\lambda+p_\lambda} )\nonumber
	\end{align}
	where $ 	\Psi_{\lambda, X}=  \sum_{j=1}^p \beta_0v_j\psi_{Xj}^2q_{\lambda, j},  \psi_{Xj}^2= \frac{n_X \var(X)}{(n_X-2) (\var(X)- \gamma_j^2\var(Z_j))}$.

	In the sequel, we deal with $ B_1 $ and $ \Psi_{\lambda, X} B_2 $ separately, and show that $ B_3+B_4=o_P(\sqrt{\kappa_\lambda p_\lambda+p_\lambda})$.
	
	The first term $ B_1 $ is studied in Lemma \ref{lemma: three sample 1}, and
	\[
	\frac{B_1}{	\{ \sum_{j=1}^p[ (w_j+v_j)+\beta_0^2(w_j+2v_j)v_j]q_{\lambda, j}  \}^{1/2}} \xrightarrow{D} N(0,1).
	\]
	The second term $ B_2 $ has been studied in Lemma \ref{lemma: dIVW}. Therefore, similar to the proof of Lemma \ref{lemma: dIVW}, as $ p\rightarrow \infty $ and $ n_X\rightarrow \infty $,
	\begin{align}
		\frac{B_1+\Psi_{\lambda, X} B_2}{ \sigma_{\lambda, dIVW}}\xrightarrow{D}N(0,1)
	\end{align}
	One can similarly show that $ B_3, B_4 $ are negligible, following the proofs in Lemma \ref{lemma: dIVW}.
	
	Finally, it remains to show that 
	\begin{align}
		&	\sum_{j=1}^p\left(\hat{\Gamma}_j \hat{\gamma}_j-\beta_0\hat{\gamma}_j^2+\beta_0\hat{\sigma}_{Xj}^2\right)\hat{\sigma}_{Yj}^{-2}  \{I(|\hat{\gamma}_j^*|>\lambda\hat{\sigma}_{Xj}^* )-I(|\hat{\gamma}_j^*|>\lambda{\sigma}_{Xj}^* )\}  \nonumber\\
		&	= o_P (\sqrt{ \kappa_\lambda p_\lambda + p_\lambda}). \nonumber
	\end{align}
	One can  decompose the above equation similarly as in  (\ref{eq: threshold}), where the leading term is 
	\begin{align}
		\sum_{j=1}^p\left(\hat{\Gamma}_j \hat{\gamma}_j-\beta_0\hat{\gamma}_j^2+\beta_0{\sigma}_{Xj}^2\right){\sigma}_{Yj}^{-2}  \{I(|\hat{\gamma}_j^*|>\lambda\hat{\sigma}_{Xj}^* )-I(|\hat{\gamma}_j^*|>\lambda{\sigma}_{Xj}^* )\}  \label{eq: threshold leading }
	\end{align}
	To prove (\ref{eq: threshold leading }) is of order $ o_P(\sqrt{\kappa_\lambda p_\lambda+p_\lambda}) $, it is easy to show that the above equation has expectation zero. To derive its variance, we first notice that the pairwise covariance equals zero, which can be seen using the relationship that for two generic random variables $ X, Y $ and sigma field $ {\cal F} $, 
	\begin{align}
		\cov(X, Y)= E(\cov(X, Y|{\cal F}) ) + \cov(E(X|{\cal F}), E(Y|{\cal F}))
	\end{align}
	For $ j\neq k $, let $ X= \left(\hat{\Gamma}_j \hat{\gamma}_j-\beta_0\hat{\gamma}_j^2+\beta_0{\sigma}_{Xj}^2\right){\sigma}_{Yj}^{-2}  \{I(|\hat{\gamma}_j^*|>\lambda\hat{\sigma}_{Xj}^* )-I(|\hat{\gamma}_j^*|>\lambda{\sigma}_{Xj}^* )\}  $ ,$ Y=\left( \hat{\Gamma}_k \hat{\gamma}_k-\beta_0\hat{\gamma}_k^2+\beta_0{\sigma}_{Xk}^2\right){\sigma}_{Yk}^{-2}  \{I(|\hat{\gamma}_k^*|>\lambda\hat{\sigma}_{Xk}^* )-I(|\hat{\gamma}_k^*|>\lambda{\sigma}_{Xk}^* )\}   $, $ {\cal F} = \{( \hat{\gamma}_j^*, \hat{\sigma}_{Xj}^*), j=1,\dots, p\}$. Because the selection dataset is independent with the exposure and the outcome datasets, we have $ E(X|{\cal F}) =E(Y|{\cal F}) =0$, and thus $  \cov(E(X|{\cal F}), E(Y|{\cal F})) =0$. Also, $ \cov(X, Y|{\cal F}) =0$, which implies that $ \cov(X, Y)=0 $.
	
	Then, some algebras reveal that 
	\begin{align}
		&E \{I(|\hat{\gamma}_j^*|>\lambda\hat{\sigma}_{Xj}^* )-I(|\hat{\gamma}_j^*|>\lambda{\sigma}_{Xj}^* )\} ^2\nonumber\\
		&= P \left(  \lambda\sigma_{Xj}^* \geq |\hat{\gamma}_j^*|>\lambda\hat{\sigma}_{Xj}^*\right) + P \left(  \lambda\sigma_{Xj}^* < |\hat{\gamma}_j^*|\leq \lambda\hat{\sigma}_{Xj}^*\right)  \nonumber\\
		&\leq P\left( \lambda (\hat{\sigma}_{Xj}^*/ \sigma_{Xj}^*-1)<0\right)+P\left( \lambda (\hat{\sigma}_{Xj}^*/ \sigma_{Xj}^*-1)>0\right) \nonumber\\
		&\leq P\left( \lambda (\hat{\sigma}_{Xj}^*/ \sigma_{Xj}^*-1)+\epsilon\leq 0\right)+P\left( \lambda (\hat{\sigma}_{Xj}^*/ \sigma_{Xj}^*-1)-\epsilon\geq 0\right)  \qquad \exists \epsilon>0 \nonumber\\
		&=  P\left( \lambda |\hat{\sigma}_{Xj}^*/ \sigma_{Xj}^*-1|\geq  \epsilon\right)   \nonumber\\
		&\leq \frac{\lambda^2}{\epsilon^2} E\left[ \left( \frac{\hat{\sigma}_{Xj}^*}{\sigma_{Xj}^*}-1\right)^2\right] = \Theta\left( \frac{\lambda^2}{n_X^*}\right)\nonumber
	\end{align}
	where the last inequality is from Markov Inequality. 
	
	Combining the above derivations, we have that the variance of (\ref{eq: threshold leading }) is of order $ \Theta(n_X^{-1}\lambda^2 (\kappa p+p)) $. From the fact that $ \kappa p/n_X = \Theta(1) $ and the assumptions that  $ p/n_X=o(1) , \lambda^2/(\kappa_\lambda\sqrt{p_\lambda})=o(1)$, we conclude that (\ref{eq: threshold leading }) is of order $ o_P(\sqrt{\kappa_\lambda p_\lambda+p_\lambda}) $.
	
\end{proof}

\subsection{ Proof of the dIVW estimator under balanced horizontal pleiotropy}This proof is similar to the proof of Theorem \ref{theo: dIVW}. In particular, notice that Assumption $ 2' $ implies $\hat{\Gamma}_j\sim N(\beta_0\gamma_j,\sigma_{Yj}^2+\tau_0^2)$, that is $\hat{\Gamma}_j$ follows a normal distribution with mean $\beta_0\gamma_j$ and variance $\sigma_{Yj}^2+\tau_0^2$, with  $\sigma_{Xj}^2/(\sigma_{Yj}^2+\tau_0^2)=\Theta (1)$. Therefore, under balanced horizontal pleiotropy, Assumptions \ref{assump: 1}-\ref{assump: 2} still hold if we replace $\sigma_{Yj}^2$ with $\sigma_{Yj}^2+\tau_0^2$. Here, Lemma \ref{lemma: three sample 2} also holds because the distribution of $\hat{\gamma}_j$ is unchanged. However, Lemma \ref{lemma: three sample 1} needs to be modified to account for the pleiotropy effect

Specifically, under Assumptions \ref{assump: 1}-\ref{assump: 2} (with $\sigma_{Yj}^2$ replaced with $\sigma_{Yj}^2+\tau_0^2$), suppose that $\max_j(\gamma_j^2\sigma_{Xj}^{-2}q_{\lambda,j})/(\kappa_\lambda p_\lambda+p_\lambda)\rightarrow 0$, then as
$p\rightarrow \infty$, 
\begin{eqnarray}
	\frac{\sum_{j=1}^p\left(\hat{\Gamma}_j \hat{\gamma}_j-\beta_0 \hat{\gamma}_j^2+\beta_0\sigma_{Xj}^2\right)\sigma_{Yj}^{-2}\indpit}{\left[\sum_{j=1}^p\left\{  (1+\tau_0^2\sigma_{Yj}^{-2})(w_j+v_j)+\beta_0^2v_j(w_j+2v_j)\right\}q_{\lambda,j}\right]^{1/2}} \xrightarrow{D} N(0,1).\nonumber
\end{eqnarray}
Hence, following Slutsky's theorem, we have
\[
W_{\lambda,\rmdivw}^{-1/2}(\hat{\beta}_{\lambda,\rmdivw}-\beta_0)\xrightarrow{D} N(0,1).
\]
where 
\[
W_{\lambda,\rmdivw} = \frac{\sum_{j=1}^p\left\{  (1+\tau_0^2\sigma_{Yj}^{-2})(w_j+v_j)+\beta_0^2v_j(w_j+2v_j)\right\}q_{\lambda,j} }{ [ \sum_{j=1}^{p} w_jq_{\lambda, j}]^2}
\]
Next, we  prove the plug-in variance estimator in (\ref{evar5}), hereafter denoted by $\hat{W}_{\lambda,\rmdivw}$, is consistent. Similar to the proof of Theorem \ref{theo: p}(a), it is straightforward to show that 
\[
\frac{\sum_{j \in S_\lambda}\left\{ \hat{\beta}_{\lambda,\rmdivw}^2v_j(\hat{w}_j+v_j)\right\}}{ \left[\sum_{j \in S_\lambda}(\hat{w}_j-v_j)\right]^2}= \frac{\sum_{j=1}^p\beta_0^2v_j (w_j+2v_j)q_{\lambda,j}}{(\sum_{j=1}^p w_jq_{\lambda,j})^2}+o_P(1).
\]
It remains to show that 
\begin{eqnarray}
	\frac{\sum_{j \in S_\lambda}\left\{  \hat{w}_j(1+\hat{\tau}^2\sigma_{Yj}^{-2})\right\}}{\sum_{j=1}^p (1+\tau_0^2\sigma_{Yj}^{-2})(w_j+v_j)q_{\lambda,j}}\xrightarrow{P} 1. \label{appeq: tau}
\end{eqnarray}
We define $\tilde{\tau}^2$ as an analogue of  $\hat{\tau}^2$ with $\hat{\beta}_{0, \rmdivw}$ replaced by $\beta_0$, i.e.,
\[
\tilde{\tau}^2=\frac{\sum_{j =1}^p\left[(\hat{\Gamma}_j-\beta_0\hat{\gamma}_j)^2-\sigma_{Yj}^2-\beta_0^{2}\sigma_{Xj}^2\right]\sigma_{Yj}^{-2}}{\sum_{j=1}^p\sigma_{Yj}^{-2}}.
\]
Some algebra reveals that we have $E(\tilde{\tau}^2)=\tau_0^2$ and 
\[
\var(\tilde{\tau}^2)=\frac{2\sum_{j=1}^p (1+\beta_0^2v_j+\tau_0^2\sigma_{Yj}^{-2})^2}{(\sum_{j=1}^p \sigma_{Yj}^{-2})^2}=\frac{\Theta(p)}{(\sum_{j=1}^p \sigma_{Yj}^{-2})^2}.
\]
In this proof, for a random variable $X$ with finite second moments, we can write
\[
X= E(X)+O_P((\var(X))^{1/2}).
\]
from Chebyshev inequality. Using this, we have
\[
(\sum_{j=1}^p \sigma_{Yj}^{-2})(\tilde{\tau}^2-\tau_0^2)=O_P(\sqrt{p}). 
\]
Under a Taylor series expansion of the above quantity, the first and second derivatives are non-zero, and we have
\begin{align}
	&(\sum_{j=1}^p \sigma_{Yj}^{-2})(\hat{\tau}^2-\tilde{\tau}^2) \nonumber\\
	=& (\hat{\beta}_{\rmdivw}-\beta_0) \left[\sum_{j=1}^p 2(-\hat{\Gamma}_j\hat{\gamma}_j+\beta_0\hat{\gamma}_j^2-\beta_0\sigma_{Xj}^2)\sigma_{Yj}^{-2}\right]\nonumber \\
	&\qquad+(\hat{\beta}_{\rmdivw}-\beta_0)^2\left[\sum_{j=1}^{p} (\hat{\gamma}_j^2-\sigma_{Xj}^2)\sigma_{Yj}^{-2}\right]\nonumber\\
	=&\frac{ -\left[\sum_{j=1}^p (\hat{\Gamma}_j\hat{\gamma}_j-\beta_0\hat{\gamma}_j^2+\beta_0\sigma_{Xj}^2)\sigma_{Yj}^{-2}\right]^2}{\sum_{j=1}^{p} (\hat{\gamma}_j^2-\sigma_{Xj}^2)\sigma_{Yj}^{-2}}  \label{appeq: pleiotropy}
\end{align}
where the second equality is from plugging in the expression for $\hat{\beta}_{\rmdivw}-\beta_0$. From the proof of Theorem \ref{theo: dIVW}, where we have already derived the asymptotic distribution of the numerator, the normalized version in (\ref{appeq: pleiotropy}) asymptotically follows a chi-square distribution and 
\begin{align*}
	&\frac{\sum_{j=1}^{p}w_j}{\sum_{j  =1}^p (w_j+v_j)+\beta_0^2(w_j+2v_j)v_j}  (\sum_{j=1}^p \sigma_{Yj}^{-2})(\hat{\tau}^2-\tilde{\tau}^2) =O_P(1) \\
	&(\sum_{j=1}^p \sigma_{Yj}^{-2})(\hat{\tau}^2-\tilde{\tau}^2) =O_P((\kappa p+p)/(\kappa p))
\end{align*}
From $\kappa \sqrt{p}\rightarrow \infty$ and $p\rightarrow \infty$, it is easy to see that $(\kappa p+p) / (\kappa p \sqrt{ p}) \rightarrow 0$ and $(\sum_{j=1}^p \sigma_{Yj}^{-2})(\hat{\tau}^2-\tilde{\tau}^2)=o_P(\sqrt{p})$. Therefore, 
\[
(\sum_{j=1}^p \sigma_{Yj}^{-2}) (\hat{\tau}^2-\tau_0^2)=(\sum_{j=1}^p \sigma_{Yj}^{-2})(\hat{\tau}^2-\tilde{\tau}^2)-(\sum_{j=1}^p \sigma_{Yj}^{-2})(\tilde{\tau}^2-\tau_0^2)=O_P (\sqrt{p}).
\]

Notice that the denominator of (\ref{appeq: tau}) is of order $\Theta(\kappa_\lambda p_\lambda+p_\lambda)$. As such, we will prove (\ref{appeq: tau}) by showing the difference of the numerator and the denominator is $o_P(\kappa_\lambda p_\lambda+p_\lambda)$. The difference between the numerator and the denominator is 
\[
\underbrace{\sum_{j \in S_\lambda} \hat{w}_j-\sum_{j=1}^p (w_j+v_j)q_{\lambda,j}}_{o_P(\kappa_\lambda p_\lambda+p_\lambda)}+\underbrace{\hat{\tau}^2\sum_{j \in S_\lambda} \sigma_{Yj}^{-2}\hat{w}_j-\tau_0^2 \sum_{j=1}^p \sigma_{Yj}^{-2} (w_j+v_j)q_{\lambda,j}}_{B_1},
\]
where the first term is $o_P(\kappa_\lambda p_\lambda+p_\lambda)$ from Lemma \ref{lemma: three sample 2}. We decompose $B_1$ as
\[
B_1=\underbrace{(\hat{\tau}^2-\tau_0^2)\sum_{j \in S_\lambda} \sigma_{Yj}^{-2}\hat{w}_j}_{B_{11}}+\underbrace{\tau_0^2 \left\{\sum_{j \in S_\lambda} \sigma_{Yj}^{-2}\hat{w}_j-\sum_{j=1}^p \sigma_{Yj}^{-2} (w_j+v_j)q_{\lambda,j}\right\}}_{B_{12}}.
\]
Some algebra reveals that
\begin{eqnarray}
	&E\left(\sum_{j \in S_\lambda} \sigma_{Yj}^{-2}\hat{w}_j\right)=\sum_{j=1}^p \sigma_{Yj}^{-2}(w_j+v_j)q_{\lambda,j} \nonumber\\
	& \left[\var\left(\sum_{j \in S_\lambda} \sigma_{Yj}^{-2}\hat{w}_j\right)\right]^{1/2}= o\left(\frac{\sum_{j  =1}^p \sigma_{Yj}^{-2}}{p} (\kappa_\lambda p_\lambda+p_\lambda)\right) \nonumber
\end{eqnarray}
where the second line uses the assumption   $\max_j \sigma_{Yj}^{-2}\leq cp^{-1}\sum_{k=1}^p \sigma_{Yk}^{-2}$ with a positive constant $c$, together with a similar argument in the proof of Lemma \ref{lemma: ivw three sample 2}.
%\begin{align}
%&\sum \sigma_{Yj}^{-4}\mu_j^4q^*_{\lambda,j}(1-q_{\lambda,j})=O(p^{-2}(\sum\sigma_{Yj}^{-2})^2 (\lambda^4+1) p_\lambda) \nonumber \\
%&\Theta(\sum \sigma_{Yj}^{-4} \mu_j^2q_{\lambda,j}) =O(p^{-2}(\sum\sigma_{Yj}^{-2})^2 \kappa_\lambda p_\lambda) \nonumber  \\
%&\Theta(\sum \sigma_{Yj}^{-4} q_{\lambda,j})=O(p^{-2}(\sum\sigma_{Yj}^{-2})^2 p_\lambda)\nonumber.
%\end{align}
Combining the results, the order of $B_{11}$ is
\begin{align}
	B_{11}&=(\hat{\tau}^2-\tau_0^2) \left[\sum_{j=1}^p \sigma_{Yj}^{-2} (w_j+v_j)q_{\lambda,j}+o_P\left(\frac{\sum_{j=1}^p \sigma_{Yj}^{-2}}{p} (\kappa_\lambda p_\lambda+p_\lambda)\right)\right] \nonumber\\
	&= \frac{O_P (\sqrt{p})}{\sum_{j=1}^p\sigma_{Yj}^{-2}} \left[\sum_{j=1}^p \sigma_{Yj}^{-2} (w_j+v_j)q_{\lambda,j}\right] +o_P(\kappa_\lambda p_\lambda+p_\lambda) \nonumber \\
	&\leq  \frac{O_P (\sqrt{p})}{\sum_{j=1}^p\sigma_{Yj}^{-2}}\left[\sum_{j=1}^p (w_j+v_j)q_{\lambda,j} \right] \frac{c}{p}\sum_{j=1}^p\sigma_{Yj}^{-2}+o_P(\kappa_\lambda p_\lambda+p_\lambda)=o_P(\kappa_\lambda p_\lambda+p_\lambda) \nonumber
\end{align}
Similarly $B_{12}=o_P(\kappa_\lambda p_\lambda+p_\lambda)$ from $\tau_0\leq c_{+} \sigma_{Yj}$ for all $j$ in the Assumption $ 2' $, completing the proof.

(b) Proof for consistency is similar to the proof in Theorem \ref{theo: dIVW}(b) and is omitted.

\end{document}